\documentclass[10pt,journal,  compsoc]{IEEEtran}

 \usepackage{hyperref}
\usepackage{fancybox}   
\usepackage{psfig}      
\usepackage{graphicx}
\usepackage{amsmath}
\usepackage{color}
\usepackage{epsfig}
\usepackage{amssymb}
\usepackage{verbatim}
\usepackage{enumitem}
\usepackage{setspace}
\usepackage{subfigure}

 \usepackage[T1]{fontenc}
 \usepackage[english]{babel}

\makeatletter
\newcommand{\Rmnum}[1]{\expandafter\@slowromancap\romannumeral#1@}
\makeatother

\newtheorem{thm}{Theorem}
\newtheorem{cor}{Corollary}
\newtheorem{lem}{Lemma}
\newtheorem{prop}{Proposition}

\newtheorem{claim}{Claim}

\newtheorem{condition}{Condition}

\newcommand{\qed}{\hfill \ensuremath{\Box}}

\begin{document}

\title{Deploy-As-You-Go Wireless Relay Placement: An Optimal Sequential Decision Approach
using the Multi-Relay Channel Model\thanks{This work was supported by 
the Department of Science and Technology (DST), India, through the J.C. Bose Fellowship, by an Indo-Brazil cooperative 
project on ``WIreless Networks and techniques with applications to SOcial Needs (WINSON)," and by a project funded
by the Department of Electronics and Information Technology, India, and NSF, USA, titled ``Wireless Sensor
Networks for Protecting Wildlife and Humans in Forests.''}\thanks{This paper is an  
extension of  \cite{chattopadhyay-etal12optimal-capacity-relay-placement-line}, and 
is also available in \cite{chattopadhyay-etal12optimal-capacity-relay-placement-line-arxiv-17April2013}. 
}\thanks{Arpan Chattopadhyay and Anurag Kumar are with the Electrical Communication Engineering (ECE) Department, 
Indian Institute of Science (IISc), Bangalore-560012, India (e-mail: arpanc.ju@gmail.com, anurag@ece.iisc.ernet.in). 
Abhishek Sinha is with the Laboratory for Information and Decision Systems (LIDS), 
Massachusetts Institute of Technology, Cambridge, MA 02139 (e-mail: sinhaa@mit.edu). 
Marceau Coupechoux is with Telecom ParisTech and CNRS LTCI, Dept. Informatique et R\'eseaux, 
23, avenue d'Italie, 75013 Paris, France (e-mail: marceau.coupechoux@telecom-paristech.fr). This work was done during the 
period when he was a Visiting Scientist in the ECE Deparment, IISc.}
}

\author{
Arpan~Chattopadhyay, Abhishek~Sinha, Marceau~Coupechoux, and Anurag~Kumar~\IEEEmembership{Fellow,~IEEE}\\
}


\IEEEcompsoctitleabstractindextext{
\begin{abstract}
We use information theoretic achievable rate formulas for the multi-relay channel to study the 
problem of as-you-go deployment of relay nodes. The achievable rate formulas are for full-duplex radios at 
the relays and for decode-and-forward relaying. Deployment is done along 
the straight line joining a source node and a sink node at an unknown distance from the source. 
The problem is for a deployment agent to walk from the source to the sink, deploying relays as he walks, 
given the knowledge of the wireless path-loss model, and given that the distance 
to the sink node is exponentially distributed with known mean.  As a precursor to the formulation of the deploy-as-you-go problem, 
we apply the multi-relay channel achievable rate formula to obtain the optimal power allocation to relays placed along a line, 
at fixed locations. This permits us to obtain the optimal placement of a given number of nodes when the 
distance between the source and sink is given. Numerical work for the  fixed source-sink distance case  
suggests that, at low attenuation, the relays are mostly clustered close to the source in order to be able 
to cooperate among themselves, whereas at high attenuation they are uniformly placed and work as repeaters. 
We also prove that the effect of path-loss can be entirely mitigated if a large enough number of relays are placed uniformly 
between the source and the sink. The structure of the optimal power allocation for a given placement of the nodes, 
then motivates us to formulate 
the problem of as-you-go placement of relays along a line of exponentially distributed length, 
and with the exponential path-loss model, 
so as to minimize a cost function that is additive over hops. The hop cost trades off a capacity limiting term, 
motivated from the optimal power allocation solution, against the cost of adding a relay node. We formulate the problem 
as a total cost Markov decision process, establish results for the value function, and provide insights into 
the placement policy and the performance of the deployed network via numerical exploration. 
\end{abstract}

\vspace{-2mm}

\begin{keywords}
  Multi-relay channel, optimal relay placement, impromptu wireless networks, as-you-go relay placement.
\end{keywords}
}

\maketitle

\vspace{-20mm}
\section{Introduction}
\label{Introduction}
\vspace{-2mm}

Wireless interconnection of devices such as smart 
phones, or wireless sensors, to the wireline 
communication infrastructure is an important requirement. 
These are battery operated, resource constrained devices. Hence, 
due to the physical placement of these devices, or due to channel 
conditions, a direct one-hop link to the infrastructure ``base-station'' 
might not be feasible. In such situations, other nodes could 
serve as \emph{relays} in order to realize a multi-hop path 
between the source device and the infrastructure. In the wireless sensor network context, the relays could be other 
wireless sensors or battery operated radio routers deployed 
specifically as relays. The relays are also 
resource constrained and a cost might be involved in 
placing them. Hence, there arises the problem of \emph{optimal relay placement}. 
Such a  problem involves the joint optimization 
of node placement and   operation of the resulting network, 
where by ``operation'' we mean activities such as transmission 
scheduling, power allocation, and channel coding.

\begin{figure}[t!]
\centering
\includegraphics[height=1.2cm, width=8cm]{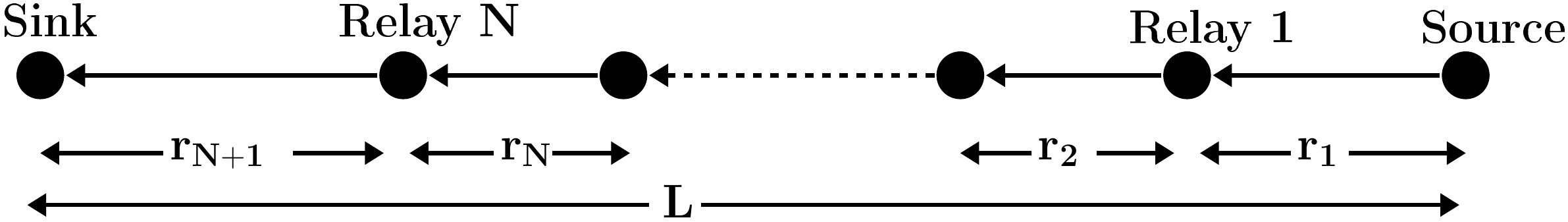}
\caption{A source and a sink connected by a multi-hop path comprising $N$ relay nodes along a line.}
\label{fig:general_line_network}
\vspace{-8mm}
\end{figure}

Our work in this paper is motivated by recent interest in problems of {\em impromptu (as-you-go)}  
deployment of wireless relay networks in various 
situations; for example, ``first responders'' in emergency situations, or 
quick deployment (and redeployment) of sensor networks in large terrains, such as forests 
(see \cite{souryal-etal07real-time-deployment-range-extension}, \cite{howard-etal02incremental-self-deployment-algorithm-mobile-sensor-network}, 
\cite{bao-lee07rapid-deployment-ad-hoc-backbone}, \cite{sinha-etal12optimal-sequential-relay-placement-random-lattice-path}, 
\cite{chattopadhyay-etal13measurement-based-impromptu-placement_wiopt}). 
In this paper, we are concerned with the situation in which a deployment agent walks {\em from 
the source node to the sink node}, along the line joining these two nodes, and places wireless relays 
(in an ``as-you-go'' manner) so as to create a source-to-sink multi-relay channel network with high data rate;   
see Figure \ref{fig:general_line_network}. We first consider the scenario where the length $L$  of the line 
in Figure \ref{fig:general_line_network} is known; the results of this case 
are used to formulate the as-you-go deployment in the case where $L$ is a priori unknown, 
but has exponential distribution with known mean $\overline{L}$.

In order to capture the fundamental trade-offs involved in such problems, we consider an 
information theoretic model. For a placement of the relay nodes and allocation of transmission powers to these relays, 
we model the ``quality'' of communication between the source and the 
sink by the information theoretic achievable rate of the multi-relay 
channel (see \cite{cover-gamal79capacity-relay-channel}, \cite{reznik-etal04degraded-gaussian-multirelay-channel} and 
\cite{xie-kumar04network-information-theory-scaling-law} for the single and multi-relay channel models). The relays are equipped 
with full-duplex radios\footnote{Full-duplex radios are becoming practical; see \cite{khandani13two-way-full-duplex-wireless}, 
\cite{khandani10spatial-multiplexing-two-way-channel}, \cite{choi-etal10single-channel-full-duplex}, 
\cite{jain-etal11real-time-full-duplex}.}, and carry out 
decode-and-forward relaying. We consider scalar, memoryless, time-invariant, additive white Gaussian noise (AWGN) channels. 
We assume synchronous operation across all transmitters and receivers, and consider the exponential path-loss model 
for radio wave propagation.

\vspace{-4mm}
\subsection{Related Work}
\label{subsection:related_work}
\vspace{-1mm}

A formulation of the  relay placement problem requires a 
model of the wireless network at the physical (PHY) and medium access control (MAC) layers. Most researchers have adopted 
the link scheduling and interference model, i.e., 
a scheduling algorithm determines radio resource allocation (channel and power) and interference is treated 
as noise (see~\cite{georgiadis-etal06resource-allocation-cross-layer-control}); 
treating interference as noise leads to the model that simultaneous 
transmissions ``collide'' at receiving nodes, and transmission scheduling aims to avoid collisions. 

However, node placement for throughput
maximization  with this model is intractable because the optimal throughput is 
obtained by first solving for the optimum schedule 
assuming fixed node locations, followed by an optimization over those locations. 
Hence, with such a model, there appears to be little work on the problem of jointly optimizing 
the relay node placement and the 
transmission schedule.  Reference~\cite{firouzabadi-martins08optimal-node-placement} is one  
such work where the authors considered placing a set of nodes in an 
existing network such that a certain network utility is optimized subject to a set of linear constraints on link rates, 
under the link scheduling and interference model. 
They posed the problem as one of geometric 
programming assuming exponential path-loss, and 
proposed a distributed solution. The authors of \cite{zheng-etal12robust-relay-placement_arxiv} 
consider relay placement for utility maximization, assuming there are several source nodes, sink nodes 
and a few candidate locations for placing relays; they ignore interference because of highly 
directional antennas used in $60$~GHz mmWave  networks, which may not always be valid. Relay placement 
for capacity enhancement has been studied in \cite{lu-etal11relay-placement-80216}, but there  
interference is mitigated by scheduling transmissions over multiple channels.

On the other hand, an information theoretic model for a wireless network often provides a closed-form expression for the 
channel capacity, or at least an achievable rate region. These results are asymptotic, and make idealized assumptions such 
as full-duplex radios, perfect interference cancellation, etc., but provide algebraic expressions that can be used to formulate 
tractable optimization problems which can provide useful insights. In the context of  
optimal relay placement, some researchers have already exploited this approach. 
Thakur et al., in \cite{thakur-etal10optimal-relay-location-power-allocation}, report on the problem of 
placing a single  relay node to maximize the capacity of a broadcast 
relay channel in a wideband regime. Lin et al., in \cite{lin-etal07relay-placement-80216}, 
numerically solve the problem of a single relay node placement, 
under  power-law path loss and individual power constraints at the source and the relay; however, our work 
is primarily focused on multi-relay placement, under the exponential path-loss model 
and a sum power constraint among the nodes. 
The linear deterministic 
channel model (\cite{avestimehr-etal11wireless-network-deterministic}) is used by Appuswamy et al. in  
\cite{appuswamy-etal10relay-placement-deterministic-line} to study the 
problem of placing two or more relay nodes along a line so as to maximize the end-to-end 
data rate. Our present paper is in a similar spirit; however, 
we use the achievable rate formulas for the $N$-relay channel 
(with decode and forward relays) to study the problem of placing relays on a line having length $L$, under 
a sum power constraint over the nodes.

The most important difference of our paper with the literature reported above is that 
we  address the problem of {\em sequential placement} of relay nodes along a line of an unknown random length. 
This paper extends our previous work in \cite{chattopadhyay-etal12optimal-capacity-relay-placement-line}, 
which presents the analysis for the case of given $L$ and $N$; the study under given $L$ and $N$ is a precursor to the formulation of 
as-you-go deployment problem, since it motivates 
an additive cost structure that is essential for the formulation of the 
sequential deployment problem as a Markov  decision process (MDP).  

The deploy-as-you-go problem has been addressed by previous researchers. For example, 
Howard et al., in \cite{howard-etal02incremental-self-deployment-algorithm-mobile-sensor-network}, provide heuristic
algorithms for incremental deployment of sensors in order
to cover a deployment area. Souryal et al., in \cite{souryal-etal07real-time-deployment-range-extension}, 
propose heuristic deployment algorithms for the problem of impromptu wireless network deployment, with
an experimental study of indoor RF link quality variation.  
 The authors of \cite{bao-lee07rapid-deployment-ad-hoc-backbone} propose a collaborative deployment method 
for multiple deployment agents, so that the contiguous coverage area of relays is maximized 
subject to a total number of relays constraint. 
However, until the work in  
\cite{mondal-etal12impromptu-deployment_NCC} 
and \cite{sinha-etal12optimal-sequential-relay-placement-random-lattice-path}, 
there appears to have been no effort to rigorously formulate as-you-go deployment problem in 
order to derive optimal deployment algorithms. 
The authors of \cite{mondal-etal12impromptu-deployment_NCC} and 
\cite{sinha-etal12optimal-sequential-relay-placement-random-lattice-path}  
used MDP based formulations 
to address the problem of placing relay nodes sequentially along a line and along a random lattice path, respectively. 
The formulations in \cite{mondal-etal12impromptu-deployment_NCC} and 
\cite{sinha-etal12optimal-sequential-relay-placement-random-lattice-path} 
are based on the so-called ``lone packet traffic model" under which, at any time instant, 
there can be no more than one packet traversing the network, thereby eliminating contention between wireless links. 
This work was later extended in \cite{chattopadhyay-etal13measurement-based-impromptu-placement_wiopt} 
 to the scenario where the traffic is still lone packet, 
but a measurement-based approach is employed to account for the spatial variation of 
link qualities due to shadowing.  

In this paper, we consider as-you-go deployment along a line, but move away 
from the lone-packet traffic assumption by employing information theoretic achievable rate formulas 
(for full-duplex radios and decode-and-forward relaying). 
We assume exponential path-loss model (see \cite{franceschetti-etal04random-walk-model-wave-propagation} 
and Section~\ref{subsection:motivation-for-exponential-path-loss}). To the best of our knowledge, there is no prior work that 
considers as-you-go deployment under this physical layer model.

\vspace{-4mm}
\subsection{Our Contribution}
\vspace{-2mm}

\begin{itemize}
\item  {\bf Optimal Offline Deployment:} Given the location of $N$ full-duplex relays 
to connect a source and a sink separated by a given distance $L$, and under the  
exponential path-loss model and a sum power constraint among the nodes, 
the optimal power split among the nodes and the achievable rate are expressed (Theorem~\ref{theorem:multirelay_capacity}) 
in terms of the channel gains. 
We find expression for optimal relay location in the single relay placement problem 
(Theorem~\ref{theorem:single_relay_total_power}). For the $N$ relay placement problem, 
numerical study shows that, the relay nodes are clustered 
near the source at low attenuation and are placed uniformly at high attenuation. Theorem~\ref{theorem:large_nodes_uniform} 
shows that, by placing large number of relays uniformly, 
we can achieve a rate arbitrarily close to the AWGN capacity. Only this part of our current paper 
was published in the conference version  
\cite{chattopadhyay-etal12optimal-capacity-relay-placement-line}.

\item {\bf Optimal As-You-Go Deployment:} In Section \ref{sec:mdp_total_power}, we consider the problem of placing relay nodes in a deploy-as-you-go manner, 
so as to connect 
a source and a sink separated by an unknown distance, modeled as an exponentially distributed random variable $L$. 
Specifically, the problem is to start from a source, and walk along a line, placing relay nodes as we go, 
until the line ends, at which point the sink is placed. 
With a sum power constraint, the aim is to 
maximize a  capacity limiting term derived from the deployment problem for known $L$, while 
constraining the expected number of relays. 
We  ``relax'' the expected number of relays constraint via a Lagrange multiplier, and  
formulate the problem as a total cost MDP with   uncountable state space and non-compact 
action sets. We prove the existence of an optimal policy and convergence of value iteration 
(Theorem~\ref{theorem:convergence_of_value_iteration}); these results 
for uncountable state space and non-compact action space 
are not evident from standard literature. We study properties of the value function analytically. 
This is the first time that the as-you-go deployment problem is formulated 
to maximize the end-to-end data rate under the full-duplex multi-relay channel model. 

\item {\bf Numerical Results on As-You-Go Deployment:} In Section~\ref{section:numerical_work_information_theoretic_model}, we study the policy structure numerically. 
We also demonstrate  
numerically that the proposed as-you-go algorithm achieves an end-to-end data rate sufficiently close to the maximum possible   
achievable data rate for offline placement. This is particularly important since there is no other benchmark in the literature,  
with which we can make a fair comparison of our  policy.

\item The material in Section~\ref{sec:mdp_total_power} and Section~\ref{section:numerical_work_information_theoretic_model} 
were absent in the conference version 
\cite{chattopadhyay-etal12optimal-capacity-relay-placement-line}.
\end{itemize}

\vspace{-3mm}
\subsection{Organization of the Paper}
\vspace{-1mm}
In Section \ref{sec:system_model_and_notation}, we describe our system model and notation. 
In Section \ref{sec:total_power_constraint}, we address the problem of relay placement on a line of known length. 
Section \ref{sec:mdp_total_power} deals with the problem of as-you-go deployment along a line of unknown random length. 
Numerical work on as-you-go deployment has been presented in Section~\ref{section:numerical_work_information_theoretic_model}. 
Some discussions are provided in 
Section~\ref{section:additional_discussion}. Conclusions are drawn in Section \ref{conclusion}.

\vspace{-2mm}
\section{System Model and Notation}
\label{sec:system_model_and_notation}

\subsection{The Multi-Relay Channel}
The multi-relay channel was studied in \cite{xie-kumar04network-information-theory-scaling-law} and 
\cite{reznik-etal04degraded-gaussian-multirelay-channel} and is an extension of the single relay model presented in 
\cite{cover-gamal79capacity-relay-channel}. We 
consider a network deployed on a line with a source node and a sink node at the end of the line, 
and $N$ full-duplex relay nodes as shown in Figure 
\ref{fig:general_line_network}. The relay nodes are numbered as 
$1, 2,\cdots,N$. 
The source and sink are indexed by $0$ and $N+1$, respectively. The distance of the $k$-th node {\em from the source} is denoted by 
$y_{k}:=r_{1}+r_{2}+\cdots+r_{k}$. Thus, $y_{N+1}=L$. As in \cite{xie-kumar04network-information-theory-scaling-law} 
and \cite{reznik-etal04degraded-gaussian-multirelay-channel}, we consider the {\em scalar, 
time-invariant, memoryless, AWGN setting.}

We use the model that a symbol transmitted by node~$i$ is received at node~$j$ after 
multiplication by the  (positive, real valued) channel gain $h_{i,j}$ (an assumption often made in the literature, see 
e.g., \cite{xie-kumar04network-information-theory-scaling-law} and \cite{gupta-kumar03capacity-wireless-networks}). 
 The \emph{power gain} from Node $i$ to Node $j$ is 
denoted by $g_{i,j} = h_{i,j}^2$.  We define $g_{i,i}=1$ and $h_{i,i}=1$. 
The Gaussian additive noise at any receiver is independent 
and identically distributed from symbol to symbol and has variance $\sigma^2$.

\vspace{-4mm}
\subsection{An Inner Bound to the Capacity}\label{subsec:achievable_rate_multirelay_xie_and_kumar}
\vspace{-1mm}
For the multi-relay channel, we denote the symbol transmitted by the $i$-th node at time $t$ ($t$ is discrete) by $X_{i}(t)$
 for $i=0,1,\cdots,N$. $Z_{k}(t) \sim \mathcal{N}(0,\sigma^{2})$ is the additive white Gaussian noise at 
node $k$ and time $t$, and is assumed to be {\em independent and identically distributed across $k$ and $t$.} 
Thus, at symbol time $t$, node~$k, 1 \leq k \leq N+1$ receives:
\begin{equation}
Y_{k}(t)= \sum_{j\in \{0,1,\cdots,N\}, j \neq k} h_{j,k}X_{j}(t)+Z_{k}(t) \label{eqn:network_equation}
\end{equation}
An inner bound to the capacity of this network, under any path-loss model, 
is given by (see  \cite{xie-kumar04network-information-theory-scaling-law}):
\begin{equation}
 R=\min_{1 \leq k \leq N+1} C \bigg(     \frac{1}{\sigma^{2}×} \sum_{j=1}^{k} ( \sum_{i=0}^{j-1} h_{i,k} \sqrt{P_{i,j}}  )^{2}      \bigg)    \label{eqn:achievable_rate_multirelay}
\end{equation}
where $C(x):=\frac{1}{2}\log_{2}(1+x)$, and node $i$ transmits to node $j$ at power $P_{i,j}$ (expressed in mW).


In Appendix~\ref{section:coding_scheme_description}, we provide 
a descriptive overview of the coding and decoding scheme 
proposed in \cite{xie-kumar04network-information-theory-scaling-law}. A sequence of messages are sent from the 
source to the sink; each message is encoded in a block of symbols and transmitted by using the relay nodes. 
The scheme involves {\em coherent transmission} by the source and relay nodes (this requires  
{\em symbol-level synchronization 
among the nodes}), and {\em successive interference cancellation} 
at the relay nodes and the sink. 
A node receives information 
about a message in two ways (i) by the message being directed to it cooperatively by all the previous nodes, and (ii) 
by overhearing previous transmissions of the message to the previous nodes. 
Thus node $k$ receives codes corresponding to a message $k$ times before it attempts to decode the message (a discussion 
on the practical feasibility of full-duplex decode-and-forward relaying scheme is provided 
in Section~\ref{subsection:motivation-for-full-duplex-decode-forward}).  
Note that, $C \bigg( \frac{1}{\sigma^{2}×} \sum_{j=1}^{k} ( \sum_{i=0}^{j-1} h_{i,k} \sqrt{P_{i,j}}  )^{2} \bigg)$ 
in (\ref{eqn:achievable_rate_multirelay}), for any $k$, denotes a possible rate that can be achieved by node $k$ from 
the transmissions from nodes $0,1,\cdots,k-1$. The smallest of these terms becomes the bottleneck, see 
(\ref{eqn:achievable_rate_multirelay}). 

For the single relay channel, $N=1$. Thus, by (\ref{eqn:achievable_rate_multirelay}), 
an achievable rate is given by (see also \cite{cover-gamal79capacity-relay-channel}):

\footnotesize
 \begin{eqnarray}
R=\min & \bigg\{ &  C \left(\frac{g_{0,1}P_{0,1}}{\sigma^{2}×}\right),  \nonumber\\
 && C  \left( \frac{g_{0,2}P_{0,1}+(h_{0,2}\sqrt{P_{0,2}}+h_{1,2}\sqrt{P_{1,2}})^{2}}{\sigma^{2}×}  \right)  \bigg\}\label{eqn:single_relay_genaral_capacity_formula}
\end{eqnarray}
\normalsize

Here, the first term in the $\min\{\cdot,\cdot\}$ of (\ref{eqn:single_relay_genaral_capacity_formula}) is the achievable rate at 
node $1$ (i.e., the relay node) due to the transmission from the source. The second term in the $\min\{\cdot,\cdot\}$ 
corresponds to the possible achievable rate at the sink node due to direct coherent transmission  from the 
source and the relay and due to the overheard transmission from the source to the relay. 
The higher the channel attenuation, the less will be the contribution of 
farther nodes, ``overheard'' transmissions become less relevant, 
and coherent transmission reduces to a simple transmission from the 
previous relay. The system is then closer to simple store-and-forward 
relaying.

The authors of \cite{xie-kumar04network-information-theory-scaling-law} have shown that any rate strictly less
than $R$ is achievable through the coding and
decoding scheme. This achievable rate formula can 
also be obtained from the capacity formula of a physically
degraded multi-relay channel (see \cite{reznik-etal04degraded-gaussian-multirelay-channel}),
since the capacity of the degraded relay channel is a lower
bound to the actual channel capacity. {\em In this paper, we will
seek to optimize $R$ in (\ref{eqn:achievable_rate_multirelay}) over power allocations to the nodes and
the node locations, keeping in mind that $R$ 
is a lower bound to the actual capacity. We denote the value of $R$ optimized over
power allocation and relay locations by $R^*$.}

\vspace{-2mm}
\subsection{Path-Loss Model} 
\vspace{-1mm}

We model the power gain via the exponential path-loss model: the power 
gain at a distance $r$ is $e^{-\rho r}$  where $\rho > 0$. This is a simple model used for tractability 
(see \cite{firouzabadi-martins08optimal-node-placement}, 
\cite{appuswamy-etal10relay-placement-deterministic-line}) and \cite[Section~$2.3$]{altman-etal11greec-cellular} 
for prior work assuming exponential path-loss). However, for propagation scenarios involving randomly 
placed  scatterers (as would be the case in a dense urban environment, or a forest, for example) analytical and experimental 
support has been provided for the exponential path-loss model (a discussion has been provided in 
Section~\ref{subsection:motivation-for-exponential-path-loss}). We also discuss in 
Section~\ref{subsection:insights_for_power_law_from_exponential} 
how the insights obtained from the results  for  exponential path-loss  can be used 
for power-law path-loss (power gain at a distance $r$ is $r^{-\eta}$,  $\eta>0$). 
Deployment with other path-loss models is left in this paper as a possible future work.

Under exponential 
path-loss, the channel gains and power gains in the line network become multiplicative,  
e.g., $h_{i,i+2}=h_{i,i+1}h_{i+1,i+2}$ 
and $g_{i,i+2}=g_{i,i+1}g_{i+1,i+2}$ for $i \in \{0,1,\cdots,N-1\}$.

We discuss in Section~\ref{subsection:shadowing-fading} 
how shadowing and fading can be taken care of in our model, by providing a fade-
margin in the power at each transmitter.

\vspace{-4mm}
\subsection{Motivation for the Sum Power Constraint}\label{subsec:sum_power_constraint}

In this paper we consider the sum power constraint $\sum_{i=0}^{N}P_{i}=P_{T}$ (in mW) over the source and the relays. 
This constraint has the following motivation. 
Let the fixed power expended in a relay (for reception and driving the electronic circuits) be denoted by 
$P_{\mathrm{rcv}}$ (expressed in mW), and the initial battery energy in each node be denoted by $E$ (in mJ unit). 
The information theoretic approach utilized in this paper requires that the nodes in the network are always on.  
Hence, the lifetime of node $i, 1 \leq i \leq N$, is $\tau_i=\frac{E}{P_i+P_{rcv}×}$,  
the lifetime of the source is $\tau_0=\frac{E}{P_0}$, and that of the sink is 
$\tau_{N+1}=\frac{E}{P_{rcv}}$. The rate of battery replacement at node $i$ is 
$\frac{1}{\tau_i}$. Hence, the rate at which we have to replace the batteries in 
the network is $\sum_{i=0}^{N+1}\frac{1}{\tau_i×}=\frac{1}{E×}(\sum_{i=0}^{N}P_i+(N+1)P_{rcv})$. 
The depletion rate $\frac{P_{rcv}}{E}$ is inevitable at any node, and  it does not affect the achievable data rate. 
Hence, in order to reduce 
the battery replacement rate, we must reduce the sum transmit power in the entire network.

\vspace{-2mm}
\section{Placement on a Line of Known Length}
\label{sec:total_power_constraint}

As a precursor to addressing the deploy-as-you-go problem over a line of unknown length, 
in this section we solve the problem of power constrained deployment of {\em a given number of 
relays on a line of known length}. We will often refer to this problem as {\em offline deployment problem}.  
The results of this section provide 
(i) first insights into the relay placements we obtain using the multi-relay channel model, 
(ii) a starting point for the formulation of as-you-go deployment problem, and (iii) a benchmark with which we can compare the 
performance of our as-you-go deployment algorithm.

\vspace{-2mm}
\subsection{Optimal Power Allocation}\label{subsec:optimal_power_allocation}
In this section, we consider the optimal placement of relay nodes 
on a line of given length, $L$, so as to to maximize $R$ (see (\ref{eqn:achievable_rate_multirelay})), 
subject to a total power constraint on the source and relay nodes given by $\sum_{i=0}^{N}P_{i}=P_{T}$. We will first maximize 
$R$ in (\ref{eqn:achievable_rate_multirelay}) over $P_{i,j}, 0 \leq i < j \leq (N+1)$ for any given placement 
of nodes (i.e., given $y_{1}, y_{2},\cdots,y_{N}$). This will provide an expression of achievable rate in terms of channel gains, which has to 
be maximized over $y_{1}, y_{2},\cdots,y_{N}$. Let $\gamma_{k}:=\sum_{i=0}^{k-1}P_{i,k}$ for $k \in \{1,2,\cdots,N+1\}$ 
(expressed in mW). Hence, the sum power constraint becomes $\sum_{k=1}^{N+1}\gamma_{k}=P_{T}$.

\begin{thm}\label{theorem:multirelay_capacity}
 \begin{enumerate}[label=(\roman{*})]
  \item Under the exponential path-loss model, for fixed location of relay nodes, 
    the optimal power allocation that maximizes the achievable rate for the sum power constraint is given by: 
\begin{equation}
 P_{i,j}=
\begin{cases}
\frac{g_{i,j}}{\sum_{l=0}^{j-1}g_{l,j}×}\gamma_{j}\,\, & \forall 0 \leq i <j \leq (N+1) \\
0, \,\,  &\text{if}\,\,  j \leq i 
\end{cases}\label{eqn:power_gamma_relation}
\end{equation}
where
\begin{eqnarray}
\gamma_{1}&=&\frac{P_{T}×}{1+g_{0,1}\sum_{k=2}^{N+1} \frac{(g_{0,k-1}-g_{0,k})}{g_{0,k}g_{0,k-1}\sum_{l=0}^{k-1}\frac{1}{g_{0,l}×}×}  ×} \label{eqn:gamma_one}\\
 \gamma_{j}&=&\frac{g_{0,1}\frac{(g_{0,j-1}-g_{0,j})}{g_{0,j}g_{0,j-1}\sum_{l=0}^{j-1}\frac{1}{g_{0,l}×}×}×}{1+g_{0,1}\sum_{k=2}^{N+1}
  \frac{(g_{0,k-1}-g_{0,k})}{g_{0,k}g_{0,k-1}\sum_{l=0}^{k-1}\frac{1}{g_{0,l}×}×}  ×} P_{T} \,\,\,\, \forall \, j \geq 2 \label{eqn:gamma_k}\nonumber
\end{eqnarray}
\item The achievable rate optimized over the power allocation for a given placement of nodes is given by:
\footnotesize
\begin{equation}
 R^{opt}_{P_T}(y_1,y_2,\cdots,y_N)=C \bigg( \frac{\frac{P_{T}}{\sigma^{2}×}}
{\frac{1}{g_{0,1}×}+\sum_{k=2}^{N+1}
 \frac{(g_{0,k-1}-g_{0,k})}{g_{0,k}g_{0,k-1}\sum_{l=0}^{k-1}\frac{1}{g_{0,l}×}×}  ×} \bigg)\label{eqn:capacity_multirelay}
\end{equation}
 \end{enumerate}
\normalsize
\end{thm}

\begin{proof}
 The basic idea is to choose the power levels (i.e., $P_{i,j},\, 0 \leq i<j \leq N+1$) in 
(\ref{eqn:achievable_rate_multirelay}) so that all the terms in the $\min\{\cdot\}$ 
in (\ref{eqn:achievable_rate_multirelay}) become equal. We provide explicit expressions 
for $P_{i,j},\,0 \leq i<j \leq N+1$ and the 
achievable rate (optimized over power allocation) in terms of the power gains. 
See Appendix~\ref{appendix:proof_of_multirelay_channel_capacity_theorem_after_power_allocation}  
for the detailed proof. A result on the equality of certain terms under optimal power allocation has also been proved in 
\cite{reznik-etal04degraded-gaussian-multirelay-channel} for the 
coding scheme used in \cite{reznik-etal04degraded-gaussian-multirelay-channel}. But it was 
proved in the context of a degraded Gaussian multi-relay channel, and the proof 
depends on an inductive argument, 
whereas our proof  utilizes LP (linear programming) duality. 
\end{proof}

Recalling the exponential path-loss parameter $\rho$, and the source-sink distance  $L$, let us define $\lambda := \rho L$, 
which can be treated as a measure of attenuation in the line.

Let us now comment on the results of Theorem~\ref{theorem:multirelay_capacity}:
\begin{itemize}
 \item In order to maximize $R^{opt}_{P_T}(y_1,y_2,\cdots,y_N)$, we need to place the relay nodes such that 
$\frac{1}{g_{0,1}}+\sum_{k=2}^{N+1}\frac{(g_{0,k-1}-g_{0,k})}{g_{0,k}g_{0,k-1}\sum_{l=0}^{k-1}\frac{1}{g_{0,l}×}×}$ is minimized. 
This quantity is viewed as the net attenuation the power $P_T$ faces.

\item When no relay is placed, the   attenuation is $e^{\lambda}$. 
The ratio of   attenuation with no relay and  attenuation 
with relays is called the ``relaying gain'' $G(N,\lambda)$.
\begin{equation}
 G(N,\lambda):=\frac{e^{\lambda}}{\frac{1}{g_{0,1}×}+\sum_{k=2}^{N+1}\frac{(g_{0,k-1}-g_{0,k})}{g_{0,k}g_{0,k-1}\sum_{l=0}^{k-1}\frac{1}{g_{0,l}×}×}×}
\label{eqn:gain_definition}
\end{equation}
Rate is increasing with the number of relays, and is bounded above 
by $C(\frac{P_T}{\sigma^2})$. Hence, $G(N,\lambda) \in [1, e^{\lambda}]$. 
Also, note that $G(N,\lambda)$ does not depend on $P_T$.
 
\item  By Theorem~\ref{theorem:multirelay_capacity}, we have $P_{k,j} \geq P_{i,j}$ for $i<k<j$.

\item Note that we have derived Theorem \ref{theorem:multirelay_capacity} using the fact that $g_{0,k}$ is nonincreasing in $k$. 
If there exists some $k \geq 1$ 
such that $g_{0,k}=g_{0,k+1}$, i.e, if $k$-th and $(k+1)$-st nodes are placed at the same position, then $\gamma_{k+1} = 0$, i.e., the nodes $i < k$ do not direct any power specifically to relay $k+1$. However, 
relay $k+1$ can decode the symbols received at relay $k$, and those transmitted by relay $k$. Then
relay $(k+1)$ can transmit coherently with the nodes $l \leq k$ to improve effective received power in the nodes $j > k+1$. 
\end{itemize}

\begin{figure}[t!]
\centering
\includegraphics[height=3cm, width=8cm]{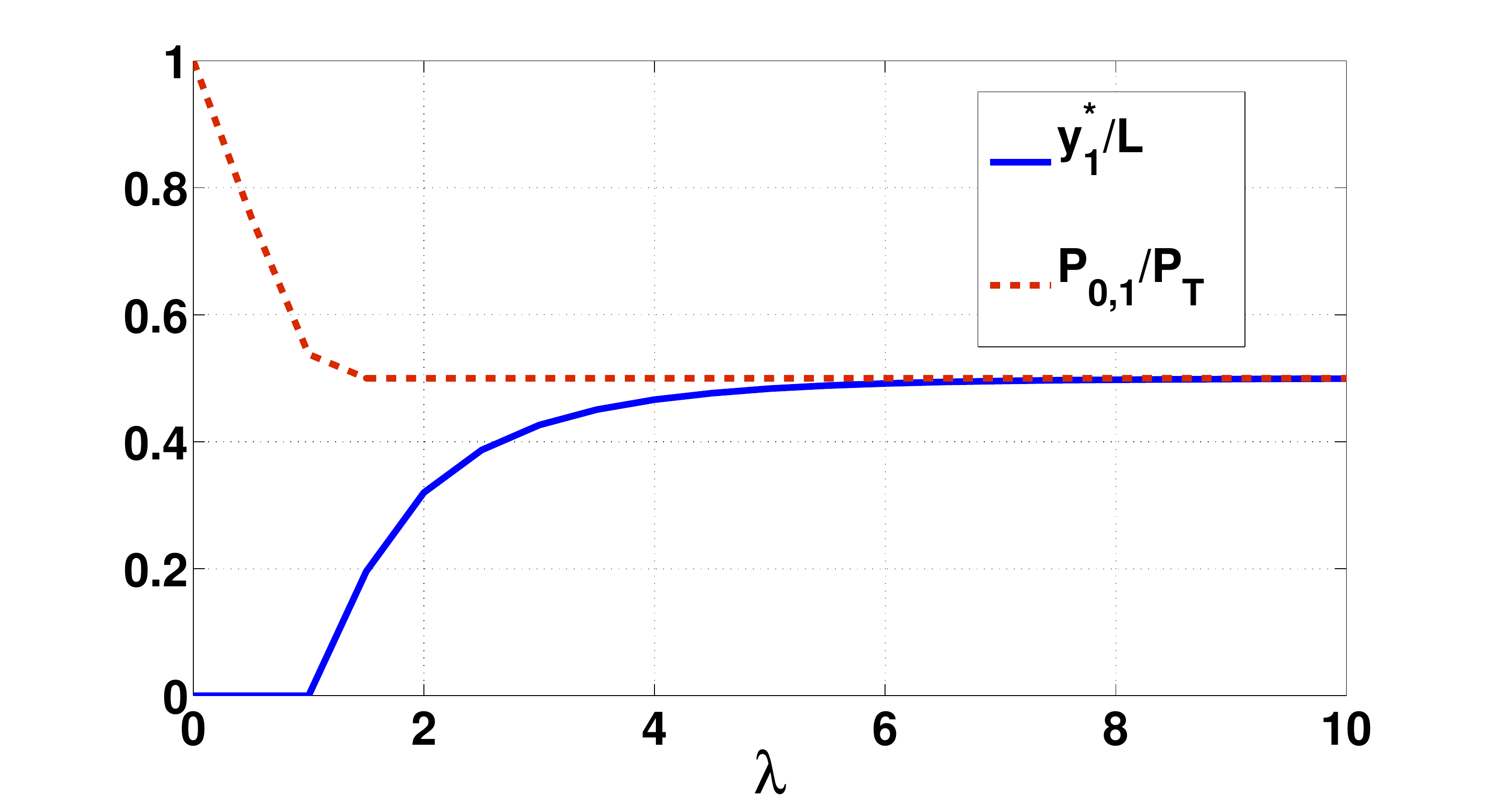}
\caption{Single relay placement, total power constraint, exponential path-loss: $\frac{y_{1}^{*}}{L×}$ and optimum $\frac{P_{0,1}}{P_{T}×}$ versus $\lambda$.}
\label{fig:single_relay_total_power}
\vspace{-6mm}
\end{figure}

\begin{figure*}[t]
\begin{minipage}[r]{0.30\linewidth}
\subfigure[$N=2$]{
\includegraphics[width=\linewidth, height=1.5cm]{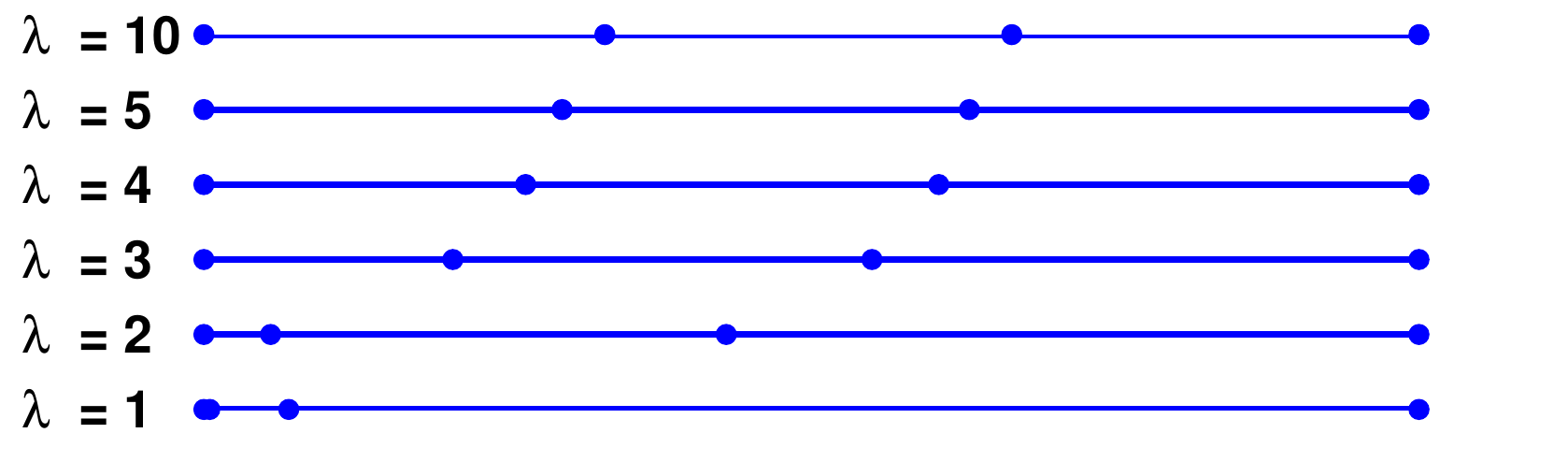}
\label{fig:2relays}}
\end{minipage}  \hfill
\begin{minipage}[c]{0.30\linewidth}
\subfigure[$N=3$]{
\includegraphics[width=\linewidth, height=1.5cm]{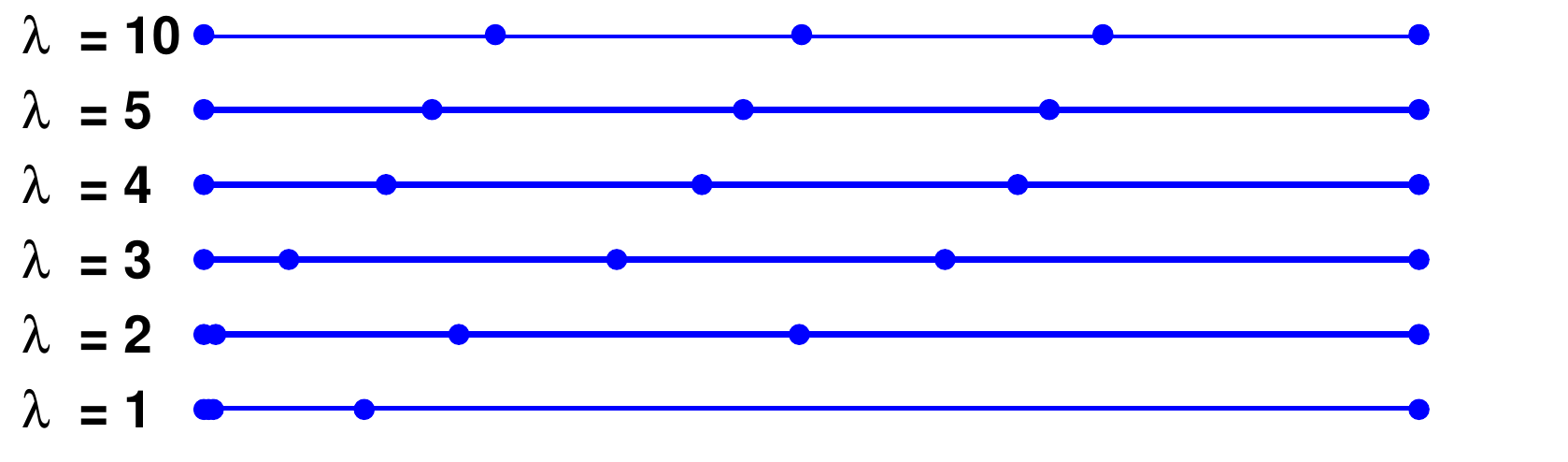}
\label{fig:3relays}}
\end{minipage} \hfill
\begin{minipage}[r]{0.30\linewidth}
\subfigure[$N=5$]{
\includegraphics[width=\linewidth, height=1.5cm]{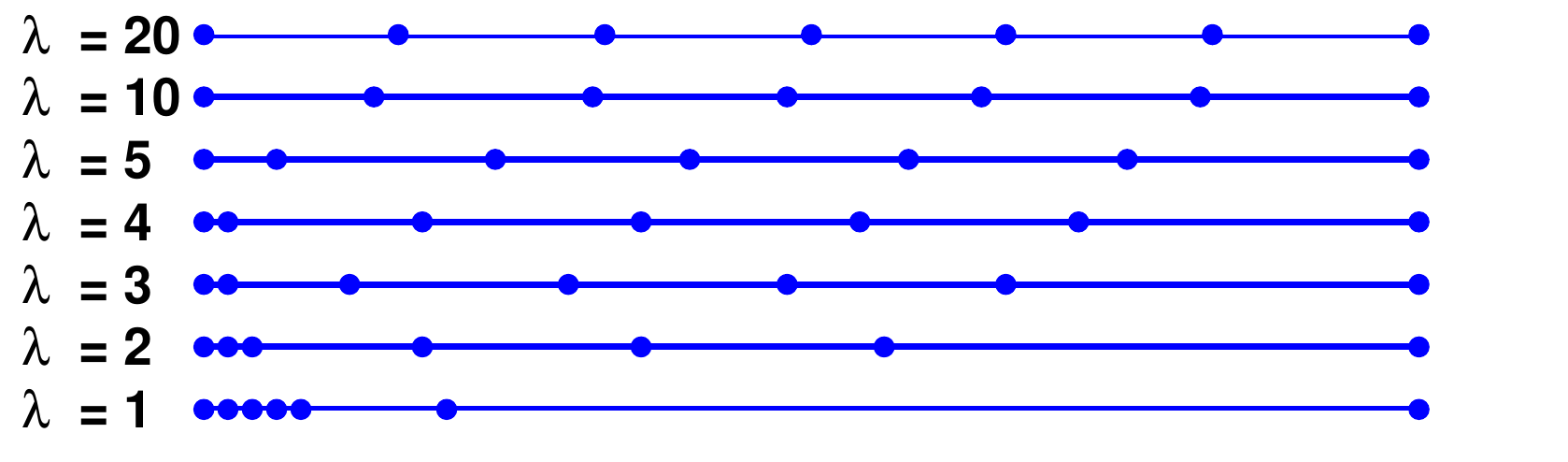}
\label{fig:5relays}}
\end{minipage}
\vspace{-6mm}
\caption[]{$N$ relays, exponential path-loss model: depiction of the optimal relay positions for $N = 2, 3, 5$ for various values of $\lambda$. 
On any line, the leftmost dot is the source node, the rightmost dot is the sink, and the intermediate 
dots denote the relays.}
\label{fig:multirelay-optimal-location}
\vspace{-4mm}
\end{figure*}

\vspace{-2mm}
\subsection{Optimal Placement of a Single Relay Node}\label{subsection:optimal_placement_single_relay_sum_power}
\vspace{-0mm}

In the following theorem, we derive the optimal power allocation, node location and data rate when a single relay is placed.

\begin{thm}\label{theorem:single_relay_total_power}

 For the single relay node placement problem with sum power constraint and exponential path-loss model, the normalized optimum relay 
location $\frac{y_{1}^{*}}{L×}$, power allocation and optimized achievable rate are given as follows:\footnote{$\log (\cdot)$ 
in this paper will mean the natural logarithm unless the base is specified.} 
 
\begin{enumerate}[label=(\roman{*})]
 
\item {\em For $\lambda \leq \log 3$}, $\frac{y_{1}^{*}}{L×}=0$, $P_{0,1}=\frac{2P_{T}}{e^{\lambda}+1×}$, 
  $P_{0,2}=P_{1,2}=\frac{e^{\lambda}-1}{e^{\lambda}+1×}\frac{P_{T}}{2×}$ and $R^{*}=C \left(\frac{2P_{T}}{(e^{\lambda}+1)\sigma^{2}×}\right)$.

\item {\em For $\lambda \geq \log 3$}, $\frac{y_{1}^{*}}{L×}=\frac{1}{\lambda×} \log \left(\sqrt{e^{\lambda}+1}-1\right)$, 
$P_{0,1}=\frac{P_{T}}{2×}$, $P_{0,2}=\frac{1}{\sqrt{e^{\lambda}+1}×}\frac{P_{T}}{2×}$, 
$P_{1,2}=\frac{\sqrt{e^{\lambda}+1}-1}{\sqrt{e^{\lambda}+1}×}\frac{P_{T}}{2×}$ and 
$R^{*}=C \left( \frac{1}{\sqrt{e^{\lambda}+1}-1×}\frac{P_{T}}{2 \sigma^{2}×} \right)$
\end{enumerate}
\end{thm}
\begin{proof}
 See Appendix~\ref{appendix:proof_of_single_relay_sum_power_results}. 
\end{proof}

{\em Discussion:} It is easy to check that $R^{*}$ obtained in 
Theorem~\ref{theorem:single_relay_total_power} is strictly greater than the 
AWGN capacity $C \left(\frac{P_{T}}{\sigma^{2}×}e^{-\lambda}\right)$ for all $\lambda>0$. This happens because the source and 
relay transmit coherently to the sink. $R^{*}$ becomes equal to the AWGN capacity only at $\lambda=0$. At $\lambda=0$, 
we do not use the relay since the sink can decode any message that the relay is able to decode.

 The variation of $\frac{y_{1}^{*}}{L×}$ and $\frac{P_{0,1}}{P_{T}}$ with $\lambda$ has been shown in 
Figure~\ref{fig:single_relay_total_power}. 
We observe that (from Figure~\ref{fig:single_relay_total_power} and Theorem~\ref{theorem:single_relay_total_power}) 
$\lim_{\lambda \rightarrow \infty} \frac{y_{1}^{*}}{L}=\frac{1}{2×}$, $\lim_{\lambda \rightarrow \infty}P_{0,2}=0$
 and $\lim_{\lambda \rightarrow 0}P_{0,1}=P_{T}$. For large values of $\lambda$, source and relay cooperation provides negligible 
benefit since source to sink attenuation is very high. So it is optimal to 
place the relay at a distance $\frac{L}{2×}$. The relay 
works as a repeater which forwards data received from the source to the sink. For small $\lambda$, the gain obtained 
from coherent
transmission is dominant, and, in order to receive sufficient information (required for coherent transmission) 
from the source, the relay is placed near the source.

\vspace{-3mm}
\subsection{Optimal  Placement for a Multi-Relay Channel}\label{subsection:properties_of_gain}
\vspace{-2mm}
As we discussed earlier, we need to place $N$ relay nodes such that 
$\frac{1}{g_{0,1}×}+\sum_{k=2}^{N+1} \frac{(g_{0,k-1}-g_{0,k})}{g_{0,k}g_{0,k-1}\sum_{l=0}^{k-1}\frac{1}{g_{0,l}×}×}$ is minimized.
Here $g_{0,k}=e^{-\rho y_{k}}$. We have the constraint $0 \leq y_{1} \leq y_{2} \leq \cdots \leq y_{N} \leq y_{N+1}=L$. Now, writing 
$z_{k}=e^{\rho y_{k}}$, and defining $z_{0}:=1$, we arrive at the following  problem:
\begin{eqnarray}
& & \min \bigg\{ z_{1}+\sum_{k=2}^{N+1} \frac{z_{k}-z_{k-1}}{\sum_{l=0}^{k-1} z_{l}}\bigg\}\nonumber\\
& s.t & \,\, 1 \leq z_{1} \leq \cdots \leq z_{N} \leq z_{N+1} = e^{\lambda} \label{eqn:multirelay_optimization}
\end{eqnarray}
The objective function is 
convex in each of the variables $z_{1}, z_{2},\cdots, z_{N}$. The objective function is sum of linear fractionals, and
the constraints are linear. 

{\em Remark:} From optimization problem (\ref{eqn:multirelay_optimization}) we observe that optimum 
$z_{1},z_{2},\cdots,z_{N}$ depend only on $\lambda:=\rho L$. Since 
$z_{k}=e^{\lambda \frac{y_{k}}{L}}$, the normalized optimal distance of relays from the source depend only on $\lambda$ and 
$N$.

\begin{thm}\label{theorem:capacity_increasing_with_N}
 For fixed $\rho$, $L$ and $\sigma^{2}$, the optimized achievable rate $R^{*}$ for a sum power constraint 
\textit{strictly} increases with the 
number of relay nodes.
\end{thm}
\begin{proof}
 See Appendix~\ref{appendix:proof_of_capacity_increases_in_N}.
\end{proof}

\begin{thm}\label{theorem:G_increasing_in_lambda}
 For any fixed number of relays $N \geq 1$, $G(N,\lambda)$ is increasing in $\lambda$.
\end{thm}
\begin{proof}
 See Appendix~\ref{appendix:proof_of_G_increasing_in_lambda}.
\end{proof}

{\bf A numerical study of multi-relay placement:} 
We discretize the interval $[0,L]$ and run a search program to 
find normalized optimal relay locations for different values of $\lambda$ and $N$. 
 The results 
are summarized in Figure~\ref{fig:multirelay-optimal-location}. 

We observe that at 
low attenuation (small $\lambda$), relay nodes are clustered near the source node and are often at the source node, whereas 
at high attenuation (large $\lambda$) they are almost uniformly 
placed along the line. For large $\lambda$, the effect of long distance 
between any two adjacent nodes dominates the gain obtained by 
coherent relaying. Hence, it is beneficial to minimize the maximum distance between any two adjacent nodes and thus 
multihopping is a better strategy in this case. For small $\lambda$, 
the gain obtained by coherent transmission is dominant.
 In order to allow this, relays should be able to receive sufficient information from their previous 
nodes. Thus, they tend to be clustered near the source.

In Figure~\ref{fig:G_vs_N} we plot the relaying gain $G(N,\lambda)$ in dB  vs. 
the number of relays $N$, for various values of $\lambda$. As proved in Theorem~\ref{theorem:capacity_increasing_with_N}, 
we see that $G(N,\lambda)$ increases with $N$ for fixed $\lambda$. 
On the other hand, $G(N,\lambda)$ increases with $\lambda$ for fixed $N$, as proved in Theorem~\ref{theorem:G_increasing_in_lambda}.

\begin{figure}[t!]
\centering
\includegraphics[height=3cm, width=8cm]{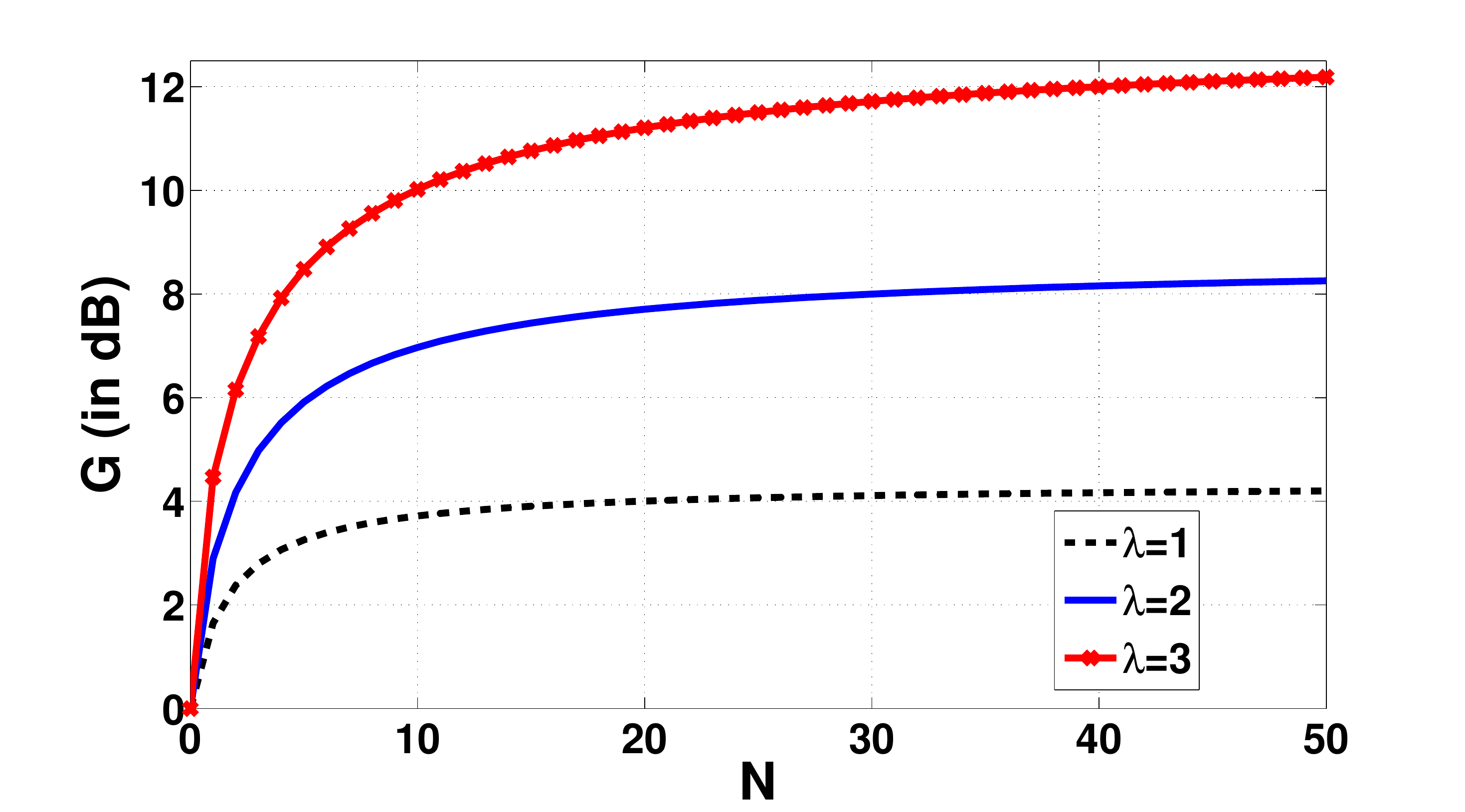}
\caption{$G (N,\lambda)$ vs $N$ for total power constraint.}
\label{fig:G_vs_N}
\vspace{-5mm}
\end{figure}

\vspace{-1mm}
\subsection{Uniformly Placed Relays, Large $N$}

When the relays are uniformly placed, the behaviour of $R^{opt}_{P_T}(y_1,\cdots,y_N)$ 
(called $R_{N}$ in the next theorem) for large number of relays 
is captured by the following:

\begin{thm}\label{theorem:large_nodes_uniform}
For exponential path-loss and sum power constraint, if $N$ relay nodes are placed uniformly between the source and the sink, 
resulting in $R_{N}$ achievable rate, then $ \lim_{N \rightarrow \infty} R_N= C \left(\frac{P_{T}}{\sigma^{2}×}\right)$.
\end{thm}
\begin{proof}
 See Appendix~\ref{appendix:proof_of_large_nodes_uniform}.
\end{proof}

 {\em Remark:} From Theorem \ref{theorem:large_nodes_uniform}, it is clear that 
we can achieve a rate arbitrarily close to  $C(\frac{P_{T}}{ \sigma ^{2}})$ (i.e., the effect of path-loss can 
completely be mitigated) by placing a large enough number of relay nodes. In this context, we would like to mention 
that the variation of broadcast capacity as a function of the number of nodes $N$ 
(located randomly inside a unit square) was studied in 
 \cite{zheng06information-dissemination}; but the broadcast capacity in their paper increases with $N$ since they assume 
per-node power constraint.

\vspace{-3mm}
\section{As-You-Go Deployment of Relays on a Line of Unknown Length}\label{sec:mdp_total_power}
Having developed the problem of placing a given number of relays over a line of fixed, 
given length, we now turn to the deploy-as-you-go problem. An agent walks 
along a line, starting from the source and heading towards the sink 
which is at an unknown distance from the source 
location, deploying relays as he goes, so as to achieve a multi-relay 
network when he encounters the sink location (and places the sink there). We model the 
distance from the source to sink as an 
exponentially distributed random variable with mean $\overline{L}=\frac{1}{\beta}$.\footnote{A motivation  
for the use of the exponential distribution, given the prior
knowledge of the mean length $\overline{L}$, is that it is the maximum entropy continuous probability density
function with the given mean. By using the exponential distribution, we
are leaving the length of the line as uncertain as we can, given the prior
knowledge of its mean.
} The deployment objective is to achieve a high data rate from the source to the sink, subject  
to a total power constraint and a constraint on the expected number of relays placed (note that, the number of relay nodes, $N$, is a 
random variable here, due to the randomness in $L$). 
Using the rate expression from 
Theorem~\ref{theorem:multirelay_capacity}, we formulate the problem as a total cost MDP.

Such a deployment problem could be motivated by a situation where it is required 
to place a sensor (say, a video camera) to monitor an event or an object from a safe distance 
(e.g., the battlefront in urban combat, or a suspicious object that needs to be detonated, 
or  a group of animals in a forest). In such a situation, 
the deployment agent, after placing the sensor, 
walk away from the scene of the event, along a forest trail, or a road, or a building corridor, 
placing relays as he walks, until a suitable safe sink location is found, 
in such a way that the number of relays is kept small while the end-to-end data rate is maximized.

\vspace{-2mm}
\subsection{Formulation as an MDP}\label{subsection:mdp_formulation} 
We now formulate the as-you-go deployment problem as an MDP.

\subsubsection{Deployment Policies}
In the as-you-go placement problem, the person carries a number of nodes and places them as he walks, 
under the control of a placement policy. A deployment policy $\pi$ is a 
sequence of mappings $\{\mu_1,\mu_2, \mu_3, \cdots \}$ from the state space to the action space;  
at the $k$-th decision instant (i.e., after placing the $(k-1)$-st relay),  
 $\mu_k$ provides the distance at which the next relay should be placed (provided that the 
line does not end before that point), given  the system state which is a function of the locations of 
previously placed nodes.  Thus, the decisions are made based on the locations of the relays placed earlier. 
The first decision instant is the start of the line, and the subsequent decision instants are the placement points of the relays. 
Let $\Pi$ denote the set (possibly uncountable) of all deployment policies. Let $\mathbb{E}_{\pi}$ denote the 
expectation under policy $\pi$.

\subsubsection{The Unconstrained Problem}
We recall from (\ref{eqn:capacity_multirelay}) that for a fixed length $L$ of the line and a fixed $N$,
$e^{\rho y_1}+\sum_{k=2}^{N+1}\frac{e^{\rho y_k}-e^{\rho y_{k-1}}}{1+e^{\rho y_1}+\cdots+e^{\rho y_{k-1}}×}$ has to be minimized in 
order to maximize $R^{opt}_{P_T}(y_1,y_2,\cdots,y_N)$. 
$e^{\rho y_1}+\sum_{k=2}^{N+1}\frac{e^{\rho y_k}-e^{\rho y_{k-1}}}{1+e^{\rho y_1}+\cdots+e^{\rho y_{k-1}}×}$ is basically a scaling 
factor which captures the effect of attenuation and relaying on the maximum possible SNR $\frac{P_T}{\sigma^2}$. 

 Let $\xi>0$ be the cost of placing a relay. 
We are interested in solving the following problem:
\begin{eqnarray}
 \inf_{\pi\in \Pi } \mathbb{E}_{\pi} \left( \left(e^{\rho y_1}+
\sum_{k=2}^{N+1}\frac{e^{\rho y_k}-e^{\rho y_{k-1}}}{1+e^{\rho y_1}+
\cdots+e^{\rho y_{k-1}}×} \right)+\xi N \right)\label{eqn:unconstrained_mdp}
\end{eqnarray}
The ``cost'' function inside the outer parentheses in (\ref{eqn:unconstrained_mdp}) has  two terms, 
one is the denominator of $G(N,\lambda)$ 
in (\ref{eqn:gain_definition}), and the other is a linear multiple of the number of relays. Thus, the cost function 
captures the tradeoff between the cost of placing relays (quantified as $\xi$ per relay), and the need 
to achieve high end-to-end data rate by making the denominator of $G(N,\lambda)$ small. Note that, due to the randomness 
in the length of the line, the $y_k, k \geq 1,$ 
and $N$ are all random variables.\footnote{Recall Section~\ref{subsec:sum_power_constraint}. The 
battery depletion rate $\frac{P_{rcv}}{E}$ of a node due to the receive power alone can be absorbed into the relay cost $\xi$.} 

We will see in Theorem~\ref{theorem:convergence_of_value_iteration} that an optimal policy 
always exists for this problem.

\subsubsection{The Constrained Problem}
Solving the problem in (\ref{eqn:unconstrained_mdp}) also helps in solving the following constrained problem:
\begin{eqnarray}
& & \inf_{\pi\in \Pi } \mathbb{E}_{\pi} \left(e^{\rho y_1}+
\sum_{k=2}^{N+1}\frac{e^{\rho y_k}-e^{\rho y_{k-1}}}{1+e^{\rho y_1}+\cdots+e^{\rho y_{k-1}}×}\right)\nonumber\\
&\textit{s.t., }&\mathbb{E}_{\pi}(N) \leq M \label{eqn:constrained_mdp}
\end{eqnarray}
where $M>0$ is a constraint on the expected number of relays. \footnote{The constraint on the mean number of relays can
be justified if we consider the relay deployment problem for
multiple source-sink pairs over several  lines of mean
length $\overline{L}$, given a large pool of relays, and we are only
interested in keeping small the total number of relays over
all these deployments.} The following standard result 
(see \cite[Theorem~$4.3$]{beutler-ross85optimal-policies-controlled-markov-chains-constraint}) gives the optimal $\xi^*$:
\begin{lem}\label{lemma:choice-of-xi}
If there exists $\xi^*>0$ and a policy $\pi_{\xi^*}^* \in \Pi$ such that 
$\pi_{\xi^*}^*$ is an optimal policy for the unconstrained problem (\ref{eqn:unconstrained_mdp}) under $\xi^*$ and  
$\mathbb{E}_{\pi_{\xi^*}^*}N=M$, then $\pi_{\xi^*}^*$ is also optimal for the constrained problem (\ref{eqn:constrained_mdp}).
\end{lem}

The motivation behind formulation (\ref{eqn:constrained_mdp}) is as follows. 
Suppose that one seeks to solve the following problem:

\footnotesize
\begin{eqnarray}
&& sup_{\pi } \mathbb{E}_{\pi} \log_2 \bigg(1+\frac{\frac{P_T}{\sigma^2×}}{e^{\rho y_1}+
\sum_{k=2}^{N+1}\frac{e^{\rho y_k}-e^{\rho y_{k-1}}}{1+e^{\rho y_1}+\cdots+e^{\rho y_{k-1}}×}}\bigg)\nonumber\\
&&\textit{s.t., } \mathbb{E}_{\pi}(N) \leq M \label{eqn:constrained_mdp_actual}
\end{eqnarray}
\normalsize

Since $\log_2 (1+\frac{1}{x})$ is convex in $x$, we can argue by Jensen's 
inequality that by solving (\ref{eqn:constrained_mdp}) we actually find a 
relay placement policy that maximizes a lower bound 
to the expected achievable data rate obtained from (\ref{eqn:constrained_mdp_actual}). 
But formulation (\ref{eqn:constrained_mdp}) (and hence formulation (\ref{eqn:unconstrained_mdp}), by 
Lemma~\ref{lemma:choice-of-xi}) allows us to write the 
objective function as a summation of hop-costs; this motivates us to 
formulate the as-you-go deployment problem as an MDP, resulting in a substantial reduction in policy 
computation. 
However, in 
Section~\ref{section:numerical_work_information_theoretic_model}, we will show numerically that 
solving (\ref{eqn:constrained_mdp}) is a reasonable approach to deal with the computational complexity of 
(\ref{eqn:constrained_mdp_actual}); we will see that 
formulation (\ref{eqn:constrained_mdp}) allows us to achieve a reasonable performance.

We now formulate the above ``as-you-go" relay placement problem 
(\ref{eqn:unconstrained_mdp}) as a total cost Markov decision process.

\subsubsection{State Space, Action Space and Cost Structure} Let us define 
$s_0:=1$, $s_k:=\frac{e^{\rho y_k}}{1+e^{\rho y_1}+\cdots+e^{\rho y_k}×}\, \forall \, k \geq 1$. Also, recall that 
$r_{k+1}=y_{k+1}-y_k$. Thus, we can rewrite (\ref{eqn:unconstrained_mdp}) as:
\begin{eqnarray}
 \inf_{\pi\in \Pi } \mathbb{E}_{\pi} \left(1+\sum_{k=0}^{N}s_k(e^{\rho r_{k+1}}-1)+\xi N \right) \label{eqn:unconstrained_mdp_with_state}
\end{eqnarray}

When the person starts walking from the source along the line, the state 
of the system is set to $s_0:=1$. At this instant the 
placement policy provides the location at which the first relay should be placed. The person walks towards the 
prescribed placement point. If the sink placement location is encountered before reaching this point, 
the sink 
is placed; if not, then the first relay is placed at the placement point. In general, the state after placing the 
$k$-th relay is denoted by $s_k$ (a function of the location of the nodes up to the $k$-th instant), 
for $k=1,2,\cdots$. At state $s_k$, 
the action is the distance $r_{k+1}$ where the next relay has 
to be placed (action $\infty$ means that no further relay 
will be placed). If the line ends before this distance, the sink node has to be placed at the end. 
{\em The randomness is coming from the random residual length of the line.} Let $l_k$ 
denote the residual length  at the $k$-th instant.

With this notation, the state of the system evolves as:
\begin{eqnarray}
 s_{k+1}=
\begin{cases}
\frac{s_k e^{\rho r_{k+1}}}{1+s_k e^{\rho r_{k+1}}×}, \text{ if $l_k > r_{k+1}$} , \\
 \mathbf{EOL}, \text{  \hspace{7mm}   else}. \label{eqn:state-evolution}
\end{cases}
\end{eqnarray}
Here $\mathbf{EOL}$ denotes the end of the line, i.e., the termination state.

The single stage cost (for problem (\ref{eqn:unconstrained_mdp_with_state})) for state  $s$, 
action $a$ and  residual length  $l$, is:
\begin{eqnarray}
 c(s,a,l)=
\begin{cases}
\xi + s(e^{\rho a}-1), \text{ if $l>a$},\\
s(e^{\rho l}-1), \text{\hspace{6.5mm}  else}. \label{eqn:single-stage-cost}
\end{cases}
\end{eqnarray}
Also, $c(\mathbf{EOL},a,\cdot)=0$ for all $a$.

From (\ref{eqn:state-evolution}), it is clear that the next state $s_{k+1}$ 
depends on the current state $s_k$, the current action $r_{k+1}$ 
and the residual length of the line. Since the length of the line is exponentially distributed, 
from any placement point, the residual line length is exponentially 
distributed, and independent of the history of the process.  The cost incurred at the $k$-th decision instant 
is given by (\ref{eqn:single-stage-cost}), which depends on $s_k$, $r_{k+1}$ and $l_k$. 

Hence, our formulation in (\ref{eqn:unconstrained_mdp_with_state}) is an MDP with 
state space $\mathcal{S}:=(0,1]\cup \{\mathbf{EOL}\}$ 
and action space $\mathcal{A}\cup \{\infty \}$ where $\mathcal{A}:=[0,\infty)$. 

{\em Remark:} An optimal policy (if it exists) for the problem (\ref{eqn:unconstrained_mdp}) will be used to place relay nodes along a line 
whose length is a sample from an exponential distribution with mean $\frac{1}{\beta×}$. After the deployment is over, the 
power $P_T$ will be shared optimally among the source and the deployed relay 
nodes (according to Theorem~\ref{theorem:multirelay_capacity}).

\vspace{-2mm}
\subsection{Optimal Value Function}\label{subsection:analysis_of_mdp}
\vspace{-2mm}
Suppose $s_k=s$ for some $k \geq 0$. Then, the optimal value function (cost-to-go) at state $s$ is defined by:
\begin{equation*}
 J_{\xi}(s)=\inf_{\pi \in \Pi} \mathbb{E} \left(\sum_{n=k}^{\infty} c(s_n, a_n, l_n)|s_k=s \right)
\end{equation*}
If we decide to place the next relay at a distance $a<\infty$ and follow the optimal policy thereafter, the expected cost-to-go 
at a state $s \in (0,1]$ becomes:

\footnotesize
\begin{equation}
  \int_{0}^{a}\beta e^{-\beta z}s(e^{\rho z}-1)dz 
 + e^{-\beta a}\bigg(s(e^{\rho a}-1)+\xi+J_{\xi}\left(\frac{se^{\rho a}}{1+se^{\rho a}×}\right)\bigg)\label{eqn:cost-to-go}
\end{equation}
\normalsize
The first term in (\ref{eqn:cost-to-go}) corresponds to the case in which the line ends 
at a distance less than $a$ and we are forced to place the sink node. 
The second term corresponds to the case where the residual length of the line is greater than $a$ and a relay is placed 
at a distance $a$. 

Note that our MDP has an uncountable state space $\mathcal{S}=(0,1] \cup \{\mathbf{EOL}\}$ and a non-compact action space 
$\mathcal{A}=[0,\infty) \cup \{\infty \}$. Several technical issues arise in this kind of problems, such as 
the existence of optimal or $\epsilon$-optimal policies, measurability of the policies, etc. 
We, therefore, invoke the results 
provided by Sch\"{a}l \cite{schal75conditions-optimality}, which deal with such issues. 
Our problem is one 
of minimizing total, undiscounted, non-negative costs over an infinite 
horizon. Equivalently, in the context of \cite{schal75conditions-optimality}, we have a problem of  
total reward maximization where the rewards are the negative of the costs. Thus, our problem specifically fits into 
the negative dynamic programming setting 
of \cite{schal75conditions-optimality} (i.e., the $\mathsf{N}$ case where single-stage rewards are non-positive).

Now,  the state $\mathbf{EOL}$ is absorbing. Also, no action is taken at 
this state and the cost at this state is $0$. Hence, we can think of this state as state $0$ in order to 
make our state space a Borel subset of the real line. 
\begin{thm}\label{thm:schal_bellman_eqn}
[\cite{schal75conditions-optimality}, Equation (3.6)]      
The optimal value function $J_{\xi}(\cdot)$ satisfies the Bellman equation. \qed
\end{thm}

Thus, $J_\xi(\cdot)$ satisfies the following Bellman equation for each $s \in (0,1]$:

\footnotesize
\begin{eqnarray}
 J_{\xi}(s) &=& \min \bigg \{\inf_{a \geq 0} \bigg[\int_{0}^{a}\beta e^{-\beta z}s(e^{\rho z}-1)dz \nonumber\\
&& + e^{-\beta a}\bigg(s(e^{\rho a}-1)+\xi+J_{\xi}\left(\frac{se^{\rho a}}{1+se^{\rho a}×}\right)\bigg)\bigg], \nonumber\\
&& \int_{0}^{\infty}\beta e^{-\beta z}s(e^{\rho z}-1)dz \bigg \} \label{eqn:bellman_unbroken} 
\end{eqnarray}
\normalsize

where the second term inside $\min\{\cdot, \cdot \}$ is the cost of not placing any relay (i.e., $a=\infty$).

We analyze the MDP for $\beta>\rho$ and $\beta \leq \rho$.

\subsubsection{Case I ($\beta>\rho$)}

We observe that the cost of not placing any relay (i.e., $a=\infty$) at state $s \in (0,1]$ is given by:
\begin{eqnarray*}
 \int_{0}^{\infty} \beta e^{-\beta z}s(e^{\rho z}-1)dz=\theta s
\end{eqnarray*}
where $\theta:=\frac{\rho}{\beta-\rho×}$ (using the fact that $\beta>\rho$). 
Since not placing a relay (i.e., $a = \infty$) is a possible action for every $s$, it follows that $J_{\xi}(s)\leq \theta s $.

The cost in (\ref{eqn:cost-to-go}), upon simplification, can be written as:
\begin{eqnarray}
 \theta s + e^{-\beta a}\bigg(-\theta s e^{\rho a}+\xi+
J_{\xi}\left(\frac{se^{\rho a}}{1+se^{\rho a}×}\right)\bigg)\label{eqn:cost_to_go_beta_greater_than_rho}
\end{eqnarray}
Since $J_{\xi}(s) \leq \theta$ for all $s \in (0,1]$, the expression in (\ref{eqn:cost_to_go_beta_greater_than_rho}) 
is strictly less that 
$\theta s$ for large enough $a<\infty$. Hence, according to (\ref{eqn:bellman_unbroken}), it is not optimal to not place any relay  
and the Bellman equation (\ref{eqn:bellman_unbroken}) can be rewritten as:

\footnotesize
\begin{eqnarray}
 J_{\xi}(s)= \theta s + \inf_{a \geq 0} e^{-\beta a}\bigg(-\theta s e^{\rho a}+
\xi+J_{\xi}\left(\frac{se^{\rho a}}{1+se^{\rho a}×}\right)\bigg)\label{eqn:bellman_equation_simplified_in_a}
\end{eqnarray}
\normalsize

\subsubsection{Case II ($\beta \leq \rho$)}
Here the cost in (\ref{eqn:cost-to-go}) is $\infty$ if we do not place a relay (i.e., if $a=\infty$). 
Let us consider a policy $\pi_1$ where we place the next relay at a 
fixed distance $0 <a <\infty$ from the current relay, 
irrespective of the current state. If the residual length of the line 
is $z$ at any state $s$, we will place less than $\frac{z}{a×}$ additional 
 relays, and 
for each relay a cost less than $(\xi+(e^{\rho a}-1))$ is incurred (since $s \leq 1$). At the last step when we place the 
sink, a cost less than $(e^{\rho a}-1)$ is incurred.  Thus, the value function of this policy is upper bounded by:
\begin{eqnarray}
&& \int_{0}^{\infty} \beta e^{-\beta z} \frac{z}{a×}(\xi+(e^{\rho a}-1)) dz+(e^{\rho a}-1) \nonumber\\
&=& \frac{1}{\beta a ×}\left(\xi+(e^{\rho a}-1)  \right)+(e^{\rho a}-1)\label{eqn:upper_bound_on_cost_beta_leq_rho}
\end{eqnarray}
Hence, $J_{\xi}(s) \leq \frac{1}{\beta a ×}\left(\xi+(e^{\rho a}-1)  \right)+(e^{\rho a}-1) < \infty$. 
Thus, by the same argument as in 
the case $\beta > \rho$, the minimizer in the Bellman equation lies in $[0,\infty)$, i.e., 
the optimal placement distance lies in $[0,\infty)$. Hence, (\ref{eqn:bellman_unbroken}) can be rewritten as:

\footnotesize
\begin{eqnarray}
 J_{\xi}(s) &=& \inf_{a \geq 0} \bigg\{\int_{0}^{a}\beta e^{-\beta z}s(e^{\rho z}-1)dz + \nonumber\\
&& e^{-\beta a}\bigg(s(e^{\rho a}-1)+\xi+J_{\xi}\left(\frac{se^{\rho a}}{1+se^{\rho a}×}\right)\bigg)\bigg\}
\label{eqn:bellman_beta_leq_rho} 
\end{eqnarray} 
\normalsize

\subsection{Upper Bound on the Optimal Value Function}

\begin{prop}\label{prop:upper_bound_on_cost_beta_geq_rho}
If $\beta>\rho$, then $J_{\xi}(s) < \theta s$ for all $s \in (0,1]$.
\end{prop}
\begin{proof}
 We know that $J_{\xi}(s) \leq \theta s \leq \theta$. Now, let us consider 
the Bellman equation (\ref{eqn:bellman_equation_simplified_in_a}). 
It is easy to see that $(-\theta s e^{\rho a}+\xi+J_{\xi}(\frac{se^{\rho a}}{1+se^{\rho a}×}))$ 
is strictly negative for sufficiently large $a$. 
Hence, the R.H.S of (\ref{eqn:bellman_equation_simplified_in_a}) is strictly less than $\theta s$. 
\end{proof}
\begin{cor}
 For $\beta>\rho$, $\lim_{s \rightarrow 0} J_{\xi}(s) \rightarrow 0$ for any $\xi>0$.
\end{cor}
\begin{proof}
 Follows from Proposition \ref{prop:upper_bound_on_cost_beta_geq_rho}.
\end{proof}
\begin{prop}\label{prop:upper_bound_on_cost}
If $\beta>0$ and $\rho>0$ and $0<a<\infty$, then $J_{\xi}(s) <\frac{1}{\beta a ×}\left(\xi+(e^{\rho a}-1)  \right)+(e^{\rho a}-1)$ for all $s \in (0,1]$.
\end{prop}
\begin{proof}
 Follows from (\ref{eqn:upper_bound_on_cost_beta_leq_rho}), since the analysis is valid even for $\beta>\rho$.
\end{proof}

\vspace{-2mm}
\subsection{Convergence of the Value Iteration}

The value iteration for our MDP is given by:

\footnotesize
\begin{eqnarray}
 J_{\xi}^{(k+1)}(s) &=& \inf_{a \geq 0} \bigg\{\int_{0}^{a}\beta e^{-\beta z}s(e^{\rho z}-1)dz + e^{-\beta a}\bigg(s(e^{\rho a}-1) \nonumber\\
&&+\xi+J_{\xi}^{(k)}\left(\frac{se^{\rho a}}{1+se^{\rho a}×}\right)\bigg)\bigg\}\label{eqn:value_iteration}
\end{eqnarray} 
\normalsize
Here $J_{\xi}^{(k)}(s)$ is the $k$-th iterate of the value iteration. 
Let us start with $J_{\xi}^{(0)}(s):=0$ for all $s \in (0,1]$. We set $J_{\xi}^{(k)}(\mathbf{EOL})=0$ for all $k \geq 0$. 
$J_{\xi}^{(k)}(s)$ is the optimal value function 
for a problem with the same single-stage cost and the same transition structure, but with the horizon 
length being $k$ (instead of infinite horizon as in our original problem) and $0$ terminal cost. Here, by horizon length $k$, 
we mean that there are $k$ number of relays available for deployment.

Let $\Gamma_k(s)$ be the set of minimizers of (\ref{eqn:value_iteration}) at the 
$k$-th iteration at state $s$, if the infimum is achieved at some $a<\infty$. 
Let $\Gamma_{\infty}(s):=\{a \in \mathcal{A}:a$ be an 
accumulation point of some sequence $\{a_k\}$ where each $a_k \in \Gamma_{k}(s)\}$. Let $\Gamma^*(s)$ be the set of minimizers in 
(\ref{eqn:bellman_beta_leq_rho}). In Appendix~\ref{appendix:sequential_placement_total_power}, we show 
that $\Gamma_k(s)$ for each $k \geq 1$, $\Gamma_{\infty}(s)$ and $\Gamma^*(s)$ are nonempty.

\begin{thm}\label{theorem:convergence_of_value_iteration}
 The value iteration given by (\ref{eqn:value_iteration}) has the following properties: 
\begin{enumerate}[label=(\roman{*})]
 \item $J_{\xi}^{(k)}(s)\rightarrow J_{\xi}(s)$ for all $s \in (0,1]$, i.e., the value iteration converges 
       to the optimal value function.
\item $\Gamma_{\infty}(s) \subset \Gamma^*(s)$.
\item There is a stationary optimal policy $f^{\infty}=\{f,f,f,\cdots\}$ where $f:(0,1] \rightarrow \mathcal{A}$ and $f(s) \in \Gamma_{\infty}(s)$ 
for all $s \in (0,1]$.
\end{enumerate}
\end{thm}

\begin{proof}
The proof is given in Appendix~\ref{appendix:proof_of_value_iteration_convergence}. 
It uses some results from \cite{schal75conditions-optimality}, which have been 
discussed first. Next, we provide a general 
theorem (Theorem~\ref{thm:value_iteration_general}) on the convergence of value iteration, which has been 
used to prove Theorem~\ref{theorem:convergence_of_value_iteration}. 
\end{proof}

{\em Remark:} Since the action space is noncompact, 
it is not obvious from standard results whether the optimal policy exists. 
However, we are able to show that in our problem, for each state $s \in (0,1]$, the optimal action will lie 
in a compact set of the from $[0,a(s)]$, where $a(s)$ is 
continuous in $s$, and $a(s)$ could possibly go to 
$\infty$ as $s \rightarrow 0$. The 
results of \cite{schal75conditions-optimality} allow us to work with the scenario where for each state $s$, it is 
sufficient to focus only on a compact action space $[0,a(s)]$.

\vspace{-3mm}
\subsection{Properties of the Value Function $J_{\xi}(s)$}

\begin{prop}\label{prop:increasing_concave_in_s}
 $J_{\xi}(s)$ is increasing and concave over $s \in (0,1]$.
\end{prop}
\begin{prop}\label{prop:increasing_concave_in_lambda}
 $J_{\xi}(s)$ is increasing and concave in $\xi$ for all $s \in (0,1]$.
\end{prop}
\begin{prop}\label{prop:continuity_of_cost}
 $J_{\xi}(s)$ is continuous in $s$ over $(0,1]$ and continuous in $\xi$ over $(0,\infty)$.
\end{prop}
See Appendix~\ref{appendix:proof_of_propositions} for the proofs of these propositions.

\subsection{A Useful Normalization}\label{subsection:a_useful_normalization}

Note that, $\beta L$ 
is exponentially distributed with mean $1$. Defining 
$\Lambda:=\frac{\rho}{\beta×}$ and $\tilde{z}_k:=\beta y_k$, $k=1,2,\cdots,(N+1)$, 
we can rewrite (\ref{eqn:unconstrained_mdp}) as follows:
\begin{eqnarray}
 \inf_{\pi\in \Pi } \mathbb{E}_{\pi} \left(e^{\Lambda \tilde{z}_1}+
\sum_{k=2}^{N+1}\frac{e^{\Lambda \tilde{z}_k}-e^{\Lambda  \tilde{z}_{k-1}}}{1+e^{\Lambda \tilde{z}_1}+
\cdots+e^{\Lambda \tilde{z}_{k-1}}×}+\xi N \right)\label{eqn:unconstrained_mdp_beta_one}
\end{eqnarray}
Note that, $\Lambda$ plays the same role as $\lambda$ played in the known $L$ case 
(see Section~\ref{subsection:optimal_placement_single_relay_sum_power}). 
Since $\frac{1}{\beta×}$ is the mean length of the line, $\Lambda$ can 
be considered as a measure of attenuation in the network. 
We can think of the new problem (\ref{eqn:unconstrained_mdp_beta_one}) in the same way as (\ref{eqn:unconstrained_mdp}), but with 
the length of the line being exponentially distributed with mean $1$ ($\beta'=1$) and the path-loss exponent being changed to 
$\rho'=\Lambda=\frac{\rho}{\beta×}$. The relay locations are also normalized ($\tilde{z}_k=\beta y_k$). One can solve the new problem 
(\ref{eqn:unconstrained_mdp_beta_one}) and obtain the optimal policy. Then the solution to (\ref{eqn:unconstrained_mdp}) can be 
obtained by multiplying each control distance (from the optimal policy of (\ref{eqn:unconstrained_mdp_beta_one})) 
with the constant $\frac{1}{\beta×}$. Hence, it suffices to work with $\beta=1$.

\vspace{-2mm}
\section{A Numerical Study of As-You-Go Deployment}\label{section:numerical_work_information_theoretic_model}
\vspace{-0mm}

\begin{figure}[t!]
\centering
\includegraphics[height=3cm, width=8cm]{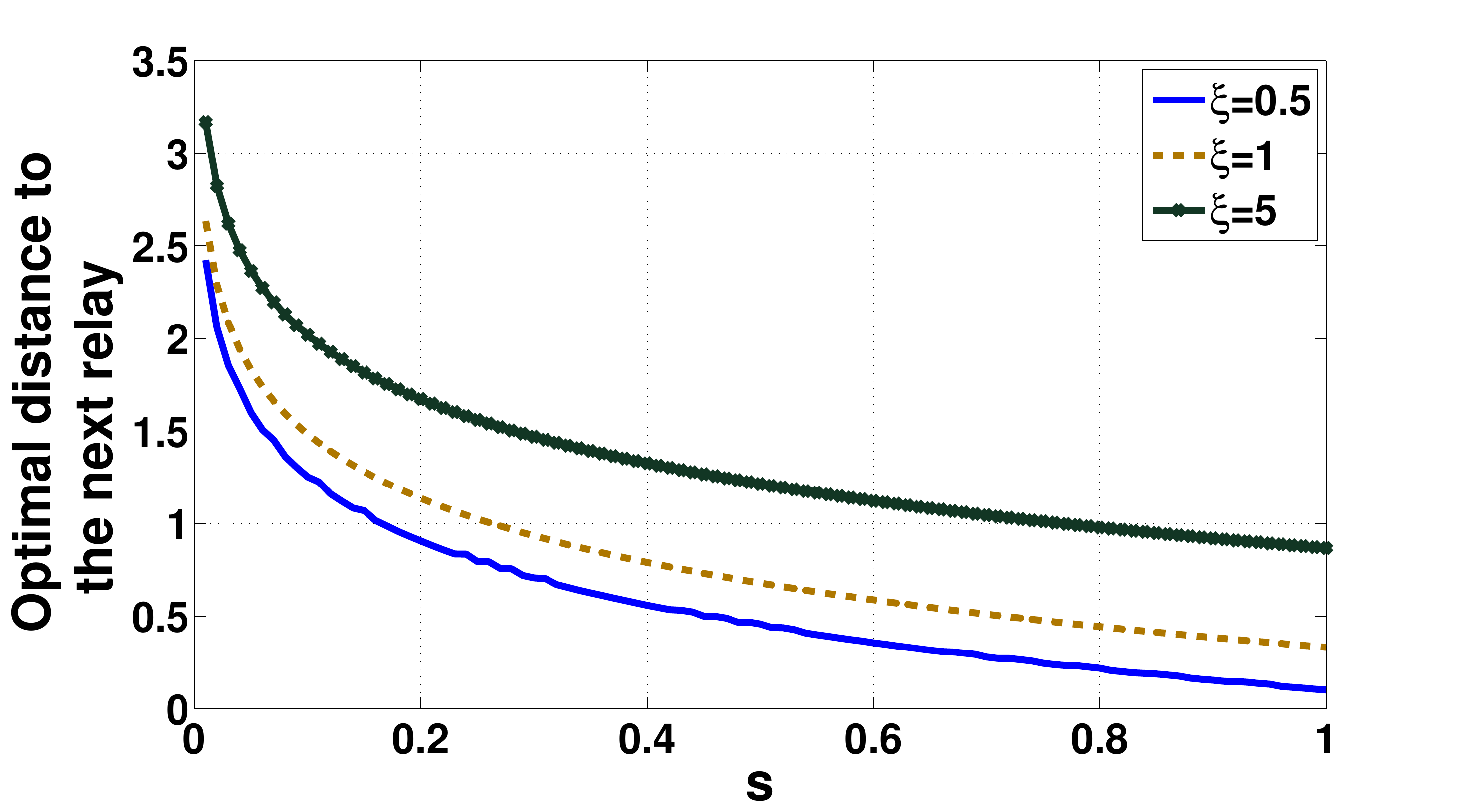}
\vspace{-0.1cm}
\caption{$\beta=1$, $\Lambda:=\frac{\rho}{\beta}=2$; $a^*$ vs. $s$.}
\label{fig:control_vs_state}
\vspace{-3mm}
\end{figure}

\begin{figure}[t!]
\centering
\includegraphics[height=3cm, width=8cm]{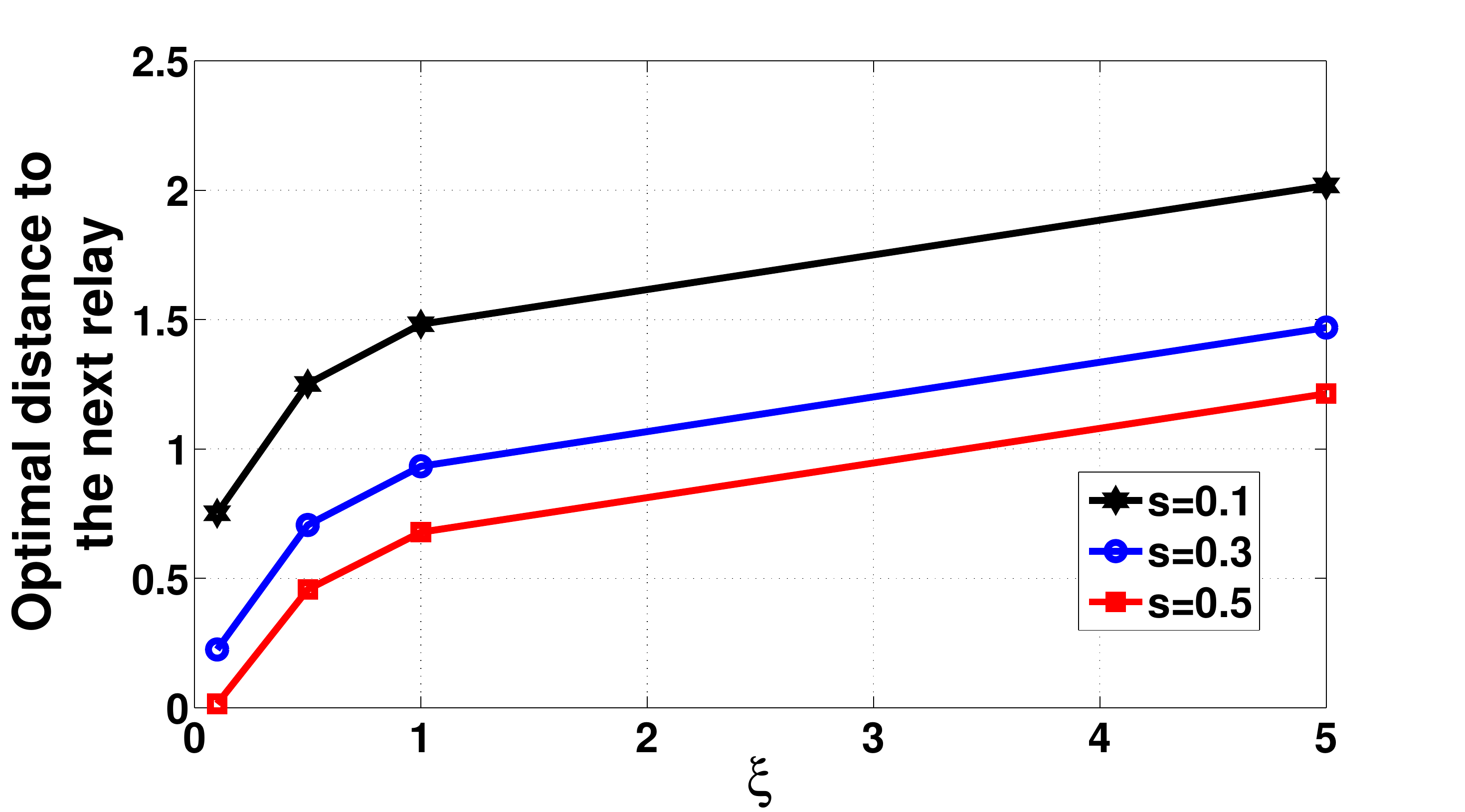}
\vspace{-0.1cm}
\caption{$\beta=1$, $\Lambda:=\frac{\rho}{\beta×}=2$; $a^*$ vs. $\xi$.}
\label{fig:control_vs_xi}
\vspace{-4mm}
\end{figure}

\begin{figure}[t!]
\centering
\includegraphics[height=3cm, width=8cm]{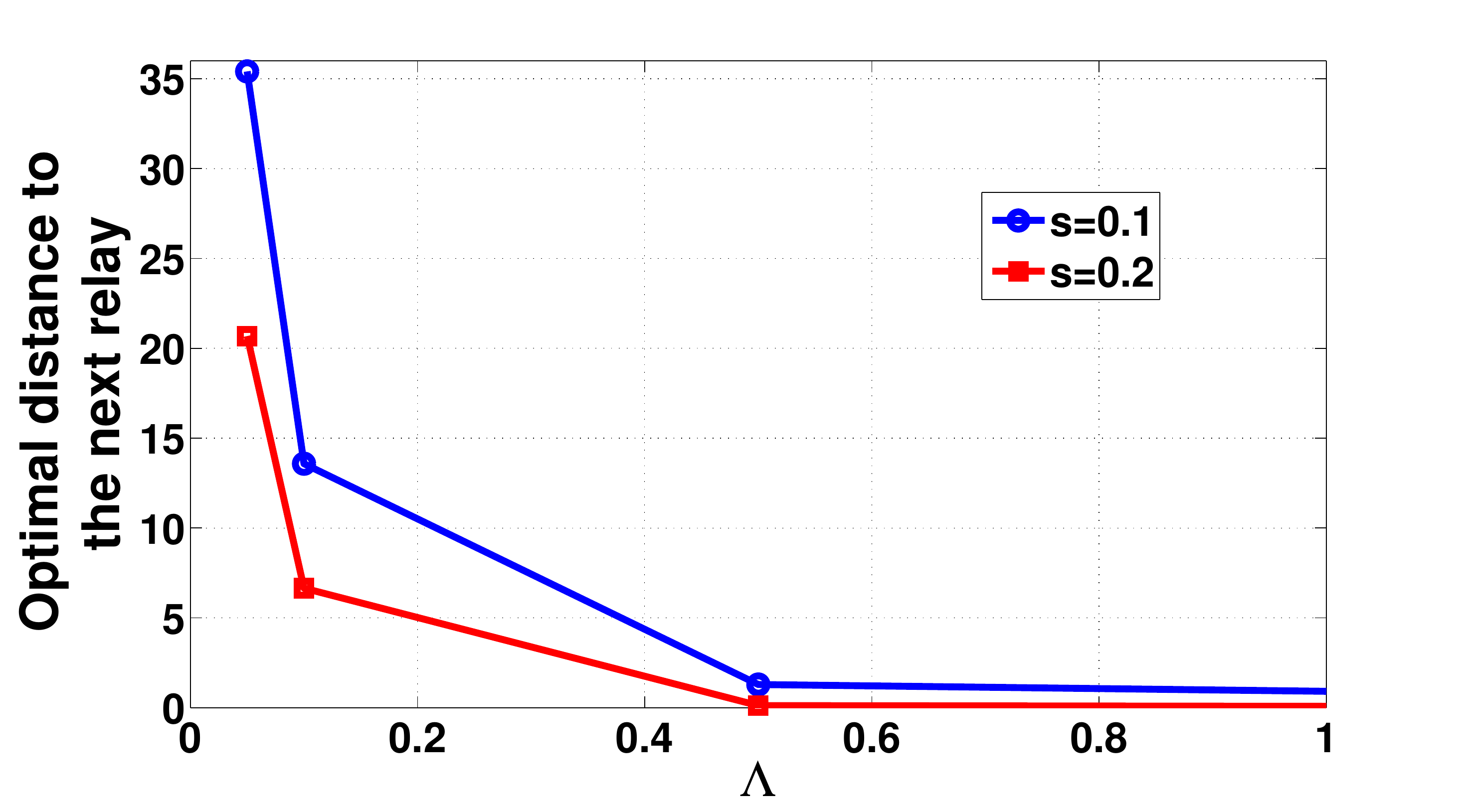}
\vspace{-0.1cm}
\caption{$\beta=1$, $\xi=0.01$; $a^*$ vs. $\Lambda$.}
\label{fig:control_vs_rho}
\vspace{-4mm}
\end{figure}

Let us recall that the state of the system after placing the $k$-th relay is given by 
$s_k=\frac{e^{\Lambda \tilde{z}_k}}{\sum_{i=0}^{k}e^{\Lambda \tilde{z}_i}}$. The action is the normalized 
distance of the next relay to be placed from the 
current location. The single stage cost function for our total cost minimization problem is given by (\ref{eqn:single-stage-cost}). 

In our numerical work, we discretized the state space $(0,1]$ into $100$ steps as  
$\{0.01,0.02,\cdots,0.99,1\}$, and discretized the action space into steps of size $0.001$, i.e., the action space becomes 
$\{0,0.001,0.002,\cdots\}$.

\vspace{-2mm}
\subsection{Structure of the Optimal Policy}\label{subsection:policy_structure_numerical}
We performed numerical experiments to study the structure of the optimal policy obtained through value 
iteration for $\beta=1$ and some values of 
 $\Lambda:=\frac{\rho}{\beta}$. The value iteration in these experiments converged and we obtained a stationary optimal policy, 
though Theorem~\ref{theorem:convergence_of_value_iteration} does not guarantee the uniqueness of the stationary optimal policy.

Figure~\ref{fig:control_vs_state}   
shows that the {\em normalized} optimal placement  distance  $a^*$ 
is decreasing with the state $s \in (0,1]$. This can be understood as follows. The state $s$ (at a placement point) is small 
only if a sufficiently large number of 
relays have already been placed.\footnote{$\frac{e^{\Lambda \tilde{z}_k}}{\sum_{i=0}^{k}e^{\Lambda \tilde{z}_i}} 
\geq \frac{e^{\Lambda \tilde{z}_k}}{(k+1)e^{\Lambda \tilde{z}_k}}=\frac{1}{k+1}$; 
hence, if $s_k$ is 
small, $k$ must be large enough.} Hence, if several relays have already been placed and $\sum_{i=0}^{k}e^{\Lambda \tilde{z}_i}$ is 
sufficiently large compared to $e^{\Lambda \tilde{z}_k}$ (i.e., $s_k$ is small), the $(k+1)$-st relay will 
be able to receive sufficient amount of power 
from the previous nodes, and hence does not need to be placed close to the $k$-th relay. 
A value of $s_k$ close to $1$ indicates that there is a 
large gap between relay $k$ and relay $k-1$, the power 
received at the next relay from the previous relays is small and hence the next relay must be placed closer to the previous one.

On the other hand, $a^*$ is increasing with $\xi$ (see Figure~\ref{fig:control_vs_xi}). 
Recall that $\xi$ is the price 
of placing a relay. This figure confirms the intuition that  if the relay price is high, then the relays should be placed less frequently.

\begin{table}[t!]

\caption{Sequential placement on a line of length $10$ for various   $\Lambda$,
  using the corresponding optimal policies for $\xi=0.001$.}
\label{table:effect_of_rho_on_placement_1}

\footnotesize
\centering
\begin{tabular}{|c |c |c |}
\hline
$\mathbf{\Lambda}$  & Normalised Optimal distances of the nodes & No. of \\ 
                    &  from the source              &    relays              \\ \hline
$0.01$ & 0,        0,   8.4180,   10.0000   &  3  \\ \hline 
$0.1$ &  0,         0,         0,         0,         0,         0,         0,    0.2950,    0.5950,    0.9810,    1.3670,  &   33    \\  
       & 1.7530,    2.1390,  $\cdots,$   9.0870,    9.4730,     9.8590,   10.0000  &  \\ \hline
 $5$ &  0,  0,  0,  0,  0,  0,  0,  0,  0,  0,  0.0020, 0.0080,  0.0140,  &  1677  \\ 
 &  0.0200, $\cdots$, 9.9860, 9.9920,  9.9980, 10.0000    &  \\ \hline 
\end{tabular}
\normalsize
\vspace{-1mm}
\end{table}

\begin{table}[t!]

\caption{Evolution of state in the process of sequential placement on a line of length $10$ for various values of $\Lambda$,
  using the corresponding optimal policies for $\xi=0.001$.}
\label{table:state-evolution_1}

\footnotesize
\centering
\begin{tabular}{|c |c |}
\hline
$\mathbf{\Lambda}$  & Evolution of state in the process of sequential placement \\ \hline
$0.01$ & 1,    0.5,    0.34,    0.27 \\ \hline 
$0.1$ &    1,    0.5,    0.34,    0.26, 0.21,  0.18, 0.16, 0.14, 0.13, 0.12,  0.12,   $\cdots$      \\  \hline
$5$ & 1,    0.5,    0.34,    0.26, 0.21,  0.18, 0.16, 0.14, 0.13,  0.12,    \\
     &0.11, 0.1,   0.1, 0.1, $\cdots$    \\ \hline
\end{tabular}
\normalsize
\vspace{-1mm}
\end{table}

\begin{table}[t!]

\caption{Sequential placement on a line of length $10$ for various   $\Lambda$,
  using the corresponding optimal policies for $\xi=0.1$.}
\label{table:effect_of_rho_on_placement_3}

\footnotesize
\centering
\begin{tabular}{|c |c |c|}
\hline
$\mathbf{\Lambda}$  & Normalised Optimal distances of the nodes & No. of \\
               &   from the source &    relays  \\ \hline
$0.01$ &  10 &  0 \\ \hline 
$0.1$ &  5.3060,   10.0000   &    1  \\  \hline
$5$ & 0,  0.0050, 0.0510,  0.1220,  0.1930, &  143  \\
    &  0.2640,$\cdots$, 9.9910, 10.0000  &   \\ \hline 
$8$  &  0, 0.003, 0.019, 0.06, 0.101, $\cdots$, 9.982, 10   & 246 \\ \hline
$20$ &  0, 0.001, 0.003, 0.016, 0.031, 0.046, $\cdots$, 9.991, 10   & 669 \\   \hline                                     
\end{tabular}
\normalsize
\vspace{-2mm}
\end{table}

\begin{table}[t!]

\caption{Evolution of state in the process of sequential placement on a line of length $10$ for various values of $\Lambda$,
  using the corresponding optimal policies for $\xi=0.1$.}
\label{table:state-evolution_3}

\footnotesize
\centering
\begin{tabular}{|c |c |}
\hline
$\mathbf{\Lambda}$  & Evolution of state in the process of sequential placement \\ \hline
$0.01$ & 1\\ \hline 
$0.1$ & 1,    0.63\\  \hline
$5$ &  1, 0.5,  0.34, 0.3, 0.3, $\cdots$  \\  \hline
$8$ &  1, 0.5,  0.34, 0.28, 0.28, $\cdots$   \\  \hline                         
$20$ &  1, 0.5,  0.34, 0.27, 0.26, 0.26, $\cdots$  \\  \hline
\end{tabular}
\normalsize
\end{table}

Figure~\ref{fig:control_vs_rho} shows that $a^*$ is decreasing with $\Lambda$, 
for fixed values of $\xi$ and $s$. This happens because increased attenuation  will require 
frequent placement of the relays.

\vspace{-2mm}
\subsection{Relay Placement Patterns}\label{subsection:numerical_relay_plecement_pattern_as_you_go}

The policy that we use corresponds to a line having 
exponentially distributed length with mean $1$, but it is applied to the scenario where the actual realization of the 
(normalised) length (see Section~\ref{subsection:a_useful_normalization}) of the line is $10$. 

Tables~\ref{table:effect_of_rho_on_placement_1}, \ref{table:effect_of_rho_on_placement_3},         
and \ref{table:effect_of_rho_on_placement_4} 
illustrate some examples of as-you-go placement of relay nodes along a line of normalised length $10$, 
using  various values of 
$\Lambda$ and $\xi$. 
Tables~\ref{table:state-evolution_1}, \ref{table:state-evolution_3}, and \ref{table:state-evolution_4} 
illustrate the corresponding evolution of state as the relays are placed in the examples in 
Tables \ref{table:effect_of_rho_on_placement_1}, \ref{table:effect_of_rho_on_placement_3}        
and \ref{table:effect_of_rho_on_placement_4}. 
 If the 
line actually ends at some point before (normalised) distance $10$, the process would end there with the corresponding placement  
of relays (as can be obtained from Tables \ref{table:effect_of_rho_on_placement_1}, \ref{table:effect_of_rho_on_placement_3},         
and \ref{table:effect_of_rho_on_placement_4}) before the sink being placed at the end-point. 
Thus, for example, reading from Table~\ref{table:effect_of_rho_on_placement_3} for $\xi=0.1$ and $\Lambda=5$, 
if the actual normalised length of the line is $0.99$, then  
one relay will be placed at $0$ (the source itself), followed  by $15$ relays at normalised distances $0.005, 0.051, 0.122, 0.193,
0.264, \cdots, 0.974$ from the source, and finally the sink is placed at a normalised distance $0.99$, the end of the line. 

We observe that as $\Lambda$ increases, more relays need to be placed since 
the optimal control decreases with $\Lambda$ for each $s$ 
(see Figure~\ref{fig:control_vs_rho}). On the other hand, the number of relays 
decreases with increasing $\xi$ (the relay cost); this is in confirmation of the observations from 
Figure~\ref{fig:control_vs_xi}. 

Note that, initially 
one or more relays are placed at or near the source if $a^*(s=1)$ is $0$ or small. But, 
after some relays have been placed, the relays are placed 
equally spaced apart. We see that this happens because,  after a few relays have been placed, the state, $s$, 
does not change, hence, resulting in the relays being subsequently placed equally spaced apart. 
This phenomenon is evident in Table~\ref{table:state-evolution_1}, Table~\ref{table:state-evolution_3}, Table~\ref{table:state-evolution_4}, 
and Figure~\ref{fig:state_evolution}. 
The state $s$ will remain unchanged after a relay placement if 
$s=\lceil{\frac{se^{\Lambda a^*(s)}}{0.01(1+se^{\Lambda a^*(s)})×}}\rceil \times 0.01$, 
since we have discretized the state space. 
After some relays are placed, the state becomes equal to a fixed point $s'$ of the function 
$\lceil{\frac{se^{\Lambda a^*(s)}}{0.01(1+se^{\Lambda a^*(s)})×}}\rceil \times 0.01$. Note that 
the deployment starts from $s_{0}:=1$, but for any value of $s_{0}$ (even with $s_{0}$ smaller than $s'$), 
we numerically observe the same phenomenon. Hence, 
$s'$ is an absorbing state.

\begin{table}[t!]

\caption{Sequential placement on a line of length $10$ for  of $\xi$,
  using the corresponding optimal policies for $\Lambda=20$.}
\label{table:effect_of_rho_on_placement_4}

\footnotesize
\centering
\begin{tabular}{|c |c |c|}
\hline
$\xi$  & Normalised Optimal distances of the nodes & No. of \\
               &   from the source &    relays  \\ \hline
$0.2$ & 0 , 0.008, 0.03, 0.052, $\cdots$, 9.996, 10 &  456 \\ \hline 
$1$ &  0.022, 0.069, 0.116, $\cdots$, 9.986, 10 &  213   \\  \hline
$2$ & 0.042, 0.103, 0.163,  0.223, $\cdots$, 9.943, 10 &  166   \\  \hline 
$10$  & 0.099, 0.205, 0.311, $\cdots$, 9.957, 10   &  94 \\ \hline                        
\end{tabular}
\normalsize
\end{table}

\begin{table}[t!]

\caption{Evolution of state in the process of sequential placement on a line of length $10$ for various values of $\xi$,
  using the corresponding optimal policies for $\Lambda=20$.}
\label{table:state-evolution_4}

\footnotesize
\centering
\begin{tabular}{|c |c |}
\hline
$\xi$  & Evolution of state in the process of sequential placement \\ \hline
$0.2$ & 1, 0.5, 0.37, 0.37, $\cdots$ \\  \hline
$1$ & 1, 0.61, 0.61, $\cdots$ \\  \hline
$2$ &  1, 0.7, 0.71, 0.71, $\cdots$    \\  \hline                         
$10$ &  1, 0.88, 0.88, $\cdots$   \\  \hline
\end{tabular}
\normalsize
\end{table}

\vspace{-2mm}
\subsection{Numerical Examples for Practical Deployment}\label{subsection:numerical-example-practical-deployment} 
In order to provide a more concrete illustration we adopt a path loss parameter from 
\cite{franceschetti-etal04random-walk-model-wave-propagation}. 
Figure~$4$ of \cite{franceschetti-etal04random-walk-model-wave-propagation} shows that the attenuation in 
the received signal power in a dense 
urban environment is roughly $50$~dB when we move from $50$~m distance to $300$~m distance away from the transmitter. This yields 
a value of $\rho$ to be $0.04$ per meter for the exponential path-loss 
(see the discussion in Section~\ref{subsection:motivation-for-exponential-path-loss} on the motivation 
for choosing the exponential path-loss model in the light of the results from 
\cite{franceschetti-etal04random-walk-model-wave-propagation}). Then, $\frac{1}{\beta}=200$~m corresponds to $\Lambda=8$, and 
$\frac{1}{\beta}=500$~m corresponds to $\Lambda=20$. 
For  $\Lambda=20$, normalised relay locations and state evolution $\{s_k\}_{k \geq 1}$ are available in 
Tables~\ref{table:effect_of_rho_on_placement_3}-\ref{table:state-evolution_4}, and, for 
$\Lambda=8$, normalised relay locations and state evolution $\{s_k\}_{k \geq 1}$ are available in 
Tables~\ref{table:effect_of_rho_on_placement_3}-\ref{table:state-evolution_3}. Note that, under $\rho=0.04$ per meter and 
$\Lambda=20$, one unit normalised distance 
in the tables correspond to $500$~m distance in the  dense urban environment (due to 
the normalization as in Section~\ref{subsection:a_useful_normalization}). 

For the sake of illustration, let us consider the sample deployment for 
$\xi=10$, $\Lambda=20$ (Table~\ref{table:effect_of_rho_on_placement_4}). In this case, the first relay will 
be placed at a distance $0.099 \times 500=49.5$~m from the source, the second relay will be placed at a distance 
$0.205 \times 500=102.5$~m from the source, etc. Also, if we choose $\xi$ such that 
few relays will be placed on a typical line whose length is several hundreds of meters, then the relays will be placed 
almost uniformly on the line. But, for small $\xi$,  more relays will be placed and some of them will be 
clustered near the source (see the deployment for $\Lambda=8$ and $\xi=0.1$  in 
Table~\ref{table:effect_of_rho_on_placement_3}).

\begin{figure}[t!]
\centering
\includegraphics[height=3cm, width=8cm]{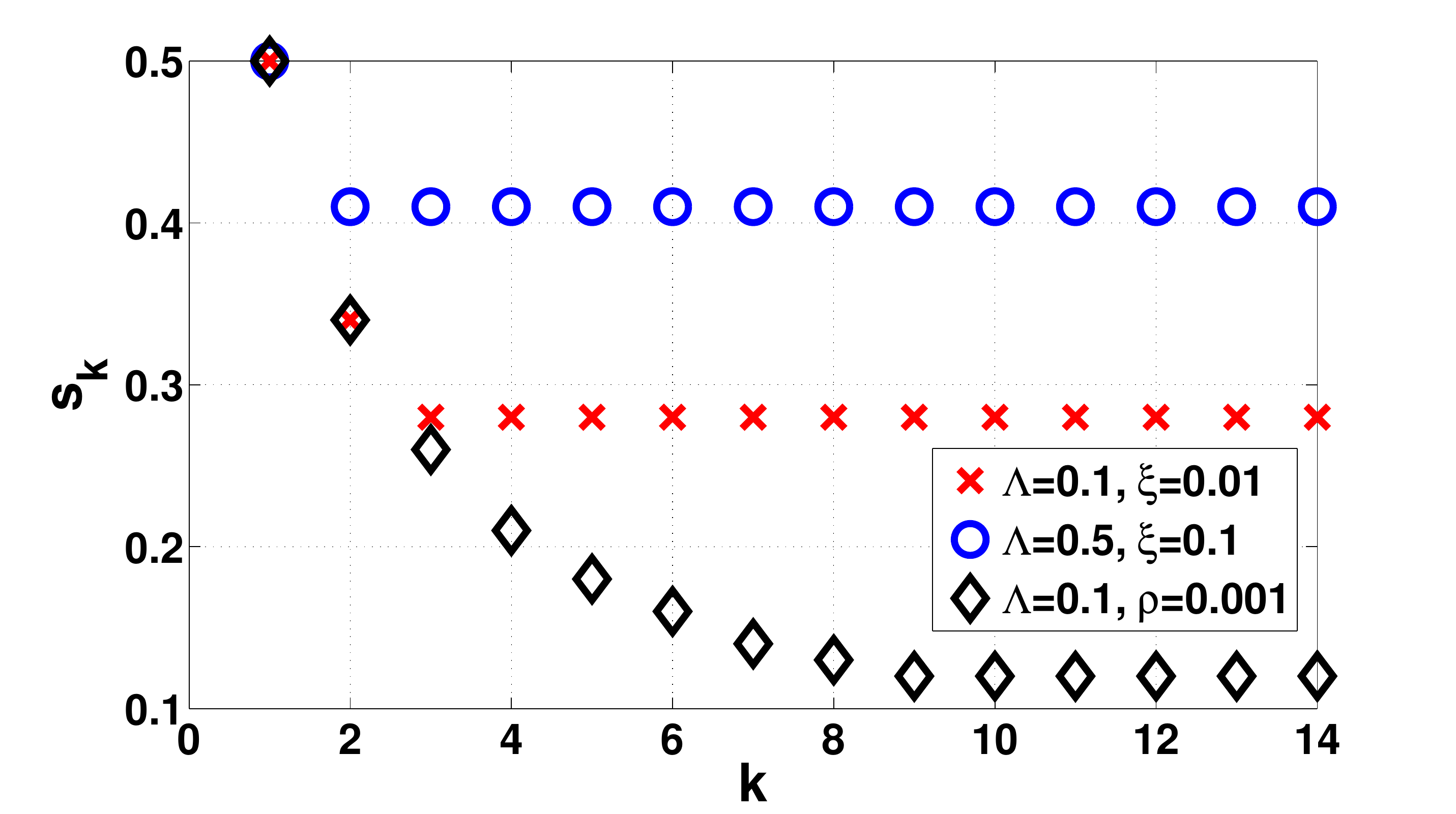}
\vspace{-0mm}
\caption{Evolution of the state $s_k$ with $k$: initial state $s_0=1$.}
\label{fig:state_evolution}
\vspace{-2mm}
\end{figure}

\begin{table}[t!]

\caption{Comparison of the performance (in terms of $H$; see text) 
of optimal sequential placement over a line of random length, with the optimal placement if the length was known. 
Results from 10000 samples of exponentially distributed line lengths.}

\label{table:comparison_optimal_mdp}

\footnotesize
\centering
\begin{tabular}{|c |c |c |c |c |c |}
\hline
 & & Average  & Mean& Number of  & Maximum  \\
$\xi$ &   $\Lambda$ & percentage  & number of & cases where & percentage \\ 
 & & difference & relays used& no relay & difference \\
 & & & &was used & \\ \hline

 0.001   & 0.01     & 0.0068       & 2.0002       & 0       & 0.7698         \\ \hline
 0.001   & 0.1      & 0.3996      &  9.4849      &  0      & 6.8947         \\ \hline
0.01 & 0.01         &  0      &    0    &   10000     &    0      \\ \hline
 0.01  & 0.1        &  0.3517      & 2.2723       &   0     &   4.6618        \\ \hline
0.01 & 0.5          &  1.5661      & 7.7572       &   0     &   4.7789       \\ \hline
0.1 & 0.01          &  0      &  0     &   10000     &   0         \\ \hline
0.1 & 0.1           &   0.1259     &   0.0056     &   9944     &    25.9098        \\ \hline
0.1   & 0.5         &   2.9869     &    1.8252    &    0    &    12.5907        \\ \hline
0.1   & 2           &  4.7023      &   7.1530     &    0    &    9.0211        \\ \hline
0.1 & 20            & 3.5472       & 27.9217      & 0         & 6.6223    \\ \hline
0.1 & 8             & 4.0097       & 21.0671       & 0       &  7.8264  \\ \hline
1   & 8             & 8.0286       & 7.8886       & 495       & 27.7362   \\ \hline
1   & 20            & 5.2158       & 11.2342      & 402       & 26.0026  \\ \hline
5  & 20            &  10.3341      & 7.1950       & 597       & 61.7460  \\ \hline
\end{tabular}
\normalsize
\vspace{-5mm}
\end{table}

\vspace{-2mm}

\subsection{Comparison with Optimal Offline Deployment}\label{subsubsection:numerical_performance_as_you_go}

Since there is no prior work in the literature with which we can make a fair comparison of our as-you-go 
deployment policy for the full-duplex wireless multi-relay network, we compare the performance of our policy with optimal offline deployment.   
Thus, the  numerical experiments reported in Table~\ref{table:comparison_optimal_mdp}  
are a result of asking the following question: how does the cost of as-you-go deployment over a 
line of exponentially distributed length compare with the cost of placing the same number of relays optimally 
over the line, once the length of the line has been revealed? 

For several combinations of $\Lambda$ and $\xi$, we generated $10000$ 
random numbers independently from an exponential distribution with parameter $\beta=1$. Each of these numbers was considered as 
a possible realization of the length of the line. Then we computed the placement of the relay nodes 
for each realization by optimal sequential placement policy, which gave us 
$H=\frac{1}{g_{0,1}×}+\sum_{k=2}^{N+1}\frac{(g_{0,k-1}-g_{0,k})}{g_{0,k}g_{0,k-1}\sum_{l=0}^{k-1}\frac{1}{g_{0,l}×}×}$, 
a quantity that we use to evaluate the quality of the relay placement. The significance of $H$ can be 
recalled from (\ref{eqn:capacity_multirelay}) where we found that the rate $C(1+\frac{P_T/\sigma^2}{H})$ 
can be achieved if total 
power $P_T$ is available to distribute among the source and the relays; i.e., $H$ can be interpreted as the 
net effective attenuation after power has been allocated optimally over the nodes. 
Also, for each realization, we computed $H$ 
for optimal relay placement, assuming that the length of the line is known before deployment and that 
the number of relays available 
is the {\em same} as the number of relays used by the corresponding sequential placement policy. 
For a given combination of $\Lambda$ 
and $\xi$, for the $k$-th realization of the length of the line, let us denote the two $H$ values by $H_{\mathsf{sequential}}^{(k)}$ 
and 
$H_{\mathsf{optimal}}^{(k)}$. Then the percentage  {\em difference} for the $k$-th realization is:

\begin{equation}
 e_k:= \frac{|H_{\mathsf{optimal}}^{(k)}-H_{\mathsf{sequential}}^{(k)}|}{H_{\mathsf{optimal}}^{(k)}} \times 100 
\label{eqn:error_or_difference_expression}
\end{equation}

The average percentage difference in Table~\ref{table:comparison_optimal_mdp} is  
the quantity $\frac{\sum_{k=1}^{10000}e_k}{10000×}$. 
The maximum percentage difference is the quantity $\max_{k \in \{1,2,\cdots,10000\}}e_k$. 

\textbf{Discussion of Table~\ref{table:comparison_optimal_mdp}:}

\begin{enumerate}[label=(\roman{*})]

\item For small enough 
$\xi$, some relays will be placed at the source itself. For example, for $\Lambda=0.01$ and $\xi=0.001$, we will place two relays 
at the source (Table~\ref{table:effect_of_rho_on_placement_1}). 
After placing the first relay, the next state will become $s=0.5$, and 
$a^*(s=0.5)=0$. The state after placing the second relay becomes $s=0.34$, for which $a^*(s=0.34)=8.41$ 
(see the placement in Table~\ref{table:effect_of_rho_on_placement_1}). Now, 
the line having an exponentially distributed length with mean $1$ will 
end before $a^*(s=0.34)=8.41$ distance with high probability, 
and the probability of placing the third relay will be very small. As a result, the mean number of relays will be $2.0002$. 
In case only $2$ relays are placed by the sequential deployment policy and we seek to place 
$2$ relays optimally 
for the same length of the line (with the length known), the optimal locations for both relays are close to 
the source location 
if the length of the line is small (i.e., if the attenuation $\lambda$ is small, recall the definition of $\lambda$ from 
Section~\ref{subsection:optimal_placement_single_relay_sum_power}). If the line is long (which has a very small probability), 
the optimal placement 
will be significantly different from the sequential placement. Altogether, the difference 
(from (\ref{eqn:error_or_difference_expression})) will be small.

\item For $\Lambda=0.01$ and $\xi=0.1$, $a^*(1)$ is so large that with high probability 
the line will end in a distance less than 
$a^*(1)$ and no relay will be placed.

\item From (\ref{eqn:capacity_multirelay}) 
we know that for a given placement of relays on a line of given length 
$L$, the optimal power allocation yields an achievable rate $\log_2(1+\frac{P_T/\sigma^2}{H})$. 
At the end of as-you-go 
deployment the power is allocated optimally among the nodes deployed, and a rate 
$\log_2(1+\frac{P_T/\sigma^2}{H_{sequential}})$ can be achieved. 
If the same number of relays are optimally placed over the same line, 
with the same total power, then the inner bound is given by $\log_2(1+\frac{P_T/\sigma^2}{H_{optimal}})$. 
We seek to compare these two  rates numerically.

The maximum fractional difference 
in Table~\ref{table:comparison_optimal_mdp} is less than $\frac{2}{3}$, and substantially 
smaller than $\frac{2}{3}$ in most cases. 
Since, in (\ref{eqn:error_or_difference_expression}), $H_{sequential}^{(k)}$ is always greater than 
$H_{optimal}^{(k)}$, we have $H_{sequential}^{(k)} \leq \frac{5}{3} H_{optimal}^{(k)}$ for all $k \geq 1$ 
(i.e., for all realizations of $L$ in the simulation). Now,  
by the monotonicity of $\log_2(\cdot)$:

\footnotesize
\begin{eqnarray*}
&& \frac{1}{2} \log_2 \bigg(1+\frac{P_T/\sigma^2}{H_{optimal}^{(k)}} \bigg)- \frac{1}{2} \log_2 \bigg(1+\frac{P_T/\sigma^2}{H_{sequential}^{(k)}} \bigg) \nonumber\\
& \leq & \frac{1}{2} \log_2 \bigg(1+\frac{P_T/\sigma^2}{H_{optimal}^{(k)}} \bigg)- \frac{1}{2} \log_2 \bigg(1+\frac{P_T/\sigma^2}{\frac{5}{3} H_{optimal}^{(k)}} \bigg) \nonumber\\
\end{eqnarray*}
\normalsize

Since $\log_2(\cdot)$ is a concave function, for any $x>y>0$, we have 
$\log_2 (1+x)-\log_2(1+y) \leq \log_2 (x) -\log_2 (y)$. Using this inequality, we can upper bound the difference in achievable 
rate from the previous equation by:

\footnotesize
\begin{eqnarray*}
 \frac{1}{2} \log_2 \bigg(\frac{P_T/\sigma^2}{H_{optimal}^{(k)}} \bigg)- \frac{1}{2} \log_2 \bigg(\frac{P_T/\sigma^2}{\frac{5}{3} H_{optimal}^{(k)}} \bigg) 
 =  0.3685
\end{eqnarray*}
\normalsize

This calculation implies that, for the large number of cases reported in 
Table~\ref{table:comparison_optimal_mdp}, by using the approximation in  
(\ref{eqn:constrained_mdp}) and by using the corresponding optimal policy for 
as-you-go deployment, we lose at most 
$0.3685$ bits per channel use, compared 
to the case when the realization of the exponentially distributed source to sink distance  
is known apriori and when we use the same number of relays as used in the as-you-go 
deployment case. Note that, the statement of this claim holds {\em with high probability} 
since the maximum difference is taken over $10000$ sample deployments. Hence, it is reasonable to solve (\ref{eqn:constrained_mdp}) instead of 
(\ref{eqn:constrained_mdp_actual}) which is intractable. 

\end{enumerate}

\vspace{-2mm}
\section{Discussion}\label{section:additional_discussion}
%
\vspace{-1mm}
\subsection{Exponential Path-Loss Model}\label{subsection:motivation-for-exponential-path-loss}
Exponential path-loss model  has been used before in the context of relay placement 
(see \cite{firouzabadi-martins08optimal-node-placement}, 
\cite{appuswamy-etal10relay-placement-deterministic-line}) and in the 
context of cellular networks (see \cite[Section~$2.3$]{altman-etal11greec-cellular}).  
Analytical and experimental support for the exponential path-loss model have been provided by 
Franceschetti et al. (\cite{franceschetti-etal04random-walk-model-wave-propagation}). 
Franceschetti et al.  used a random scattering model 
(applicable to an urban environment, or a forest environment) to show that the path-loss in such an environment 
is the product of an exponential function and a power function
of the distance (see \cite[Equation~$(14)$]{franceschetti-etal04random-walk-model-wave-propagation}). Figure~$4$ of their paper, which is obtained from 
measurements made in an urban environment, shows
that path-loss (in dB) varies linearly with distance beyond a distance of 
$40-50$~meters, which implies exponential path-loss for longer distance. 
These distances are practical for urban scenarios where the network is deployed over several hundreds of meters or
several kilometers. 

Exponential path-loss was also proposed by 
Marano and Franceschetti  for urban environment, and 
validated by theory and experiment (see \cite[Figure~$10$]{marano-franceschetti05ray-propagation-random-lattice}).\footnote{Marano and Franceschetti 
(\cite{marano-franceschetti05ray-propagation-random-lattice}) modeled a city as a random lattice, 
and the distance from the transmitter to the receiver is measured along the edges of the lattice instead of the 
Euclidean distance. Hence, this result renders the analysis in our paper valid even for
deployment along the streets of a city with turns; deployment algorithm in that case will only consider the distances along
the streets and not on the actual Euclidean distances.}

\vspace{-2mm}
\subsection{Incorporating Shadowing and Fading}\label{subsection:shadowing-fading}
\vspace{-1mm}
Shadowing (which is typically viewed as being static once a link is deployed) and time varying fading, 
can be incorporated in our setting by providing a fade-margin in the power at each transmitter. 
Thus, when expressed in dBm, the actual transmit power for any transmitter-receiver pair is the fade-margin  
plus the power used in the 
information theoretic capacity formulas; this fade-margin 
does not depend on the distance between the transmitter-receiver pair. 
Note that, this approach, though very conservative in nature, 
can remove the complexity in analysis arising 
out of fading in the network. Also note that, if the actual power gain between two nodes $r$ distance apart 
is $c_0e^{-\rho r}$ with $c_0 > 0$, then $c_0$ can be absorbed in the fade margin.

\vspace{-3mm}
\subsection{Full-Duplex Decode-and-Forward Relaying}\label{subsection:motivation-for-full-duplex-decode-forward}

Full-duplex radios might become a reality soon; 
see \cite{khandani13two-way-full-duplex-wireless}, 
\cite{khandani10spatial-multiplexing-two-way-channel}, \cite{choi-etal10single-channel-full-duplex}, 
\cite{jain-etal11real-time-full-duplex} for recent efforts to realize them. 
Decode-and-forward relaying requires symbol-level synchronous operation
across all nodes in the network. The requirement of globally coherent transmission and reception seems 
to be restrictive at the moment, but this problem
will be solved with the advent of better clocks (with less drift) 
and efficient clock synchronization algorithms. Any  
research on impromptu deployment assuming imperfect synchronization, or  
half-duplex communication, or no interference cancellation, can use this paper 
as a benchmark for performance analysis.

\vspace{-3mm}
\subsection{Insights on Power-Law Path-Loss}
\label{subsection:insights_for_power_law_from_exponential}
\vspace{-1mm}

In \cite{chattopadhyay-etal12optimal-capacity-relay-placement-line}, 
we studied the problem of single-relay placement under a per-node power constraint at the source
and the relay, for both exponential and power-law path-loss models. The variation of optimal relay location, as the amount of
attenuation in the network varies, follow slightly different (but mostly similar) trends (see Figures~$2$ and $3$ of 
\cite{chattopadhyay-etal12optimal-capacity-relay-placement-line}) because of
the fact that power-law model allows unbounded power gain (unlike the exponential model) when the distance $r$ tends to $0$ 
($\lim_{r \rightarrow 0} r^{-\eta}=\infty$). 
The findings are even more similar when we bound the power gain from above by some constant value
in case of the power-law model (power gain is $\min\{r^{-\eta}, b^{-\eta} \}$ for some $b > 0$); 
see the similarity between Figures~$2$ and $4$
in \cite{chattopadhyay-etal12optimal-capacity-relay-placement-line}. 
The results on the fixed node power
case provide the insight that when the power gain is $r^{-\eta}$ or $\min\{r^{-\eta}, b^{-\eta} \}$;  
under the sum power constraint, the variation
of the relay locations as a function of attenuation will follow a pattern similar to that 
in case of exponential path-loss.

\vspace{-2mm}
\section{Conclusion}
\label{conclusion}
\vspace{-1mm}

Motivated by the problem of as-you-go deployment of wireless relay networks, we 
first studied the problem of placing relay nodes along a line, 
in order to connect a sink at the end of the line to a source 
at the start of the line, so as to  maximize the end-to-end achievable data 
rate. For the multi-relay channel with exponential 
path-loss and sum power constraint, we derived an expression 
for the achievable rate in terms of the power gains among all possible node pairs, 
and formulated an optimization problem in order to 
maximize the end-to-end data rate. Numerical work for the fixed source-sink distance suggests that at low attenuation
the relays are mostly clustered close to the source in order to be able to cooperate among themselves, 
whereas at high attenuation
they are uniformly placed and work as repeaters. Next, the deploy-as-you-go 
sequential placement problem was addressed; 
a sequential relay placement problem along a line having unknown random length 
was formulated as an MDP, the value function was characterized 
analytically, and the policy structure was investigated numerically. We found numerically that at the initial stage 
of the deployment process the inter-relay distances are smaller, and, as deployment progresses, the 
inter-relay distances increase gradually, and finally the relays start being  
placed at regular intervals.

Our results are based on information theoretic achievable rate results. 
In order to utilize currently commercially available wireless devices, 
we have also been exploring non-information theoretic, packet forwarding models 
for optimal relay placement, with the aim of obtaining placement algorithms that can be easily reduced to practice 
(see \cite{chattopadhyay-etal13measurement-based-impromptu-placement_wiopt} for reference). 
The study of as-you-go deployment under the information theoretic model and under the packet forwarding model 
provides two complementary approaches for two different conditions in the physical layer and 
the MAC layer, and provides a more comprehensive development 
of the problem.

\vspace{-2mm}
\bibliographystyle{IEEEtran}
\bibliography{IEEEabrv,arpan-techreport}

\vspace{-13mm}

\begin{IEEEbiography}[{\includegraphics[width=1in,height=1in,clip,keepaspectratio]{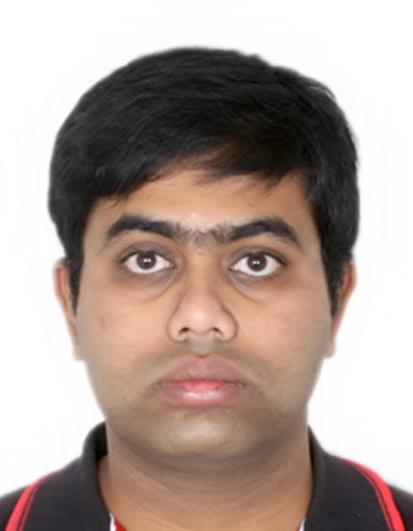}}]{Arpan 
Chattopadhyay} obtained his B.E. in Electronics and Telecommunication Engineering from Jadavpur University, 
Kolkata, India in the year 2008, and M.E. and Ph.D in Telecommunication Engineering from Indian Institute of Science, 
Bangalore, India in the year 2010 and 2015, respectively. He is currently working in INRIA, Paris as a postdoctoral researcher. 
His research interests include  networks and machine learning.
    \end{IEEEbiography}

    \vspace{-15mm}

\begin{IEEEbiography}[{\includegraphics[width=1in,height=1in,clip,keepaspectratio]{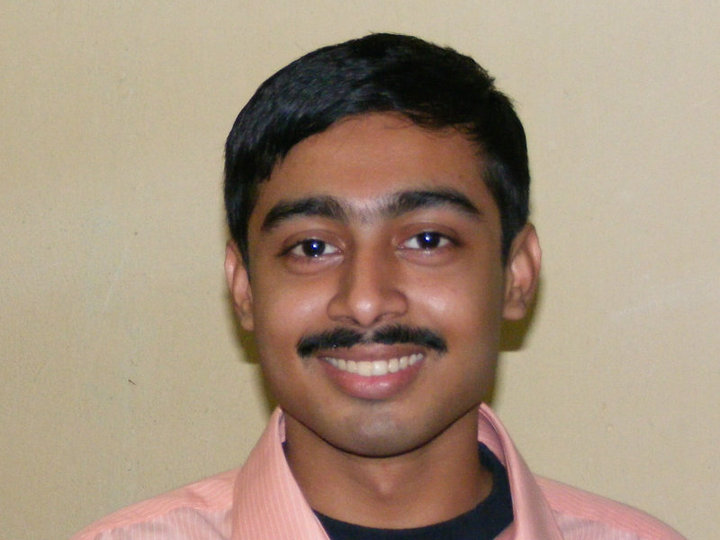}}]
{Abhishek Sinha} is currently a graduate student in the Laboratory for Information and Decision 
Systems (LIDS), at Massachusetts Institute of Technology, Cambridge, MA. Prior to joining MIT, 
he completed his Master's studies in Telecommunication Engineering at the Indian Institute of Science, 
Bangalore, in the year 2012. His areas of interests include stochastic processes, information theory and network control. 
        
    \end{IEEEbiography}

    \vspace{-15mm}

   \begin{IEEEbiography}[{\includegraphics[width=1in,height=1in,clip,keepaspectratio]{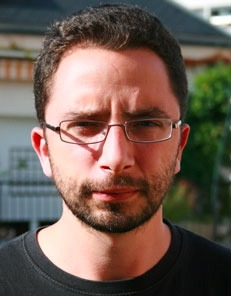}}]
   {Marceau Coupechoux}  is an Associate Professor at Telecom  ParisTech since 2005. He obtained his master from 
   Telecom ParisTech in 1999 and from University of Stuttgart, Germany in  2000, and his Ph.D. from Institut Eurecom, 
   Sophia-Antipolis, France, in 2004. From 2000 to 2005, he  was with Alcatel-Lucent (Bell Labs former 
   Research \& Innovation and then in the Network Design department). In the Computer and Network Science  
   department of Telecom ParisTech, he is working on cellular networks,  wireless networks, ad hoc networks, 
   cognitive networks, focusing mainly on  layer 2 protocols, scheduling and resource management. 
   From August 2011 to August 2012 he was a visiting  scientist at IISc Bangalore.
   \end{IEEEbiography}

   \vspace{-50mm}
   
    \begin{IEEEbiography}[{\includegraphics[width=1in,height=1in,clip,keepaspectratio]{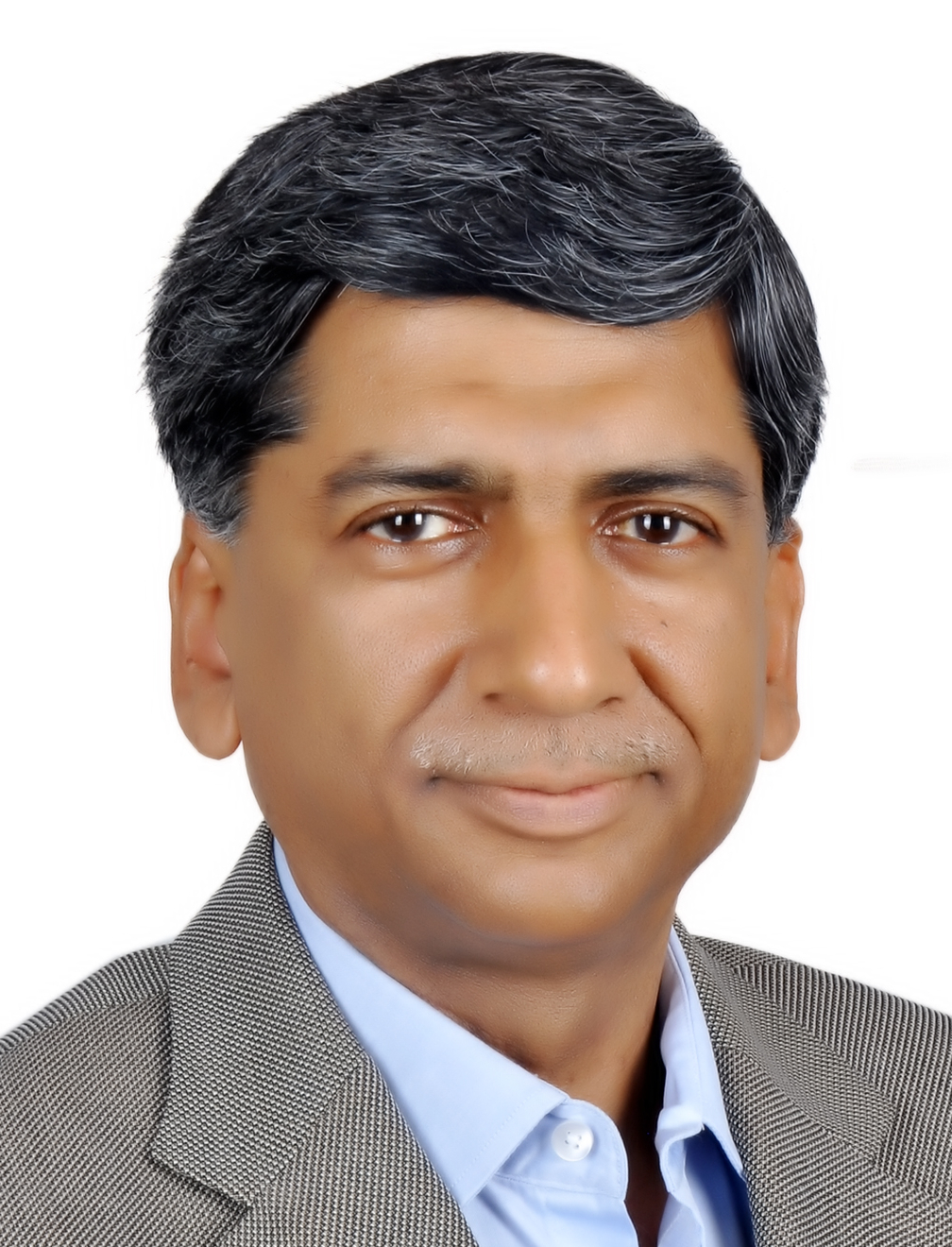}}]
    {Anurag Kumar} obtained his B.Tech. degree from the Indian Institute of 
 Technology at Kanpur, and the PhD degree from Cornell University, 
 both in Electrical Engineering. He was then with Bell Laboratories, 
 Holmdel, N.J., for over 6 years.  Since 1988 he has been on the 
 faculty of the Indian Institute of Science (IISc), Bangalore, in the 
 Department of Electrical Communication Engineering.  He is currently 
 also the Director of the Institute.  From 1988 to 2003 he was the 
 Coordinator at IISc of the Education and Research Network Project 
 (ERNET), India's first wide-area packet switching network.  His area 
 of research is communication networking, specifically, modeling, 
 analysis, control and optimisation problems arising in communication 
 networks and distributed systems. Recently his research has focused 
 primarily on wireless networking.  He is a Fellow of the IEEE, of the 
 Indian National Science Academy (INSA), of the Indian Academy of 
 Science (IASc), of the Indian National Academy of Engineering (INAE), 
 and of The World Academy of Sciences (TWAS). He is a recepient of the 
 Indian Institute of Science Alumni Award for Engineering Research for 
 2008.
 \end{IEEEbiography}

%

\renewcommand{\thesubsection}{\Alph{subsection}}

\appendices

\section{A Brief Description of the Coding Scheme of \cite{xie-kumar04network-information-theory-scaling-law}}
\label{section:coding_scheme_description}

Transmissions take place via block codes of $T$ symbols each. The transmission blocks at the source and the $N$ 
relays are synchronized. The coding and decoding scheme is such that a message generated at the source at the beginning 
of block $b, b \geq 1,$ is decoded by the sink at the end of block $b + N$, i.e., $N+1$ block durations after 
the message was generated (with probability tending to 1, as $T \to \infty)$. Thus, at the end of $B$ blocks, $B \geq N+1$, 
the sink is able to decode $B-N$ messages. It follows, by taking $B \to \infty$, that, if the code rate 
is $R$ bits per symbol, then an information rate of $R$ bits per symbol can be achieved from the source to the 
sink.

As mentioned earlier, we index the source by $0$, the relays 
by $k, 1 \leq k \leq N$, and the sink by $N+1$. There are  $(N+1)^2$ independent Gaussian random codebooks, each containing $2^{TR}$ codes, 
each code being of length $T$; these codebooks are available to all nodes. At the beginning of block $b$, the source 
generates a new message $w_b$, and, at this stage, we assume that 
each node $k, 1 \leq k \leq N+1,$ has a reliable estimate of all the 
messages $w_{b-j}, j \geq k$.  In block $b$, the source uses a new codebook to encode $w_b$. In addition, 
relay $k, 1 \leq k \leq N,$ and {\em all} of its previous transmitters (indexed $0 \leq j \leq k-1$), use {\em another} 
codebook to encode $w_{b-k}$ (or their estimate of it). Thus, if the relays $1,2,\cdots,k$ have a perfect estimate of $w_{b-k}$ 
at the beginning of block $b$, they will transmit the same codeword for $w_{b-k}$. Therefore, in block $b$, 
the source and relays $1, 2, \cdots, k$ \emph{coherently transmit} the codeword for $w_{b-k}$. 
In this manner, in block $b$, transmitter $k, 0 \leq k \leq N,$ 
generates $N+1 - k$ codewords, corresponding to $w_{b-k}, w_{b-k-1}, \cdots, w_{b-N}$, which are transmitted with powers 
$P_{k,k+1}, P_{k,k+2}, \cdots, P_{k,N+1}$. In block $b$, node $k, 1 \leq k \leq N+1,$ receives a superposition of transmissions 
from all other nodes. Assuming that node $k$ knows all the powers, and all the channel gains, and recalling that it has a reliable 
estimate of  all the messages $w_{b-j}, j \geq k$, it can subtract the interference from transmitters $k+1, k+2, \cdots, N$. 
At the end of block $b$, after subtracting the signals it knows, node $k$ is left with the $k$ received signals from 
nodes $0, 1, \cdots, (k-1)$ (received in blocks $b, b-1, \cdots, b-k+1$), which all carry an encoding of the message $w_{b-k+1}$. These $k$ signals are then jointly used 
to decode  $w_{b-k+1},$ using joint typicality decoding. The codebooks are cycled through in a manner so that in any block all 
nodes encoding a message (or their estimate of it) use the same codebook, but different (thus, independent) codebooks are used 
for different messages. Under this encoding and decoding scheme, 
any rate strictly less than $R$ displayed in (\ref{eqn:achievable_rate_multirelay}) is achievable.

\section{Proof of Theorem~\ref{theorem:multirelay_capacity}}\label{appendix:proof_of_multirelay_channel_capacity_theorem_after_power_allocation}

We want to maximize $R$ given in (\ref{eqn:achievable_rate_multirelay}) subject to the total power constraint, assuming fixed 
relay locations. 
Let us consider $C (\frac{1}{\sigma^{2}×} \sum_{j=1}^{k} ( \sum_{i=0}^{j-1} h_{i,k} \sqrt{P_{i,j}}  )^{2})$, i.e., the $k$-th 
term in the argument of $\min \{\cdots\}$ in (\ref{eqn:achievable_rate_multirelay}). By the monotonicity of $C(\cdot)$, 
it is sufficient to consider $\sum_{j=1}^{k} ( \sum_{i=0}^{j-1} h_{i,k} \sqrt{P_{i,j}}  )^{2}$. 
Now since the channel gains are multiplicative, 
we have:

\footnotesize
\begin{equation}
 \sum_{j=1}^{k} ( \sum_{i=0}^{j-1} h_{i,k} \sqrt{P_{i,j}}  )^{2}=g_{0,k}\sum_{j=1}^{k} \bigg( \sum_{i=0}^{j-1}\frac{\sqrt{P_{i,j}}}{h_{0,i}×}\bigg)^{2}\nonumber
\end{equation}
\normalsize

Thus our optimization problem becomes:

\footnotesize
\begin{eqnarray}
 & & \max \, \min_{k \in \{1,2,\cdots,N+1\}} g_{0,k} \sum_{j=1}^{k} \bigg( \sum_{i=0}^{j-1}\frac{\sqrt{P_{i,j}}}{h_{0,i}×}\bigg)^{2}\nonumber\\
& \textit{s.t} & \, \sum_{j=1}^{N+1}\gamma_{j} \leq P_{T} \,\,
 \textit{and} \,\sum_{i=0}^{j-1}P_{i,j}=\gamma_{j} \, \forall \, 1 \leq j \leq (N+1) \label{eqn:optimization_problem}
\end{eqnarray}
\normalsize

Let us fix $\gamma_{1}, \gamma_{2},\cdots,\gamma_{N+1}$ such that their sum is equal to $P_{T}$. We observe that $P_{i,N+1}$ for $i \in \{0,1,\cdots,N\}$ appear 
in the objective function
only once: for $k=N+1$ through the term $( \sum_{i=0}^{N}\frac{\sqrt{P_{i,N+1}}}{h_{0,i}×})^{2}$.
Since we have fixed $\gamma_{N+1}$, we need to maximize this term over $P_{i,N+1},\, i \in \{0,1,\cdots,N\}$. So we have the following optimization problem:

\footnotesize
\begin{eqnarray}
\max \sum_{i=0}^{N} \frac{\sqrt{P_{i,N+1}}}{h_{0,i}×} \,\,\,\,\,  \textit{s.t}  \,\,\,\,\, \sum_{i=0}^{N} P_{i,N+1}=\gamma_{N+1}\label{eqn:problem}
\end{eqnarray}
\normalsize

By Cauchy-Schwartz inequality, the objective function in this optimization problem is upper bounded by 
(using the fact that $g_{0,i}=h_{0,i}^2$ $\forall i \in \{0,1,\cdots,N\}$):

\footnotesize
\begin{equation*}
 \sqrt{(\sum_{i=0}^{N} P_{i,N+1})(\sum_{i=0}^{N}\frac{1}{g_{0,i}×})}=\sqrt{\gamma_{N+1}\sum_{i=0}^{N}\frac{1}{g_{0,i}}}
\end{equation*}
\normalsize

The upper bound is achieved if 
there exists some $c>0$ such that $\frac{\sqrt{P_{i,N+1}}}{\frac{1}{h_{0,i}×}×}=c$ $\forall i \in \{0,1,\cdots,N\}$. So we have:

\footnotesize
\begin{equation}
 P_{i,N+1}=\frac{c^2}{g_{0,i}×} \,\, \forall i \in \{0,1,\cdots,N\}\nonumber\\
\end{equation}
\normalsize

Since $\sum_{i=0}^{N}P_{i,N+1}=\gamma_{N+1}$, we obtain $c^2=\frac{\gamma_{N+1}}{\sum_{l=0}^{N}\frac{1}{g_{0,l}×}×}$.
 Thus, $ P_{i,N+1}=\frac{\frac{1}{g_{0,i}×}}{\sum_{l=0}^{N}\frac{1}{g_{0,l}×}×}\gamma_{N+1}$.

Here we have used the fact that $h_{0,0}=1$. Now $\{P_{i,N}: i=0,1,\cdots,(N-1)\}$ appear only through the sum $\sum_{i=0}^{N-1}\frac{\sqrt{P_{i,N}}}{h_{0,i}×}$,
 and it appears twice: for $k=N$ and $k=N+1$. We need to maximize this sum subject to the constraint $\sum_{i=0}^{N-1}P_{i,N}=\gamma_{N}$. This optimization 
can be solved in a similar way as before. Thus by repeatedly using this argument and solving optimization problems similar 
in nature to (\ref{eqn:problem}), we obtain:

\footnotesize
\begin{equation}
 P_{i,j}=\frac{\frac{1}{g_{0,i}×}}{\sum_{l=0}^{j-1}\frac{1}{g_{0,l}×}×}\gamma_{j} \,\, \forall 0 \leq i < j \leq (N+1)
\end{equation}
\normalsize

Substituting for $P_{i,j},0 \leq i < j \leq (N+1)$ in (\ref{eqn:optimization_problem}), 
we obtain the following optimization problem:

\footnotesize
\begin{eqnarray}
& & \max \min_{k \in \{1,2,\cdots,N+1\}} g_{0,k} \sum_{j=1}^{k}\bigg(\gamma_{j} \sum_{i=0}^{j-1}\frac{1}{g_{0,i}×} \bigg)\nonumber\\
& \textit{s.t.} & \,\, \sum_{j=1}^{N+1}\gamma_{j} \leq P_{T}
\end{eqnarray}
\normalsize

Let us define $b_{k}:=g_{0,k}$ and $a_{j}:=\sum_{i=0}^{j-1}\frac{1}{g_{0,i}×}$. Observe that $b_{k}$ is decreasing and $a_{k}$ is 
increasing with $k$. Let us define:

\footnotesize
\begin{equation}
 \tilde{s}_{k}(\gamma_{1},\gamma_{2},\cdots,\gamma_{N+1}) := b_{k} \sum_{j=1}^{k} a_{j} \gamma_{j}
\end{equation}
\normalsize

With this notation, our optimization problem becomes: 

\footnotesize
\begin{eqnarray}
\max \min_{1 \leq k \leq N+1} \tilde{s}_{k}(\gamma_{1},\gamma_{2},\cdots,\gamma_{N+1})
\,\,\,\, \textit{s.t.}  \,\, \sum_{j=1}^{N+1}\gamma_{j} \leq P_{T} \label{eqn:modified_optimization_problem}
\end{eqnarray}
\normalsize

\begin{claim}
 Under optimal allocation of $\gamma_{1}, \gamma_{2},\cdots,\gamma_{N+1}$ for the optimization problem 
(\ref{eqn:modified_optimization_problem}),
  $\tilde{s}_{1}=\tilde{s}_{2}=\cdots=\tilde{s}_{N+1}$. \qed
\end{claim}

\begin{proof}
 (\ref{eqn:modified_optimization_problem}) can be rewritten as:
\footnotesize
\begin{eqnarray}
&& \max \zeta \nonumber\\
\textit{s.t} && \zeta \leq b_{k}\sum_{j=1}^{k}a_{j}\gamma_{j} \, \forall \, k \in \{1,2,\cdots,N+1\}, \nonumber\\
&& \sum_{j=1}^{N+1}\gamma_{j} \leq P_{T}, \,\,\, \gamma_{j} \geq 0 \,\, \forall \, 1 \leq j \leq N+1\label{eqn:equality_problem_primal}
\end{eqnarray}
\normalsize
The dual of this linear program is given by:
\footnotesize
\begin{eqnarray}
 &&\min P_{T}\theta \nonumber\\
\textit{s.t} && \sum_{k=1}^{N+1}\mu_{k}=1,\,\,\, \theta \geq 0,\nonumber\\
&& a_{l}\sum_{k=l}^{N+1}b_{k}\mu_{k}+\nu_{l}=\theta \, \forall \, l \in \{1,2,\cdots,N+1\},\nonumber\\
&& \mu_{l} \geq 0, \nu_{l} \geq 0 \, \forall \, l \in \{1,2,\cdots,N+1\} \label{eqn:equality_problem_dual}
\end{eqnarray}
\normalsize
Now, let us consider a primal feasible solution $(\{\gamma_j^*\}_{1 \leq j \leq N+1}, \zeta^*)$ which satisfies:
\footnotesize
\begin{eqnarray}
&&  b_{k}\sum_{j=1}^{k}a_{j}\gamma_{j}^*=\zeta^* \, \forall \, k \in \{1,2,\cdots,N+1\}, \nonumber\\
&& \sum_{j=1}^{N+1}\gamma_{j}^* = P_{T}
\end{eqnarray}
\normalsize
Thus we have, $b_1 a_1 \gamma_1^*=\zeta^*$, i.e., $\gamma_1^*=\frac{\zeta^*}{b_1 a_1×}$. Again, $b_2(a_1 \gamma_1^*+a_2 \gamma_2^*)=\zeta^*$, 
which implies $ \frac{b_2}{b_1×}\zeta^*+b_2 a_2 \gamma_2^*=\zeta^*$.

Thus we obtain $\gamma_2^*=\frac{\zeta^*}{a_2×}(\frac{1}{b_2×}-\frac{1}{b_1×})$. In general, we can write:

\footnotesize
\begin{equation}
 \gamma_k^*=\frac{\zeta^*}{a_k×}\left(\frac{1}{b_k×}-\frac{1}{b_{k-1}×}\right) \, \forall k \in \{1,2,\cdots,N+1\}\nonumber\\
\end{equation}
\normalsize
with $\frac{1}{b_{0}×}:=0$. Now, since $\sum_{k=1}^{N+1}\gamma_k^*=P_T$, we obtain:

\footnotesize
\begin{eqnarray}
  \zeta^*&=&\frac{P_{T}}{\sum_{k=1}^{N+1}\frac{1}{a_{k}×}(\frac{1}{b_k×}-\frac{1}{b_{k-1}×})×}\nonumber\\
 \gamma_j^*&=&\frac{ \frac{1}{a_j×} \left(\frac{1}{b_j×}-\frac{1}{b_{j-1}×}\right) }
{\sum_{k=1}^{N+1}\frac{1}{a_{k}×}(\frac{1}{b_k×}-\frac{1}{b_{k-1}×})×}P_{T}, \, j \in \{1,2,\cdots,N+1\}
\label{eqn:primal_optimal}
\end{eqnarray}
\normalsize
It should be noted that since $b_{k}$ is nonincreasing in $k$, 
the primal variables above are nonnegative and satisfies feasibility 
conditions.
Again, let us consider a dual feasible solution $(\{\mu_j^*,\nu_j^*\}_{1 \leq j \leq N+1}, \theta^*)$ which satisfies:

\footnotesize
\begin{eqnarray}
&& \sum_{k=1}^{N+1}\mu_{k}^*=1, \,\,\, \nu_l^*=0 \, \forall \, l \in \{1,2,\cdots,N+1\}\nonumber\\
&& a_{l}\sum_{k=l}^{N+1}b_{k}\mu_{k}^*+\nu_{l}^*=\theta^* \, \forall \, l \in \{1,2,\cdots,N+1\}
\end{eqnarray}
\normalsize
Solving these equations, we obtain:

\footnotesize
\begin{eqnarray}
 \theta^*&=&\frac{1}{\sum_{k=1}^{N+1}\frac{1}{b_k×}\left(\frac{1}{a_k×}-\frac{1}{a_{k+1}×}\right)×}\nonumber\\
 \mu_j^*&=&\frac{\frac{1}{b_j×}(\frac{1}{a_j×}-\frac{1}{a_{j+1}×})}
{\sum_{k=1}^{N+1}\frac{1}{b_k×}\left(\frac{1}{a_k×}-\frac{1}{a_{k+1}×}\right)×},\, j \in \{1,2,\cdots,N+1\}
\label{eqn:dual_optimal}
\end{eqnarray}
\normalsize
where $\frac{1}{a_{N+2}×}:=0$. Since $a_k$ is increasing in $k$, 
all dual variables are feasible. It is easy to check that $\zeta^*=P_T \theta^*$, 
which means that there is no duality gap. Since the primal is a linear program, the solution 
$(\gamma_1^*, \gamma_2^*,\cdots,\gamma_{N+1}^*, \zeta^*)$ 
is primal optimal. Thus we have established the claim, since the primal optimal solution satisfies it. 
\end{proof}

So let us obtain $\gamma_{1},\gamma_{2},\cdots,\gamma_{N+1}$ for which $\tilde{s}_{1}=\tilde{s}_{2}=\cdots=\tilde{s}_{N+1}$. 
Putting $\tilde{s}_{k}=\tilde{s}_{k-1}$, we obtain 
$b_{k} \sum_{j=1}^{k} a_{j} \gamma_{j}=b_{k-1} \sum_{j=1}^{k-1} a_{j} \gamma_{j}$. 
Thus, we obtain, $ \gamma_{k}=\frac{(b_{k-1}-b_{k})}{b_{k}×} \frac{1}{a_{k}×} \sum_{j=1}^{k-1}a_{j}\gamma_{j}$

Let $d_{k}:=\frac{(b_{k-1}-b_{k})}{b_{k}×} \frac{1}{a_{k}×}$. Hence, $ \gamma_{k}=d_{k} \sum_{j=1}^{k-1}a_{j}\gamma_{j}$. 
From this recursive equation, we have $\gamma_{2}=d_{2}a_{1}\gamma_{1}$,  
$\gamma_{3}=d_{3}(a_{1}\gamma_{1}+a_{2}\gamma_{2})=d_{3}a_{1}(1+a_{2}d_{2})\gamma_{1}$, 
and, in general for $k \geq 3$,
\begin{equation}
 \gamma_{k}=d_{k}a_{1}\Pi_{j=2}^{k-1}(1+a_{j}d_{j})\gamma_{1}\label{eqn:gamma_k_gamma_1}
\end{equation}
\normalsize

Using the fact that $\gamma_{1}+\gamma_{2}+\cdots+\gamma_{N+1}=P_{T}$, we obtain: 

\footnotesize
\begin{equation}
 \gamma_{1}=\frac{P_{T}}{1+d_{2}a_{1}+ \sum_{k=3}^{N+1}d_{k}a_{1} \Pi_{j=2}^{k-1}(1+a_{j}d_{j}) ×} \label{eqn:gamma_1}
\end{equation}
\normalsize

Thus if $\tilde{s}_{1}=\tilde{s}_{2}=\cdots=\tilde{s}_{N+1}$, there is a unique 
allocation $\gamma_{1},\gamma_{2},\cdots,\gamma_{N+1}$. So this must be the one 
maximizing $R$. Hence, optimum $\gamma_{1}$ is obtained by (\ref{eqn:gamma_1}). Then, 
substituting the values of $\{a_{k}:k=0,1,\cdots,N\}$ 
and $d_{k}:k=1,2,\cdots,N+1$ in (\ref{eqn:gamma_k_gamma_1}) and (\ref{eqn:gamma_1}), 
we obtain the values of $\gamma_{1},\gamma_{2},\cdots,\gamma_{N+1}$ 
as shown in Theorem~\ref{theorem:multirelay_capacity}.

Now under these optimal values of $\gamma_{1},\gamma_{2},\cdots,\gamma_{N+1}$, all terms in the argument of $\min \{\cdots\}$ 
in (\ref{eqn:achievable_rate_multirelay}) are equal. So we can consider the first term alone.
 Thus we obtain the expression for $R$ optimized over power allocation among all the nodes 
for fixed relay locations as : $R_{P_T}^{opt}(y_1,y_2,\cdots,y_N)=C \left(\frac{g_{0,1}P_{0,1}}{\sigma^{2}×}\right)=C \left(\frac{g_{0,1}\gamma_{1}}{\sigma^{2}×}\right)$. 
Substituting the expression for $\gamma_{1}$ from (\ref{eqn:gamma_one}), we obtain the achievable rate 
formula (\ref{eqn:capacity_multirelay}).\qed

\footnotesize
\begin{figure*}[!t]
 \begin{eqnarray}\label{eqn:capacity_increasing_with_N}
 & & z_{1}^{*}+\frac{z_{2}^{*}-z_{1}^{*}}{1+z_{1}^{*}×}+\cdots+\frac{e^{\rho y}-z_{i}^{*}}{1+z_{1}^{*}+\cdots+z_{i}^{*}×}+\frac{z_{i+1}^{*}-e^{\rho y}}{1+z_{1}^{*}+\cdots+z_{i}^{*}+e^{\rho y}×}+\cdots.+\frac{e^{\rho L}-z_{N}^{*}}{1+z_{1}^{*}+\cdots+z_{i}^{*}+e^{\rho y}+z_{i+1}^{*}+\cdots+z_{N}^{*}×}\nonumber\\
&  & < z_{1}^{*}+\frac{z_{2}^{*}-z_{1}^{*}}{1+z_{1}^{*}×}+\cdots+\frac{e^{\rho L}-z_{N}^{*}}{1+z_{1}^{*}+\cdots+z_{i}^{*}+z_{i+1}^{*}+\cdots+z_{N}^{*}×} \label{eqn:intermediate_eqn:capacity_increasing_in_N}
\end{eqnarray}
\hrule
\end{figure*}
\normalsize

\section{Proof of Theorem~\ref{theorem:single_relay_total_power}}
\label{appendix:proof_of_single_relay_sum_power_results}

Here we want to place the relay node at a distance $r_{1}$ from the source 
 to minimize $\bigg\{\frac{1}{g_{0,1}×}+\frac{g_{0,1}-g_{0,2}}{g_{0,2}(1+g_{0,1})×}\bigg\}$ 
(see Equation~(\ref{eqn:capacity_multirelay})). Hence, our optimization problem becomes :
\begin{equation}
\min_{r_{1} \in [0,L]} \bigg\{e^{\rho r_{1}}+\frac{e^{-\rho r_{1}}-e^{-\rho L}}{e^{-\rho L}(1+e^{-\rho r_{1}})×}\bigg\}\nonumber\\
\end{equation}
Writing $z_{1}=e^{\rho r_{1}}$, the problem becomes :
\begin{equation}
 \min_{z_{1} \in [1,e^{\rho L}]}  \bigg\{z_{1}-1+\frac{e^{\rho L}+1}{z_{1}+1×}\bigg\}\nonumber\\
\end{equation}
This is a convex optimization problem. Equating the derivative of the objective function to zero, we obtain
 $1-\frac{e^{\rho L}+1}{(z_{1}+1)^{2}×}=0$. Thus the derivative becomes zero at $z_{1}'=\sqrt{1+e^{\rho L}}-1>0$. 
Hence, the objective function is decreasing in $z_{1}$ for $z_{1} \leq z_{1}'$ and increasing in $z_{1} \geq z_{1}'$. 
So the minimizer is $z_{1}^{*}=\max \{z_{1}',1 \}$. 
So the optimum distance of 
the relay node from the source is $y_{1}^{*}=r_{1}^{*}=\max \{0,r_{1}' \}$, where $r_{1}'=\frac{1}{\rho×} \log (\sqrt{1+e^{\rho L}}-1)$. 
Hence, $\frac{y_{1}^{*}}{L×}=\max \{\frac{1}{\lambda×} \log \left(\sqrt{e^{\lambda}+1}-1 \right),0\}$. Now 
$r_{1}' \geq 0$ if and only if $\lambda \geq \log 3$. Hence, $\frac{y_{1}^{*}}{L×}=0$ for $\lambda \leq \log 3$ and 
$\frac{y_{1}^{*}}{L×}=\frac{1}{\lambda×} \log \left(\sqrt{e^{\lambda}+1}-1 \right)$ for $\lambda \geq \log 3$.

{\em For $\lambda \leq \log 3$}, the relay is placed at the source. Then $g_{0,1}=1$ and $g_{0,2}=g_{1,2}=e^{-\lambda}$. 
Then $P_{0,1}=\gamma_{1}=\frac{2P_{T}}{e^{\lambda}+1×}$ (by Theorem~$1$) 
and $R^{*}=C \left(\frac{2P_{T}}{(e^{\lambda}+1)\sigma^{2}×}\right)$. Also $\gamma_{2}=\frac{e^{\lambda}-1}{e^{\lambda}+1×}P_{T}$. 
Hence, $P_{0,2}=P_{1,2}=\frac{e^{\lambda}-1}{e^{\lambda}+1×}\frac{P_{T}}{2×}$.

{\em for $\lambda \geq \log 3$}, the relay is placed at $r_{1}'$. Substituting the value of $r_{1}'$ into 
Equation~($\ref{eqn:power_gamma_relation})$, 
we obtain $P_{0,1}=\gamma_{1}=\frac{P_{T}}{2×}$, $P_{0,2}=\frac{1}{\sqrt{e^{\lambda}+1}×}\frac{P_{T}}{2×}$, 
$P_{1,2}=\frac{\sqrt{e^{\lambda}+1}-1}{\sqrt{e^{\lambda}+1}×}\frac{P_{T}}{2×}$. So in this case
$R^{*}=C \left(\frac{g_{0,1}P_{0,1}}{\sigma^{2}×} \right)$. Since $P_{0,1}=\frac{P_{T}}{2×}$, we have 
$R^{*}=C \left( \frac{1}{\sqrt{e^{\lambda}+1}-1×}\frac{P_{T}}{2 \sigma^{2}×} \right)$.
\qed

\section{Proof of Theorem~\ref{theorem:capacity_increasing_with_N}}
\label{appendix:proof_of_capacity_increases_in_N}

For the $N$-relay problem, let the minimizer in (\ref{eqn:multirelay_optimization}) 
be $z_{1}^{*}, z_{2}^{*},\cdots,z_{N}^{*}$ and let 
$y_{k}^{*}=\frac{1}{\rho×} \log z_{k}^{*}$. Clearly, 
there exists $i \in \{0,1,\cdots,N\}$ such that $y_{i+1}^{*}>y_{i}^{*}$. Let us insert a new relay at a distance $y$ from the source such 
that $y_{i}^{*}<y<y_{i+1}^{*}$. Now we find that we can easily reach 
(\ref{eqn:capacity_increasing_with_N}) (see next page) just by simple comparison. For example, 
\begin{eqnarray*}
& & \frac{e^{\rho y}-z_{i}^{*}}{1+z_{1}^{*}+\cdots+z_{i}^{*}×}+\frac{z_{i+1}^{*}-e^{\rho y}}{1+z_{1}^{*}+\cdots+z_{i}^{*}+e^{\rho y}×} \\
& < & \frac{e^{\rho y}-z_{i}^{*}}{1+z_{1}^{*}+\cdots+z_{i}^{*}×}+\frac{z_{i+1}^{*}-e^{\rho y}}{1+z_{1}^{*}+\cdots+z_{i}^{*}×}\\
& =& \frac{z_{i+1}^{*}-z_{i}^{*}}{1+z_{1}^{*}+\cdots+z_{i}^{*}×}
\end{eqnarray*}
First $i$ terms in the summations of L.H.S (left hand side) and R.H.S (right hand side) 
of (\ref{eqn:capacity_increasing_with_N}) are identical. Also sum of the remaining terms in L.H.S 
is smaller than that of the R.H.S since there is an additional $e^{\rho y}$ in the denominator of each fraction for the L.H.S. Hence, 
we can justify (\ref{eqn:capacity_increasing_with_N}). Now 
R.H.S is precisely the optimum objective function for the $N$-relay placement problem 
(see (\ref{eqn:multirelay_optimization})). On the 
other hand, L.H.S is a particular value of the objective in (\ref{eqn:multirelay_optimization}), for $(N+1)$-relay placement problem.
This clearly implies that by adding one additional relay we can strictly 
improve from $R^{*}$ of the $N$ relay channel. Hence, $R^{*}(N+1)>R^{*}(N)$.
\qed

\section{Proof of Theorem~\ref{theorem:G_increasing_in_lambda}}
\label{appendix:proof_of_G_increasing_in_lambda}

Consider the optimization problem  as shown in (\ref{eqn:multirelay_optimization}). Let us consider $\lambda_1$, $\lambda_2$, with 
$\lambda_1<\lambda_2$, the respective 
minimizers being $(z_{1}^{*},\cdots,z_{N}^{*})$ and $(z_{1}',\cdots,z_{N}')$. Clearly,
\begin{eqnarray}
 G(N,\lambda_1)=\frac{e^{\lambda_1}}{ z_{1}^{*}+\sum_{k=2}^{N+1} \frac{z_{k}^{*}-z_{k-1}^{*}}{\sum_{l=0}^{k-1} z_{l}^{*}}×}
\end{eqnarray}
with $z_{N+1}^{*}=e^{\lambda_1}$ and $z_{0}^{*}=1$. 
With $N \geq 1 $, note that 
$z_{1}^{*}-\frac{z_{N}^{*}}{1+z_{1}^{*}+\cdots+z_{N}^{*}×} \geq 0$, since $z_{1}^{*} \geq 1$ and  
$\frac{z_{N}^{*}}{1+z_{1}^{*}+\cdots+z_{N}^{*}×} \leq 1$. Hence, it is easy to see that 
$\frac{e^{\lambda}}{ z_{1}^{*}-\frac{z_{N}^{*}}{1+z_{1}^{*}+\cdots+z_{N}^{*}×}+\sum_{k=2}^{N} \frac{z_{k}^{*}-z_{k-1}^{*}}{\sum_{l=0}^{k-1} z_{l}^{*}}+\frac{e^{\lambda}}{1+z_{1}^{*}+\cdots+z_{N}^{*}×}×}$ is increasing in 
$\lambda$ where $(z_{1}^{*},\cdots,z_{N}^{*})$ is the optimal solution of (\ref{eqn:multirelay_optimization}) with 
$\lambda=\lambda_1$. Hence,
\begin{eqnarray}
 G(N,\lambda_1)&=&\frac{e^{\lambda_1}}{ z_{1}^{*}+\sum_{k=2}^{N} \frac{z_{k}^{*}-z_{k-1}^{*}}{\sum_{l=0}^{k-1} z_{l}^{*}}+\frac{e^{\lambda_1}-z_{N}^{*}}{\sum_{l=0}^{N} z_{l}^{*}}×}\nonumber\\
&\leq & \frac{e^{\lambda_2}}{ z_{1}^{*}+\sum_{k=2}^{N} \frac{z_{k}^{*}-z_{k-1}^{*}}{\sum_{l=0}^{k-1} z_{l}^{*}}+\frac{e^{\lambda_2}-z_{N}^{*}}{1+z_{1}^{*}+\cdots+z_{N}^{*}×}×}\nonumber\\
&\leq & \frac{e^{\lambda_2}}{ z_{1}^{'}+\sum_{k=2}^{N} \frac{z_{k}^{'}-z_{k-1}^{'}}{\sum_{l=0}^{k-1} z_{l}^{'}}+\frac{e^{\lambda_2}-z_{N}^{'}}{1+z_{1}^{'}+\cdots+z_{N}^{'}×}×}\nonumber\\
&=& G(N,\lambda_2)
\end{eqnarray}
The second inequality follows from the fact that $(z_{1}',\cdots,z_{N}')$ minimizes 
$z_{1}+\sum_{k=2}^{N+1} \frac{z_{k}-z_{k-1}}{\sum_{l=0}^{k-1} z_{l}}$ subject to the constraint 
$1 \leq z_1 \leq z_2 \leq \cdots \leq z_N \leq z_{N+1} =e^{\lambda_2}$.

Hence, $G(N,\lambda)$ is increasing in $\lambda$ for fixed $N$.
\qed

\section{Proof of Theorem~\ref{theorem:large_nodes_uniform}}
\label{appendix:proof_of_large_nodes_uniform}

When $N$ relay nodes are uniformly placed along a line, we will have $y_{k}=\frac{kL}{N+1×}$. 
Then our formula for achievable rate $R_{P_T}^{opt}(y_1,y_2,\cdots,y_N)$
for sum power constraint becomes:
$R_{N}=C(\frac{P_{T}}{\sigma^{2}×}\frac{1}{f(N)×})$ where $f(N)=a_N+\sum_{k=2}^{N+1}\frac{a_N^{k}-a_N^{k-1}}{1+a_N+\cdots+a_N^{k-1}×}$ with 
$a_N=e^{ \frac{\rho L}{N+1×}}=e^{\frac{\lambda}{N+1×}}$.

Since $a_N>1$ for all $N<\infty$ and $\rho>0$, we have 
$f(N) > a_N$ for all $N \geq 1$ and hence, 
$\liminf_N f(N) \geq \lim_{N \rightarrow \infty} a_N =1$.

Now,
\begin{eqnarray}
f(N)&=& a_N+\sum_{k=1}^{N}\frac{a_N^{k+1}-a_N^{k}}{1+a_N+\cdots+a_N^{k}×} \nonumber\\
&=& a_N+ (a_N-1)^{2}\sum_{k=1}^{N} \frac{a_N^{k}}{a_N^{k+1}-1×}\nonumber\\
&\leq& a_N+ (a_N-1)^{2}\sum_{k=1}^{N} \frac{a_N^{k}}{a_N^{k}-1×}\nonumber\\
&=& e^{\frac{\lambda}{N+1×}}+ (e^{\frac{\lambda}{N+1×}}-1)^{2}\sum_{k=1}^{N} \frac{e^{\frac{k \lambda}{N+1×}}}{e^{\frac{k \lambda}{N+1×}}-1×}\nonumber\\
&\leq& e^{\frac{\lambda}{N+1×}}+ (e^{\frac{\lambda}{N+1×}}-1)^{2}\sum_{k=1}^{N} \frac{e^{\frac{k \lambda}{N+1×}}}{\frac{k \lambda}{N+1×}}\nonumber\\
&=& e^{\frac{\lambda}{N+1×}}+ (e^{\frac{\lambda}{N+1×}}-1)^{2} \frac{(N+1)}{\lambda×} \sum_{k=1}^{N} \frac{e^{\frac{k \lambda}{N+1×}}}{k}
\label{eqn:inequality_of_fN}
\end{eqnarray}

where the first inequality follows from the fact that $a_N>1$ and the second inequality follows from the 
fact that $e^{\frac{k \lambda}{N+1×}} \geq 1+ \frac{k \lambda}{N+1×}$.

Now, by Cauchy-Schwartz inequality, 
\begin{eqnarray}
\sum_{k=1}^{N} \frac{e^{\frac{k \lambda}{N+1×}}}{k} \leq \sqrt{(\sum_{k=1}^{N}e^{\frac{2 k \lambda}{N+1×}}) (\sum_{k=1}^{N} \frac{1}{k^2×})}
\end{eqnarray}

Since $\sum_{k=1}^{\infty}\frac{1}{k^2×}=\frac{\pi^2}{6×}$, we can write:

\begin{eqnarray}
\sum_{k=1}^{N} \frac{e^{\frac{k \lambda}{N+1×}}}{k} \leq \sqrt{(\sum_{k=1}^{N}e^{\frac{2 k \lambda}{N+1×}}) \frac{\pi^2}{6×}} \label{eqn:inequality_using_the_series}
\end{eqnarray}

Hence, by (\ref{eqn:inequality_using_the_series}) and (\ref{eqn:inequality_of_fN}), 

\footnotesize
\begin{eqnarray}
 f(N) &\leq& e^{\frac{\lambda}{N+1×}}+ (e^{\frac{\lambda}{N+1×}}-1)^{2} \frac{(N+1)\pi}{\sqrt{6}\lambda×} \sqrt{\sum_{k=1}^{N}e^{\frac{2 k \lambda}{N+1×}}} \nonumber\\
&=&  e^{\frac{\lambda}{N+1×}}+ (e^{\frac{\lambda}{N+1×}}-1)^{2} \frac{(N+1)\pi}{\sqrt{6}\lambda×} \sqrt{e^{\frac{2 \lambda}{N+1×}} \frac{(e^{\frac{2 N \lambda}{N+1×}}-1)}{(e^{\frac{2 \lambda}{N+1×}}-1)×}}\nonumber\\
\end{eqnarray}
\normalsize

Now, since $e^{\frac{2 N \lambda}{N+1×}}-1 \leq e^{\frac{2 N \lambda}{N+1×}} \leq e^{2\lambda}$, we obtain:

\footnotesize
\begin{eqnarray*}
 f(N) &\leq& e^{\frac{\lambda}{N+1×}}+ (e^{\frac{\lambda}{N+1×}}-1)^{2} \frac{(N+1)\pi}{\sqrt{6}\lambda×} e^{\lambda} e^{\frac{\lambda}{N+1×}} \sqrt{\frac{1}{(e^{\frac{2 \lambda}{N+1×}}-1)×}} \nonumber\\
  &=& e^{\frac{\lambda}{N+1×}}+ (e^{\frac{\lambda}{N+1×}}-1)^{\frac{3}{2×}} \frac{(N+1)\pi}{\sqrt{6}\lambda×} e^{\lambda} e^{\frac{\lambda}{N+1×}} \sqrt{\frac{1}{(e^{\frac{ \lambda}{N+1×}}+1)×}} \nonumber\\
\end{eqnarray*}
\normalsize
 
Hence, 
\begin{eqnarray*}
 && \limsup_N f(N) \nonumber\\
& \leq & 1+ \frac{\pi e^{\lambda}}{\sqrt{12}×} \lim_{N \rightarrow \infty} \frac{(N+1)}{\lambda} (e^{\frac{\lambda}{N+1×}}-1)^{\frac{3}{2×}} \nonumber\\
\end{eqnarray*}

Putting $q=\frac{\lambda}{N+1×}$,
 
\begin{eqnarray*}
 \limsup_N f(N) &\leq&  1+ \frac{\pi e^{\lambda}}{\sqrt{12}×} \lim_{q \rightarrow 0} \sqrt{q}  \lim_{q \rightarrow 0} (\frac{e^q-1}{q×})^{\frac{3}{2×}} \nonumber\\
   &=& 1 \nonumber\\
\end{eqnarray*}

Now, we have proved that $\limsup_N f(N) \leq 1 \leq \liminf_N f(N)$ and hence $\lim_{N \rightarrow \infty} f(N)=1$. 
Hence, $\lim_{N \rightarrow \infty} R_N = C(\frac{P_T}{\sigma^{2}})$ and the theorem is proved.
\qed

\section{Proof of Theorem~\ref{theorem:convergence_of_value_iteration}}\label{appendix:sequential_placement_total_power}
\label{appendix:proof_of_value_iteration_convergence}

As we have seen in Section \ref{sec:mdp_total_power}, our problem is a negative dynamic programming problem 
(i.e., the $\mathsf{N}$ case of \cite{schal75conditions-optimality}, 
where single-stage rewards are non-positive). It is to be noted 
that Sch\"{a}l \cite{schal75conditions-optimality} discusses two other kind of problems as well: 
the $\mathsf{P}$ case (single-stage rewards are positive) 
and the $\mathsf{D}$ case (the reward at stage $k$ is discounted by a 
factor $\alpha^k$, where $0<\alpha<1$). In this appendix, 
we first state a general-purpose theorem for the value iteration (Theorem~\ref{thm:value_iteration_general}), 
prove it by some results of \cite{schal75conditions-optimality}, 
and then we use this theorem to prove Theorem \ref{theorem:convergence_of_value_iteration}.

\subsection{A General Result (Derived from \cite{schal75conditions-optimality})}\label{appendix_subsection_schal-discussion}

Consider an infinite horizon total cost MDP whose state 
space $\mathcal{S}$ is an interval in $\mathbb{R}$ and the action space 
$\mathcal{A}$ is $[0,\infty)$. Let the set of possible actions at state $s$ be denoted by $\mathcal{A}(s)$. 
Let the single-stage cost be $c(s,a,w) \geq 0$ where $s$, $a$ and $w$ are the state, 
the action and the disturbance, respectively. Let us denote the optimal expected cost-to-go at state $s$ by $V^*(s)$. Let the state of 
the system evolve as $s_{k+1}=h(s_k,a_k,w_k)$, where $s_k$, $a_k$ and $w_k$ are the state, the action and 
the disturbance at the $k$-th 
instant, respectively. Let $s^{*} \in \mathcal{S}$ be an absorbing state with $c(s^{*},a, w)=0$ for all $a$, $w$. 
Let us consider the value iteration for all $s \in \mathcal{S}$, with $V^{(0)}(\cdot)= 0 $:

\footnotesize
\begin{eqnarray}
 V^{(k+1)}(s)&=&\inf_{a \in [0,\infty)} \mathbb{E}_{w} \bigg( c(s,a,w)+V^{(k)}(h(s,a,w)) \bigg), s \neq s^{*}\nonumber\\
V^{(k+1)}(s^{*})&=& 0\label{eqn:value_iteration_general}
\end{eqnarray}
\normalsize  
We provide some results and concepts from \cite{schal75conditions-optimality}, which will be used later to prove 
Theorem \ref{theorem:convergence_of_value_iteration}.

\begin{thm}\label{thm:schal_convergence_value_iteration}
       [{\em Theorem 4.2 (\cite{schal75conditions-optimality})}] $V^{(k)}(s)\rightarrow V^{(\infty)}(s)$ for all $s \in \mathcal{S}$, 
i.e., the value iteration (\ref{eqn:value_iteration_general}) converges. \qed
\end{thm}

Let us recall that $\Gamma_k(s)$ is the set of minimizers of (\ref{eqn:value_iteration_general}) at the 
$k$-th iteration at state $s$, if the infimum is achieved at some $a<\infty$. 
$\Gamma_{\infty}(s):=\{a \in \mathcal{A}:a$ is an 
accumulation point of some sequence $\{a_k\}$ where each $a_k \in \Gamma_{k}(s)\}$. $\Gamma^*(s)$ is the set of minimizers in 
the Bellman Equation.

Let $\mathcal{C}(\mathcal{A})$ be the set of nonempty compact subsets of $\mathcal{A}$. 
 The Hausdorff metric $d$ on $\mathcal{C}(\mathcal{A})$ is defined as follows:
\begin{equation*}
 d(C_1,C_2)=\max \{ \sup_{c \in C_1}\rho(c,C_2), \, \sup_{c \in C_2}\rho(c,C_1) \}
\end{equation*}
where $\rho(c,C)$ is the minimum distance between the point $c$ and the compact set $C$. 

\begin{prop}\label{prop:separable_hausdorff}
     [Proposition 9.1(\cite{schal75conditions-optimality})]
$(\mathcal{C}(\mathcal{A}),d)$ is a separable metric space. 
 \end{prop}

A mapping $\phi:\mathcal{S}\rightarrow \mathcal{C}(\mathcal{A})$ is called measurable if it is measurable with respect to
the Borel $\sigma$-algebra of $(\mathcal{C}(\mathcal{A}),d)$.

 $\hat{\mathcal{F}}(\mathcal{S} \times \mathcal{A})$ is the set of all measurable functions 
$v:\mathcal{S} \times \mathcal{A} \rightarrow \mathbb{R}$ which are bounded below and where every such $v(\cdot)$ is the limit 
of a non-decreasing sequence of measurable, bounded functions $v_k:\mathcal{S} \times \mathcal{A} \rightarrow \mathbb{R}$.

We will next present a condition followed by a theorem. The condition, if satisfied, implies 
the convergence of value iteration (\ref{eqn:value_iteration_general}) 
to the optimal value function (according to the theorem).
\begin{condition} \label{condition_A}
[{\em Derived from Condition A in \cite{schal75conditions-optimality}}]
\begin{enumerate}[label=(\roman{*})]
 \item $\mathcal{A}(s)\in \mathcal{C}(\mathcal{A})$ for all $s \in \mathcal{S}$ and 
$\mathcal{A}:\mathcal{S}\rightarrow \mathcal{C}(\mathcal{A})$ is measurable.
\item $\mathbb{E}_{w}(c(s,a,w)+V^{(k)}(h(s,a,w))) $ is in $\hat{\mathcal{F}}(\mathcal{S}\times \mathcal{A})$ for all $k \geq 0$.\qed
\end{enumerate}    
\end{condition}
\begin{thm}\label{thm:schal_main_theorem}
 [{\em Theorem $13.3$, \cite{schal75conditions-optimality}}] If $c(s,a,w) \geq 0$ for all $s,a,w$ and Condition \ref{condition_A} 
holds:
\begin{enumerate}[label=(\roman{*})]
 \item $V^{(\infty)}(s)=V^*(s)$, $s \in \mathcal{S}$.
\item $\Gamma_{\infty}(s) \subset \Gamma^*(s)$.
\item There is a stationary optimal policy $f^{\infty}$ where $f:\mathcal{S} \rightarrow \mathcal{A}$ and $f(s) \in \Gamma_{\infty}(s)$ 
for all $s \in \mathcal{S}$.\qed
\end{enumerate}
 \end{thm}

The next condition 
and theorem deal with the situation where the action space is noncompact. 
\begin{condition} \label{condition_B}
[{\em Condition B (\cite{schal75conditions-optimality})}] There is a measurable mapping 
$\underline{\mathcal{A}}:\mathcal{S} \rightarrow \mathcal{C}(\mathcal{A})$ 
 such that:
\begin{enumerate}[label=(\roman{*})]
 \item $\underline{{\mathcal{A}}}(s)\subset {\mathcal{A}}(s)$ for all $s \in \mathcal{S}$.
\item \small{$\inf_{a \in \mathcal{A}(s)-\underline{\mathcal{A}}(s)} \mathbb{E}_{w}\bigg(c(s,a,w)+V^{(k)}(h(s,a,w))\bigg)> \inf_{a \in \mathcal{A}(s)}\mathbb{E}_{w}\bigg(c(s,a,w)+V^{(k)}(h(s,a,w))\bigg)$} 
for all $k \geq 0$.
\end{enumerate} \qed
 \end{condition}

This condition requires that for each state $s$, there is a compact set $\underline{\mathcal{A}}(s)$ of actions such that no optimizer 
 of the value iteration lies outside the set $\underline{\mathcal{A}}(s)$ at any stage $k \geq 0$. 
\begin{thm}\label{thm:schal_compact_action}
       [{\em Theorem $17.1$, \cite{schal75conditions-optimality}}] If Condition \ref{condition_B} is satisfied 
and if the three statements in Theorem \ref{thm:schal_main_theorem} are valid for the modified problem having 
admissible set of actions $\underline{\mathcal{A}}(s)$ for each state $s \in \mathcal{S}$, 
then those statements are valid for the original problem as well.\qed
      \end{thm}

Now we will provide an important theorem which will be used to prove Theorem~\ref{theorem:convergence_of_value_iteration}.

\begin{thm}\label{thm:value_iteration_general}
 If the value iteration (\ref{eqn:value_iteration_general}) satisfies the following conditions:
\begin{enumerate}[label=(\alph*)]
\item For each $k$, $\mathbb{E}_{w}\bigg(c(s,a,w)+V^{(k)}(h(s,a,w))\bigg)$ is jointly continuous in $a$ and $s$ for $s \neq s^{*}$.
\item The infimum in (\ref{eqn:value_iteration_general}) is achieved in $[0,\infty)$ for all $s \neq s^{*}$.
\item For each $s \in \mathcal{S}$, there exists $a(s)<\infty$ such that $a(s)$ is continuous in $s$ for $s \neq s^{*}$, 
and no minimizer of (\ref{eqn:value_iteration_general}) lies in $(a(s), \infty)$ for each $k \geq 0$.

 \end{enumerate}

Then the following hold:
\begin{enumerate}[label=(\roman{*})]
 \item The value iteration converges, i.e., $V^{(k)}(s) \rightarrow V^{(\infty)}(s)$ for all $s \neq s^{*}$.
\item $V^{(\infty)}(s)=V^*(s)$ for all $s \neq s^{*}$.
\item $\Gamma_{\infty}(s) \subset \Gamma^*(s)$ for all $s \neq s^{*}$.
\item There is a stationary optimal policy $f^{\infty}$ where $f:\mathcal{S} \setminus \{s^{*}\} \rightarrow \mathcal{A}$ and $f(s) \in \Gamma_{\infty}(s) \,\,\forall \,\, s \neq s^{*}$.\qed
\end{enumerate}

\end{thm}

\textbf{\em Proof of Theorem \ref{thm:value_iteration_general}:}
By Theorem \ref{thm:schal_convergence_value_iteration}, the value iteration converges, i.e., 
$V^{(k)}(s)\rightarrow V^{(\infty)}(s)$. Moreover, $V^{(k)}(s)$ 
is the optimal cost for a $k$-stage problem with zero terminal cost, 
and the cost at each stage is positive. Hence, $V^{(k)}(s)$ 
increases in $k$ for every $s \in \mathcal{S}$. Thus, for all $s \in \mathcal{S}$, $V^{(k)}(s) \uparrow V^{(\infty)}(s)$.

Now, Condition \ref{condition_B} and Theorem \ref{thm:schal_compact_action} 
say that if no optimizer of the value iteration in each 
stage $k$ lies outside a compact subset $\underline{\mathcal{A}}(s)$ 
of $\mathcal{A}(s) \subset \mathcal{A}$, 
then we can deal with the modified problem having a new action space $\underline{\mathcal{A}}(s)$. 
If the value iteration converges to the optimal 
value in this modified problem, then it will converge to the optimal 
value in the original problem as well, provided that the mapping 
$\underline{\mathcal{A}}:\mathcal{S}\rightarrow \mathcal{C}(\mathcal{A})$ 
is measurable. Let us choose $\underline{\mathcal{A}}(s):=[0,a(s)]$, where $a(s)$ satisfies hypothesis (c) of 
Theorem \ref{thm:value_iteration_general}. Since $a(s)$ is continuous at $s \neq s^{*}$, 
for any $\epsilon>0$ 
we can find a $\delta_{s,\epsilon}>0$ such that $|a(s)-a(s')|<\epsilon$ whenever $|s-s'|<\delta_{s,\epsilon}$, $s \neq s^{*}$, 
$s' \neq s^{*}$. Now, when  
$|a(s)-a(s')|<\epsilon$, we have $d([0,a(s)],[0,a(s')])<\epsilon$. Hence, the mapping 
$\underline{\mathcal{A}}:\mathcal{S}\rightarrow \mathcal{C}(\mathcal{A})$ is 
continuous at all $s \neq s^{*}$, and thereby measurable in this case. Hence, 
the value iteration (\ref{eqn:value_iteration_general}) satisfies Condition \ref{condition_B}.

Thus, the value iteration for $s \neq s^{*}$ can be modified as:
\footnotesize
\begin{equation}
  V^{(k+1)}(s)=\inf_{a \in [0,a(s)]} \mathbb{E}_{w} ( c(s,a,w)+V^{(k)}(h(s,a,w)) )
\label{eqn:modified_value_iteration_general}
\end{equation}
\normalsize
Now, $\mathbb{E}_{w} ( c(s,a,w)+V^{(k)}(h(s,a,w)) )$ is continuous 
(can be discontinuous at $s=s^{*}$, since this quantity 
is $0$ at $s=s^*$) 
on $\mathcal{S} \times \mathcal{A}$ (by our hypothesis). Hence, $\mathbb{E}_{w} ( c(s,a,w)+V^{(k)}(h(s,a,w)) )$ 
is measurable on $\mathcal{S} \times \mathcal{A}$. Also, it is bounded below by $0$. 
Hence, it can be approximated by an increasing sequence of bounded measurable functions $\{v_{n,k}\}_{n \geq 1}$ given by 
$v_{n,k}(s,a)=\min  \{\mathbb{E}_{w} ( c(s,a,w)+V^{(k)}(h(s,a,w)) ),\,n \}$. Hence, 
$\mathbb{E}_{w} ( c(s,a,w)+V^{(k)}(h(s,a,w)) )$ is in $\hat{\mathcal{F}}(\mathcal{S} \times \mathcal{A})$.

Thus, Condition \ref{condition_A} is satisfied for the modified problem and therefore, by Theorem \ref{thm:schal_main_theorem}, 
the modified value iteration in (\ref{eqn:modified_value_iteration_general}) converges to the optimal value function. 
Now, by Theorem 
\ref{thm:schal_compact_action}, we can argue that the value iteration (\ref{eqn:value_iteration_general}) converges 
to the optimal value function in the original problem and hence $V^{(\infty)}(s)=V^*(s)$ for all $s \in \mathcal{S} \setminus s^*$. 
Also, $\Gamma_{\infty}(s) \subset \Gamma^*(s)$ for all 
$s \in \mathcal{S} \setminus s^*$ and there exists a stationary 
optimal policy $f^{\infty}$ where $f(s) \in \Gamma_{\infty}(s)$ for all $s \in \mathcal{S} \setminus s^*$ 
(by Theorem \ref{thm:schal_main_theorem}).

\subsection{Proof of Theorem \ref{theorem:convergence_of_value_iteration}}
\label{appendix_subsection-convergence-value-iteration-proof}

This proof uses the results of Theorem \ref{thm:value_iteration_general} provided in this appendix. Remember that the state 
$\mathbf{EOL}$ is absorbing and $c(\mathbf{EOL}, a, w)=0$ for all $a$, $w$. We can think of it as state $0$ so that our state space 
becomes $[0,1]$ which is a Borel set. We will see that the state $0$ plays the role of the state $s^*$ 
as mentioned in Theorem \ref{thm:value_iteration_general}.

We need to check whether the conditions (a), (b), and (c) in Theorem \ref{thm:value_iteration_general} 
are satisfied for the value iteration (\ref{eqn:value_iteration}). Of course, 
$J_{\xi}^{(0)}(s)=0$ is concave, increasing in $s \in (0,1]$. Suppose that 
$J_{\xi}^{(k)}(s)$ is concave, increasing in $s$ for some $k \geq 0$. 
Also, for any fixed $a \geq 0$, $\frac{se^{\rho a}}{1+se^{\rho a}×}$ 
is concave and increasing in $s$. Thus, by the composition rule for the composition of a concave increasing 
function $J_{\xi}^{(k)}(\cdot)$ and a concave increasing function $\frac{se^{\rho a}}{1+se^{\rho a}×}$, 
for any $a \geq 0$ the term $J_{\xi}^{(k)}\left(\frac{se^{\rho a}}{1+se^{\rho a}×}\right)$ 
is concave, increasing over $s \in (0,1]$. Hence, 
$\int_{0}^{a}\beta e^{-\beta z}s(e^{\rho z}-1)dz + e^{-\beta a}\bigg(s(e^{\rho a}-1)+\xi+J_{\xi}^{(k)}\left(\frac{se^{\rho a}}{1+se^{\rho a}×}\right)\bigg)$ 
 (in (\ref{eqn:value_iteration})) is concave increasing over $s \in (0,1]$. Since the 
infimization over $a$ preserves concavity, we conclude that 
$J_{\xi}^{(k+1)}(s)$ is concave, increasing over $s \in (0,1]$. 
Hence, for each $k$, $J_{\xi}^{(k)}(s)$ is continuous in $s$ over $(0,1)$, since otherwise concavity w.r.t. $s$ will be violated. 
Now, we must have $J_{\xi}^{(k)}(1) \leq \lim_{s \uparrow 1} J_{\xi}^{(k)}(s)$, since otherwise 
the concavity of $J_{\xi}^{(k)}(s)$ will be violated. 
But since $J_{\xi}^{(k)}(s)$ 
is increasing in $s$, $J_{\xi}^{(k)}(1) \geq \lim_{s \uparrow 1} J_{\xi}^{(k)}(s)$. Hence, $J_{\xi}^{(k)}(1) = \lim_{s \uparrow 1}  J_{\xi}^{(k)}(s)$. 
Thus, $J_{\xi}^{(k)}(s)$ is continuous in $s$ over $(0,1]$ for each $k$.

Hence, $\int_{0}^{a}\beta e^{-\beta z}s(e^{\rho z}-1)dz + e^{-\beta a}(s(e^{\rho a}-1)+\xi+J_{\xi}^{(k)}(\frac{se^{\rho a}}{1+se^{\rho a}×}) )$ 
is continuous in $s,a$ for $s \neq 0$. Hence, condition (a) in Theorem \ref{thm:value_iteration_general} 
is satisfied. 

Now, we will check condition (c) in Theorem \ref{thm:value_iteration_general}.

By Theorem \ref{thm:schal_convergence_value_iteration}, the value iteration converges, i.e., 
$J_{\xi}^{(k)}(s)\rightarrow J_{\xi}^{(\infty)}(s)$. Also, $J_{\xi}^{(\infty)}(s)$ is concave, increasing 
in $s \in (0,1]$ and hence continuous. Moreover, $J_{\xi}^{(k)}(s)$ 
is the optimal cost for a $k$-stage problem with zero terminal cost, 
and the cost at each stage is positive. Hence, $J_{\xi}^{(k)}(s)$ 
increases in $k$ for every $s \in (0,1]$. Thus, for all $s \in (0,1]$, 
$J_{\xi}^{(k)}(s) \uparrow J_{\xi}^{(\infty)}(s)$.

Again, $J_{\xi}^{(k)}(s)$ is the optimal cost for a $k$-stage problem with zero terminal cost. Hence, it is less than or equal to 
the optimal cost for the infinite horizon problem with the same transition law and cost structure. Hence, 
$J_{\xi}^{(k)}(s)\leq J_{\xi}(s)$ for all $k \geq 1$. Since $J_{\xi}^{(k)}(s)\uparrow J_{\xi}^{(\infty)}(s)$, 
we have $J_{\xi}^{(\infty)}(s) \leq J_{\xi}(s)$.

Now, consider the following two cases:

\subsubsection{$\beta>\rho$}
Let us define a function $\psi:(0,1]\rightarrow \mathbb{R}$ by 
$\psi(s)=\frac{J_{\xi}^{(\infty)}(s)+\theta s}{2×}$. By Proposition \ref{prop:upper_bound_on_cost_beta_geq_rho}, 
$J_{\xi}(s)<\theta s$ for all $s \in (0,1]$. Hence, $J_{\xi}^{(\infty)}(s)<\psi(s)<\theta s$ and $\psi(s)$ 
is continuous over $s \in (0,1]$.
Since $\beta>\rho$ and $J_{\xi}^{(k)}(s) \in [0, \theta]$ for any 
$s$ in $(0,1]$, the expression $\theta s + e^{-\beta a}\bigg(-\theta s e^{\rho a}+\xi+J_{\xi}^{(k)}\left(\frac{se^{\rho a}}{1+se^{\rho a}×}\right)\bigg)$ 
obtained from the R.H.S of (\ref{eqn:value_iteration}) converges to $\theta s$ as $a \rightarrow \infty$. 
A lower bound to this expression is $\theta s + e^{-\beta a}(-\theta s e^{\rho a})$.
With $\beta > \rho$, for each $s$, there exists 
$a(s)<\infty$ such that $\theta s + e^{-\beta a}(-\theta s e^{\rho a})> \psi(s)$ for 
all $a > a(s)$. 
But $\theta s + \inf_{a \geq 0} e^{-\beta a}\bigg(-\theta s e^{\rho a}+\xi+J_{\xi}^{(k)}\left(\frac{se^{\rho a}}{1+se^{\rho a}×}\right)\bigg)$ 
is equal to $J_{\xi}^{(k+1)}(s)<\psi(s)$. 
Hence, for any $s \in (0,1]$, the minimizers for (\ref{eqn:bellman_equation_simplified_in_a}) always 
lie in the compact interval $[0,a(s)]$ for all $k \geq 1$.
 Since $\psi(s)$ is continuous in $s$, 
we can choose $a(s)$ as a continuous function of $s$ on $(0,1]$.

\subsubsection{$\beta \leq \rho$}
Fix $A$, $0 <A <\infty$. Let $K:=\frac{1}{\beta A ×}\left(\xi+(e^{\rho A}-1)  \right)+(e^{\rho A}-1)$. 
Then, by Proposition \ref{prop:upper_bound_on_cost}, $J_{\xi}(s) \leq K$ for all $s \in (0,1]$. Now, we observe that the objective function 
(for minimization over $a$) in the R.H.S of (\ref{eqn:value_iteration}) is lower bounded by 
$\int_{0}^{a}\beta e^{-\beta z}s(e^{\rho z}-1)dz$, which is continuous in $s,a$ and goes to $\infty$ as $a \rightarrow \infty$ for 
each $s \in (0,1]$. Hence, 
for each $s \in (0,1]$, there exists $0<a(s)<\infty$ such that 
$\int_{0}^{a}\beta s e^{-\beta z}(e^{\rho z}-1)dz >2K$ 
for all $a>a(s)$ and $a(s)$ is continuous over $s \in (0,1]$. But $J_{\xi}^{(k+1)}(s) \leq J_{\xi}(s) \leq K$ for 
all $k$. Hence, the minimizers in (\ref{eqn:value_iteration}) always lie in $[0,a(s)]$ 
where $a(s)$ is independent of $k$ and continuous over $s \in (0,1]$.

Let us set $a(0)=a(1)$.\footnote{Remember that at state $0$ (i.e., state $\mathbf{EOL}$), 
the single stage cost is $0$ irrespective of the action, 
and that this state is absorbing. Hence, any action at state $0$ can be optimal.} Then, the chosen function $a(s)$ is 
continuous over $s \in (0,1]$ and can be discontinuous only at $s=0$.
Thus, condition (c) of Theorem \ref{thm:value_iteration_general} has 
been verified for the value iteration (\ref{eqn:value_iteration}). Condition (b) of Theorem \ref{thm:value_iteration_general} is 
obviously satisfied since a continuous function over a compact set always has a minimizer.
\qed

{\em Remark:} Observe that in our value iteration (\ref{eqn:value_iteration}) it is always sufficient 
to deal with compact action spaces, and the objective 
functions to be minimized at each stage of the value iteration are continuous in $s$, $a$. Hence, $\Gamma_{k}(s)$ is nonempty 
for each $s \in (0,1]$, $k \geq 0$.
Also, since there exists $K>0$ such that $J_{\xi}^{(k)}(s) \leq K$ 
for all $k \geq 0$, $s \in (0,1]$, it is sufficient to 
restrict the action space in (\ref{eqn:value_iteration}) 
to a set $[0, a(s)]$ for any $s \in (0,1]$, $k \geq 0$. Hence, $\Gamma_k (s) \subset [0,a(s)]$ 
for all $s \in (0,1]$, $k \geq 0$. 
Now, for a fixed $s \in (0,1]$, any sequence $\{a_k\}_{k \geq 0}$ with $a_k \in \Gamma_k (s)$, in bounded. Hence, the sequence 
must have a limit point. Hence, $\Gamma_{\infty}(s)$ is nonempty for each $s \in (0,1]$. Since $\Gamma_{\infty}(s) \subset \Gamma^*(s)$,
 $\Gamma^*(s)$ is nonempty for each $s \in (0,1]$.

\section{Proofs of Propositions \ref{prop:increasing_concave_in_s}, \ref{prop:increasing_concave_in_lambda} and \ref{prop:continuity_of_cost}}\label{appendix_subsection_policy-structure}
\label{appendix:proof_of_propositions}

\subsection{Proof of Proposition \ref{prop:increasing_concave_in_s}}
Fix $\xi$. Consider the value iteration (\ref{eqn:value_iteration}).
Let us start with $J_{\xi}^{(0)}(s):=0$ for all $s \in (0,1]$. Clearly, $J_{\xi}^{(1)}(s)$ is concave and increasing in $s$, 
since pointwise infimum of linear functions is concave. Now let us assume that $J_{\xi}^{(k)}(s)$ 
is concave and increasing in $s$. Then, by the 
composition rule, it is easy to show that $J_{\xi}^{(k)}(\frac{se^{\rho a}}{1+se^{\rho a}×})$ is concave and 
increasing in $s$ for any fixed $a\geq 0$. Hence, 
$J_{\xi}^{(k+1)}(s)$ is concave and increasing, since pointwise infimum of a set of concave and increasing functions is 
concave and increasing. 
By Theorem \ref{theorem:convergence_of_value_iteration}, $J_{\xi}^{(k)}(s)\rightarrow J_{\xi}(s)$. 
Hence, $J_{\xi}(s)$ is concave and increasing in $s$.\qed

\subsection{Proof of Proposition \ref{prop:increasing_concave_in_lambda}}
Consider the value iteration (\ref{eqn:value_iteration}). Since $J_{\xi}^{(0)}(s):=0$ for all $s \in (0,1]$,  
$J_{\xi}^{(1)}(s)$ is obtained by taking infimum (over $a$) of a linear, increasing function of $\xi$. 
Hence, $J_{\xi}^{(1)}(s)$ is concave, increasing over $\xi \in (0, \infty)$. 
If we assume that $J_{\xi}^{(k)}(s)$ is concave and increasing 
in $\xi$, then $J_{\xi}^{(k)}(\frac{se^{\rho a}}{1+se^{\rho a}×})$ is also concave and increasing 
in $\xi$ for fixed $s$ and $a$. Thus, $J_{\xi}^{(k+1)}(s)$ 
is also concave and increasing in $\xi$. Now, $J_{\xi}^{(k)}(s) \rightarrow J_{\xi}(s)$ for all $s \in \mathcal{S}$, 
 and $J_{\xi}^{(k)}(s)$ is concave, increasing in $\xi$ for all $k \geq 0$, $s \in \mathcal{S}$. 
Hence, $J_{\xi}(s)$ is concave and increasing in $\xi$.
\qed

\subsection{Proof of Proposition \ref{prop:continuity_of_cost}}
Clearly, $J_{\xi}(s)$ is continuous in $s$ over $(0,1)$, since otherwise concavity w.r.t. $s$ will be violated. 
Now, since $J_{\xi}(s)$ is concave in $s$ over $(0,1]$, we must have $J_{\xi}(1) \leq \lim_{s \uparrow 1}J_{\xi}(s)$. 
But since $J_{\xi}(s)$ 
is increasing in $s$, $J_{\xi}(1) \geq \lim_{s \uparrow 1}J_{\xi}(s)$. Hence, 
$J_{\xi}(1) = \lim_{s \uparrow 1}J_{\xi}(s)$. 
Thus, $J_{\xi}(s)$ is continuous in $s$ over $(0,1]$.

Again, for a fixed $s \in (0,1]$, $J_{\xi}(s)$ is concave and increasing in $\xi$. Hence, $J_{\xi}(s)$ is continuous in 
$\xi$ over $\xi \in (0,c),\, \forall \, c>0$. Hence, $J_{\xi}(s)$ is continuous in $\xi$ over $(0,\infty)$.
\qed

\end{document}